\documentclass[10pt,a4paper,reqno]{amsart}
\usepackage{amsthm,amsfonts,amssymb,amsmath,euscript,comment,enumerate,slashed,stackengine}
\usepackage{color}
\usepackage{bbm}
\usepackage{geometry}
\usepackage{amsfonts}
\usepackage{xcolor}
 \usepackage[hidelinks]{hyperref}
\usepackage{seqsplit}
\usepackage{mathtools}
\usepackage{empheq}

\newtheorem{thm}{Theorem}[section]
\newtheorem{prp}[thm]{Proposition}
\newtheorem{cor}[thm]{Corollary}
\newtheorem{lm}[thm]{Lemma}
\newtheorem{df}[thm]{Definition}
\newtheorem{definition}[thm]{Definition}
\newtheorem{rk}[thm]{Remark}
\newtheorem{remark}[thm]{Remark}
\newtheorem{proposition}[thm]{Proposition}
\newtheorem{lemma}[thm]{Lemma}
\newtheorem{theorem}[thm]{Theorem}

\newcommand{\m}[1]{\mathbb{#1}}

\newcommand{\q}[1]{\mathcal{#1}}

\newcommand{\wht}[1]{\widetilde{#1}}
\newcommand{\gra}[1]{\mathbf{#1}}

\newcommand{\ep}{\varepsilon}

\newcommand{\f}{\frac}
\newcommand{\rd}{\partial}
\newcommand{\nab}{\nabla}
\newcommand{\alp}{\alpha}
\newcommand{\bt}{\beta}
\newcommand{\bA}{{\bf A}}
\newcommand{\bB}{{\bf B}}
\newcommand{\bC}{{\bf C}}
\newcommand{\bD}{{\bf D}}
\newcommand{\gi}{(g^{-1})}
\newcommand{\mfg}{\mathfrak g}

\newcommand{\bfG}{{\bf \Gamma}}
\newcommand{\ls}{\lesssim}
\newcommand{\de}{\delta}

\newcommand{\om}{\omega}
\newcommand{\omg}{\omega}
\newcommand{\Omg}{\Omega}
\newcommand{\vp}{\varpi}
\newcommand{\srd}{\slashed{\rd}}

\def \i {\infty}

\newcommand{\ud}{\mathrm{d}}

\newcommand{\brk}{\@ifstar{\brkb}{\brki}}
\newcommand{\brki}[1]{\langle{#1}\rangle}
\newcommand{\brkb}[1]{\left\langle{#1}\right\rangle}

\numberwithin{equation}{section}

 \newcommand{\pfstep}[1]{\vspace{.5em} {\it \noindent #1.}}

\begin{document}
	
	\title[High-frequency backreaction for the Einstein equations: from dust to Vlasov ]
	{High-frequency backreaction for the Einstein equations under $\mathbb U(1)$ symmetry: from Einstein--dust to Einstein--Vlasov}

	\begin{abstract}
	Given suitable small, localized, $\mathbb U(1)$-symmetric solutions to the Einstein--massless Vlasov system in an elliptic gauge, we prove that they can be approximated by high-frequency vacuum spacetimes. This extends previous constructions where the limiting spacetime solves the Einstein--(multiple) null dust system (i.e., where the limiting massless Vlasov field can be written as a finite sum of delta measures). The proof proceeds by first approximating solutions to the Einstein--massless Vlasov system by solutions to the Einstein--(multiple) null dust system, then approximating solutions to the Einstein--null dust system by vacuum solutions. In the process, we take the number of families of dusts to infinity. 
	\end{abstract}
	
	\author{C\'ecile Huneau}
	\address{CNRS and DMA, Ecole Normale Supérieure PSL, 45 rue d'Ulm, 75005 Paris, France}
	\email{cecile.huneau@polytechnique.edu}
	\author{Jonathan Luk}
	\address{Department of Mathematics, Stanford University, CA 94304, USA}
	\email{jluk@stanford.edu}
	
	\maketitle

\section{Introduction}

It is now known that high-frequency limits of solutions to the Einstein vacuum equations need not be vacuum, but that an ``effective matter field'' may arise in the limit. The study of this phenomenon, including the construction of examples of approximate solutions, goes back to the works of Isaacson \cite{Isaacson1, Isaacson2} and Choquet-Bruhat \cite{CB.HF}. The first examples were given in plane symmetry, where there were explicit solutions \cite{Burnett, GW2, pHtF93, SGWK, SC.standing} and analytic constructions \cite{LeFLeF2020, Lott1, Lott3, Lott2}. More recently, there are constructions with weaker or no symmetry assumptions; see \cite{HL.HF, LR.HF, Touati2, touati2024reverse}. (See Section~\ref{sec:related.HF} for further discussions.) In all these examples, the limiting spacetime metric solves the Einstein--null dust system, and effective matter field corresponds to a finite number of families of null dusts propagating in different directions. 

However, Burnett \cite{Burnett} suggested that the class of possible limits is much larger than those solving the Einstein--null dust system. In fact, he conjectured that a spacetime can arise as a high-frequency limit of vacuum spacetimes if and only if it solves the Einstein--\emph{massless Vlasov} system. This includes as a special case solutions to the Einstein--null dust system, where the Vlasov field can be thought of as a linear combination of delta measures, but is more general. One direction of Burnett's conjectures was proven in recent works \cite{HL.Burnett, GuerraTeixeira, HL.wave} (see also \cite{LR.HF}), which show that under suitable gauge conditions, the limiting effective matter field could only be massless Vlasov matter. The question remains as to which solutions to the Einstein--massless Vlasov system can be achieved as limits of vacuum spacetimes. 

In this paper, \textbf{we provide the first construction of high-frequency limits of vacuum spacetimes, where the effective limiting matter field is not given by null dusts}. The main result of this paper is to approximate a general class of small and regular solutions to Einstein-massless Vlasov equations with a $\m U(1)$ symmetry on $\m R^{3+1}$ by solutions to Einstein vacuum equation, also with a $\m U(1)$ symmetry. 

Under a $\m U(1)$ symmetry, the Einstein vacuum equations in $3+1$ dimensions reduce to the Einstein equation in $2+1$ dimension coupled to a wave map field. Our goal is thus to construct a sequence of solutions $(g_i,\phi_i,\varpi_i)$ in the $(2+1)$-dimensional space $[0,1]\times \mathbb R^2$ solving the following Einstein--wave map system 
\begin{subequations}
\label{sys}
\begin{empheq}[left=\empheqlbrace]{align}
&R_{\mu \nu}(g)= 2\partial_\mu \phi \partial_\nu \phi + \frac{1}{2}e^{-4\phi}\partial_\mu \varpi \partial_\nu \varpi, \\
&\Box_g \phi + \frac{1}{2}e^{-4\phi}g^{-1}(\ud\varpi,\ud\varpi)  = 0,\\
&\Box_g \varpi -4g^{-1}(\ud\varpi,\ud\phi)=0,
\end{empheq}
\end{subequations}
which converge to a given solution $(g_0,\phi_0, \vp_0, f(\om),u(\om))$ that solves the\footnote{For simplicity we have chosen $u(\om)$ to be initially exactly a linear function for every $\om$. This can be slightly relaxed to requiring that the level sets of $u(\om)$ to be close to planes.} Einstein--wave map--massless Vlasov system:
\begin{subequations}\label{vlasov}
\begin{empheq}[left=\empheqlbrace]{align}
&R_{\mu \nu}(g)= 2\partial_\mu \phi \partial_\nu \phi +\frac{1}{2}e^{-4\phi}\partial_\mu \varpi \partial_\nu \varpi+ \int_{\m S^1}f^2(t,x,\omega)\partial_\mu u(t,x,\omega) \partial_\nu u(t,x,\omega)\,\ud m(\omega),\label{eq:vlasov.1}\\
&\Box_g \phi + \frac{1}{2}e^{-4\phi}g^{-1}(\ud\varpi,\ud\varpi)  = 0,\label{eq:vlasov.2}\\
&\Box_g \varpi -4g^{-1}(\ud\varpi,\ud\phi)=0,\label{eq:vlasov.3}\\
&2 L f + (\Box_{{g}} u) f = 0 \quad\forall \om \in \mathbb S^1,\label{eq:vlasov.4}\\
&g^{-1}(\ud u,\ud u)=0,\quad u \restriction_{\Sigma_0} = x^1 \cos \om + x^2 \sin \om,\, \rd_t u \restriction_{\Sigma_0} >0,\quad \forall \om \in \mathbb S^1.\label{eq:vlasov.5}
\end{empheq}
\end{subequations}
In \eqref{vlasov}, $\ud m(\omega)$ is a probability measure on $\m S^1$ and $L = (g^{-1})^{\alpha \beta}\partial_{\alpha} u \partial_{\beta}$. From now on, we take $f \geq 0$.
In the following, we will denote $U=(\phi,\varpi)$ and 
\begin{equation}\label{eq:lara.def}
\langle Y_1 U,Y_2 U \rangle = Y_1 \phi Y_2 \phi + \frac{1}{4}e^{-4\phi} Y_1 \varpi Y_2 \varpi \quad \hbox{for any smooth vector fields $Y_1$, $Y_2$}.
\end{equation} 

Here, in \eqref{vlasov}, we have chosen a specific (and somewhat non-standard) parametrization of the cotangent bundle. It will become clear later (see Section~\ref{sec:ideas}) how this is convenient for the proof. Notice that the system \eqref{vlasov} in $(2+1)$-dimensions arise as a reduction of the Einstein--massless Vlasov system\footnote{Note that while $m(\omega)$ in \eqref{vlasov} could be absolutely continuous with respect to the Lebesgue measure on $\mathbb S^1$, when viewed in $(3+1)$ dimensions, the Vlasov measure is not absolutely continuous with respect to the Lebesgue measure. However, according to \cite{HL.Burnett}, this is the expected class of Vlasov matter that arises in the limit of $\mathbb U(1)$ symmetric spacetimes.} in $(3+1)$ dimensions under $\mathbb U(1)$ symmetry.

The following is our main theorem. (We refer the reader to Theorem~\ref{main.thm.2} for a more precise version.)

\begin{thm}\label{main.intro}
	Consider initial data to \eqref{vlasov} on $\Sigma_0$ such that the following holds:
	\begin{itemize}
		\item The initial data for $(\phi_0, \varpi_0, f)$ are compactly supported, sufficiently small and regular.
		\item The initial hypersurface is maximal.
		\item $\phi_0 \restriction_{\Sigma_0}$ is not identically $0$.
	\end{itemize}
	Then the following holds in a suitable system of coordinates:
	\begin{enumerate}
	\item There exists a local-in-time solution to \eqref{vlasov} on $[0,1]\times \mathbb R^2$.
	\item There exists a sequence of solutions $\{(g_{i},\phi_{i},\varpi_{i})\}_{i=1}^\infty$  to \eqref{sys} on $[0,1]\times \mathbb R^2$ such that $(g_{i},\phi_{i},\varpi_{i}) \xrightarrow{i\to \infty} (g_{0},\phi_{0},\varpi_0)$ locally uniformly and the derivatives $(\rd g_{i},\rd \phi_{i}, \rd \varpi_i)\xrightharpoonup{i\to \infty} (\rd g_{0},\rd \phi_{0},\rd \varpi_0)$ weakly in $L^2_{\mathrm{loc}}$.
	\end{enumerate}
\end{thm}

Theorem~\ref{main.intro} constructs the first examples of where the limits go beyond solutions to the Einstein--null dust system. We note that the proof still considers in a crucial way an approximation of solutions to the Einstein--null dust system by vacuum solutions. The main challenge and novelty in the paper is to take the number of families of dust to infinity to create more general limiting Vlasov fields (see Section~\ref{sec:ideas}).


The remainder of the introduction is structured as follows. In \textbf{Section~\ref{sec:ideas}}, we will briefly explain the ideas of the proof and in \textbf{Section~\ref{sec:related}}, we will describe some related works. Finally, we outline the remainder of the paper in \textbf{Section~\ref{sec:outline}}.

\subsection{Ideas of the proof}\label{sec:ideas}

The starting point of our work is the following theorem that we established previously:
\begin{theorem}[H.--L. \cite{HL.HF}]\label{thm:old}
Fix $N \in \mathbb Z_{>0}$. Then Theorem~\ref{main.intro} is true in the special case where 
\begin{enumerate}
\item the limiting spacetime is in addition polarized, i.e., $\vp \equiv 0$, and
\item the probability measure is given by $m(\om) = \f 1 N \sum_{j=1}^N \de_{\om_j}$, where $\om_j \in \mathbb S^1$ and $\de_{\om_j}$ denotes the delta measure at $\om_j$. 
\end{enumerate}
\end{theorem}

The main contribution in the present paper is thus exactly to remove the conditions (1) and (2) in Theorem~\ref{thm:old}. Removing the $\varpi \equiv 0$ condition introduces new semilinear terms, which do not cause much difficulty because these terms verify the null condition. As we will explain in Section~\ref{sec:ideas.approx.by.vacuum}, the main modification is a new decomposition of the Vlasov field for the parametrix of $\phi$ and $\varpi$, which captures the interaction of the high-frequency parts of $\phi$ and $\varpi$.

On the other hand, removing condition (2), i.e., allowing for more general probability measures, is more difficult and is the main contribution of the paper. As we mentioned above, we will still reduce to the case of dusts. Thus our strategy is divided into two steps:
\begin{enumerate}
	\item First, we approximate the solution to Einstein-Vlasov equation by solutions to Einstein equations coupled to $N$ dust (Proposition~\ref{prop:construct.dust}).
	\item Then, we approximate the solution to Einstein-null dust equations by high frequency solutions to Einstein vacuum equations, with a frequency $\frac{1}{\lambda}$ very large compared to the number of dust $N$ (Theorem~\ref{thm:dust.by.vacuum}).
\end{enumerate}

%
%

The second step can be thought of as analogous to Theorem~\ref{thm:old} (other than having $\varpi \not \equiv 0$). However, in \cite{HL.HF}, the smallness parameter $\ep$ for the amplitude of $f$ is allowed to depend on $N$, but now the dependence on $N$ becomes important as we pass to the $N\to \infty$ limit.

\subsubsection{Approximation by multiple dusts}\label{sec:intro.approx.dust}
In the first step we approximate the probability measure $m(\om)$ by delta functions so that (for constants $\alp_\bA^{(N)} \geq 0$, $\sum_{\bA =0}^{N-1} \alp_\bA^{(N)} = 1$ and points $\omega_\bA^{(N)} \in \mathbb S^1$)
\begin{equation}\label{eq:intro.approx.m}
\sum_{\bA=0}^{N-1} \alp_\bA^{(N)} \delta_{\omega_{\bA}^{(N)}} \overset{\ast}{\rightharpoonup} m
\end{equation}
in the weak-* topology as $N\to \infty$.

For every fixed large $N \in \mathbb Z_{>0}$, we then consider an auxiliary system (recall \eqref{eq:lara.def})
\begin{equation}\label{eq:intro.back}
\left\{\begin{array}{l}
R_{\mu \nu}(g)= 2\langle \partial_\mu U, \partial_\nu U\rangle  + \sum_{{\bA}} \sqrt{\alp_\bA^{(N)}} f_{\bA}^2 \partial_\mu u_{{\bA}} \partial_\nu u_{{\bA}},\\
\Box_g \phi + \frac{1}{2}e^{-4\phi}g^{-1}(\ud\varpi,\ud\varpi)  = 0,\\
\Box_g \varpi -4g^{-1}(\ud\varpi,\ud\phi)=0,\\
2(g^{-1})^{\alpha \beta}\partial_{\alpha} u \partial_{\beta} f + (\Box_{{g}} u) f = 0, \\
(g^{-1})^{\alpha \beta}\partial_\alpha u \partial_\beta u=0,
\end{array}
\right.
\end{equation}
where $f_{\bA}(t,x) =  f(t,x,\om_{\bA}^{(N)})$, $u_{\bA}(t,x) = u(t,x,\om_{\bA}^{(N)})$. The system \eqref{eq:intro.back} can be solved in a high regularity norm, and one obtains by compactness and \eqref{eq:intro.approx.m} that as $N\to \infty$, the solutions to \eqref{eq:intro.back} converge to those of \eqref{vlasov}.

\subsubsection{Approximation by vacuum}\label{sec:ideas.approx.by.vacuum}
For the second step, which is the heart of the paper, we construct solutions to Einstein vacuum equations \eqref{sys} which approximate solutions to \eqref{eq:intro.back}. For a fixed $N$, we take the frequency parameter $\lambda^{-1}$ to be much larger and construct vacuum solutions of frequency $\lambda^{-1}$. Taking $N\to \infty$ and using Section~\ref{sec:intro.approx.dust}, we then obtain an approximation of the original solution to the Einstein--massless Vlasov system by vacuum solutions. 

We first rewrite \eqref{eq:intro.back} as
\begin{equation}\label{eq:intro.back.2}
\left\{\begin{array}{l}
R_{\mu \nu}(g)= 2\langle \partial_\mu U, \partial_\nu U\rangle  + \sum_{{\bA}}  F_{\bA}^2 \partial_\mu u_{{\bA}} \partial_\nu u_{{\bA}},\\
\Box_g \phi + \frac{1}{2}e^{-4\phi}g^{-1}(\ud\varpi,\ud\varpi)  = 0,\\
\Box_g \varpi -4g^{-1}(\ud\varpi,\ud\phi)=0,\\
2(g^{-1})^{\alpha \beta}\partial_{\alpha} u_\bA \partial_{\beta} F_\bA + (\Box_{{g}} u_\bA) F_\bA = 0, \\
(g^{-1})^{\alpha \beta}\partial_\alpha u_\bA \partial_\beta u_\bA =0,
\end{array}
\right.
\end{equation}
where we wrote $F_{\bA} = \sqrt{\alp_\bA^{(N)}} f_{\bA}$ (for notational convenience) and notices that it suffices to only consider the transport equations for $F_\bA$ and $u_\bA$ (instead of all $\om\in \mathbb S^1$).

In order to handle the terms in \eqref{sys} coming from $\varpi \not \equiv 0$ (which did not exist in \cite{HL.HF}), we further introduce a decomposition of $F_\bA$ and consider the following system: 
\begin{equation}\label{eq:intro.back.split}
\left\{\begin{array}{l}
R_{\mu \nu}(g)= 2\langle \partial_\mu U, \partial_\nu U\rangle  + \sum_{{\bA}} ((F^\phi_{{\bA}})^2+\frac{e^{-4\phi}}{4}(F^\varpi_\bA)^2)\partial_\mu u_{{\bA}} \partial_\nu u_{{\bA}},\\
\Box_g \phi + \frac{1}{2}e^{-4\phi}g^{-1}(\ud\varpi,\ud\varpi)  = 0,\\
\Box_g \varpi -4g^{-1}(\ud\varpi,\ud\phi)=0,\\
2(g^{-1})^{\alpha \beta}\partial_{\alpha} u_\bA \partial_{\beta} F^{\phi}_\bA + (\Box_{{g}} u_\bA) F^\phi_\bA + e^{-4\phi}(g^{-1})^{\alpha \beta}\partial_\alpha \varpi \partial_\beta u_\bA F^\varpi_\bA = 0, \\
2(g^{-1})^{\alpha \beta}\partial_{\alpha} u_\bA \partial_{\beta} F^{\varpi}_\bA + (\Box_{{g}} u_\bA) F^\varpi_\bA -4(g^{-1})^{\alpha \beta}\partial_\alpha \varpi \partial_\beta u_\bA F^\phi_\bA -4(g^{-1})^{\alpha \beta}\partial_\alpha \phi \partial_\beta u_\bA F^\varpi_\bA=0, \\
(g^{-1})^{\alpha \beta}\partial_\alpha u_{{\bA}} \partial_\beta u_{{\bA}}=0.
\end{array}
\right.
\end{equation}
Here, the key here is to split up the transport equation for $F_{\bA}$ into carefully chosen transport equation for $F_{\bA}^\phi$ and $F_{\bA}^\varpi$, where $F_{\bA}^2 = (F^\phi_{{\bA}})^2+\frac{e^{-4\phi}}{4}(F^\varpi_\bA)^2$. The choice is so that $F_{\bA}$ satisfies the original equation and also that when $F^\phi_{{\bA}}$ and $F^\varpi_\bA$ are used in the parametrices below, they capture the interaction of the high-frequency waves.

%
%
We will construct solutions to \eqref{sys} which admit the following ansatz for the wave map fields
\begin{equation}\label{eq:parametrix.main.term.only}
\phi = \phi_0 + \sum_\bA\lambda a_\bA^{-1} F^\phi_{\bA}\cos(\tfrac{a_\bA u_{\bA}}{\lambda})+ \mathcal E^\phi,\quad \vp = \vp_0 + \sum_\bA\lambda a_\bA^{-1} F^\vp_{\bA}\cos(\tfrac{a_\bA u_{\bA}}{\lambda})+ \mathcal E^\vp.
\end{equation}
The constants $a_{\bA}$ here are chosen to ensure that the interaction terms such that $\cos (\tfrac{a_\bA u_\bA \pm a_\bB u_\bB}{\lambda})$ are also of high frequency when $\bA \neq \bB$ (see Section~\ref{sec.prepare}).

As in \cite{HL.HF}, the parametrix \eqref{eq:parametrix.main.term.only} is not sufficiently precise. In particular, we construct the parametrix up to second order, capturing all the high-frequency terms of order $\lambda^2$. The parametrix takes the form 
\begin{equation}\label{eq:parametrix.intro.phi}
\begin{split}
\phi = &\: \lambda \sum_\bA  a_\bA^{-1}F^\phi_\bA \cos(\tfrac{a_\bA u_\bA}{\lambda})+  \lambda^2 \Big( \sum_\bA a_\bA^{-1} F^{1,\phi}_\bA \sin(\tfrac{a_\bA u_\bA}{\lambda})\\ 
&\: +\sum_\bA a_\bA^{-1} F_\bA^{2,\phi}\cos(\tfrac{2a_\bA u_\bA}{\lambda})+\sum_{\pm,\bA,\bB:\bA\neq \bB} F^\phi_{\bA,\bB} \cos (\tfrac{a_\bA u_\bA \pm a_\bB u_\bB}{\lambda }) \Big) + \wht \phi,
\end{split}
\end{equation}
and similarly for $\varpi$, and the error term $\wht \phi$ is better in $L^2$-based norm so that $\|\rd \wht \phi\|_{L^2}$ is of order $\lambda^2$ (with constants depending on $N$, see Section~\ref{sec:intro.N.dependence}).

\subsubsection{Almost orthogonality and local existence}\label{sec:intro.almost.orthogonality}
Up to this point, we have not discussed the dependence on $N$, which is in fact the main difficulty for passing to the $N \to \infty$ limit. The issue at stake is that we only know $\Big(\sum_\bA  |F^\phi_{\bA}|^2 + |F^\varpi_{\bA}|^2\Big)^{\f 12}\ls \ep$, while the $\ell^1$ sum $\sum_\bA \Big( |F^\phi_{\bA}|^2 + |F^\varpi_{\bA}| \Big)$ could be $N$-dependent and large. This already creates difficulty even for understanding the main terms $\sum_\bA\lambda a_\bA^{-1} F^\phi_{\bA}\cos(\tfrac{a_\bA u_{\bA}}{\lambda})$ and $\sum_\bA\lambda a_\bA^{-1} F^\vp_{\bA}\cos(\tfrac{a_\bA u_{\bA}}{\lambda})$  in \eqref{eq:parametrix.main.term.only}.

One easy observation is that due to the almost orthogonality of the high-frequency phases, we have $\|\partial \phi\|_{L^4},\,\|\partial \varpi\|_{L^4} \ls \ep$ (independently of $N$!) after choosing $\lambda^{-1}$ to be sufficiently large with respect to $N$.

This already comes into play for qualitative local existence (i.e., local existence even when the time of existence is allowed to depend on $\lambda$). Due to the globality (in space) of the elliptic estimates for the metric components, local existence in the elliptic gauge requires a \emph{smallness} assumption. A straightforward extension of the results in \cite{HL.elliptic} would require smallness for $\|\partial \phi\|_{L^\infty}$ and $\|\partial \varpi\|_{L^\infty}$. However, we use the more recent result of Touati \cite{Touati.local} (see Theorem~\ref{thm:Touati}) which only requires the less stringent smallness assumption on $\|\partial \phi\|_{L^4}$ and $\|\partial \varpi\|_{L^4}$.

\subsubsection{Parametrix for the metric, the eikonal function and the null expansion}

In addition to \eqref{eq:parametrix.intro.phi}, we need to control the metric components which couple to $\phi$ and $\varpi$ via the Einstein equation in \eqref{sys}. We choose an elliptic gauge so that the metric components satisfy elliptic equations. Due to the ellipticity, they are better behaved in the high-frequency regime. In fact, we could close a parametrix for each metric component $\mfg$ so that 
\begin{equation}\label{eq:intro.metric.expansion}
\mfg = \mfg_0 + \mfg_2 + \wht{\mfg},
\end{equation}
where $\mfg_0$ corresponds to the background, $\mfg_2$ consists of $O(\lambda^2)$ oscillatory terms like $\lambda^2 \sum_{\bA} \mathcal G^1_{1,A}(\mfg) \sin(\tfrac{a_\bA u_\bA}\lambda)$ and $\widetilde{\underline{\mfg}}$ is to be handled with nonlinear estimates. Notice that because $\mfg_2$ has $O(\lambda^2)$ terms, it is already better than the wave part of the system.

We should point out that in \cite{HL.HF}, we needed an extra $O(\lambda^3)$ term in the analogue of \eqref{eq:intro.metric.expansion}, which was then used to exploit a subtle cancellation in the wave equation. To implement that, however, creates various issues in the $N \to \infty$ limit. Instead, we show that the parametrix \eqref{eq:intro.metric.expansion} is already sufficient if we use the \emph{actual nonlinear} eikonal function $u_{\bA}$ which satisfies the nonlinear eikonal equation $g^{\alp\bt} \rd_\alp u_\bA \rd_\bt u_\bA = 0$ (instead of the background eikonal function $u_{\bA}^0$ which satisfies $g_0^{\alp\bt} \rd_\alp u^0_\bA \rd_\bt u^0_\bA = 0$).\footnote{We remark that using the parametrix introduced in this paper would also give a simplification of the proof in the setting of \cite{HL.HF}.} 

Because we now estimate the nonlinear $u_\bA$ as an independent unknown, we also need a parametrix for it:
\begin{equation}
u_\bA = u_\bA^0 + u_\bA^2 + \wht u_\bA,
\end{equation}
where the main oscillatory term $u_\bA^2$ is also $O(\lambda^2)$ (with $N$ dependence), essentially inheriting the size of $\mfg_2$. As before $\wht u_\bA$ is better and to be controlled with nonlinear estimates.

In order for this scheme to work, we separately control the null expansion $\chi_{\bA} = \Box_g u_{\bA}$. The quantity $\Box_g u_{\bA}$ naturally arises when we control the error from the wave map system. By considering $\chi_{\bA}$ separately, we show that while general second derivatives of $u_{\bA} - u_{\bA}^0$ is only $O(\lambda^0)$ (with dependence on $N$), the special combination $\chi_{\bA} - \chi_{\bA}^0$ is better and is $O(\lambda^1)$ (again with dependence on $N$). More precisely, 
\begin{equation}
\chi_\bA = \chi_\bA^0 + \chi_\bA^1 + \wht \chi_\bA,
\end{equation}
where $\chi_\bA^1$ are $O(\lambda)$ oscillatory terms (with $N$ dependence) and the nonlinear error $\wht \chi_\bA$ obeys better estimates than what one would get from the estimates for $\wht u_\bA$. To obtain this improvement, we use the Raychaudhuri equation that $\chi_\bA$ satisfies. In this equation, the main oscillatory terms are $O_N(\lambda)$, as opposed to $O_N(1)$ in the transport equations for the general second derivatives of $u_\bA$. In order to use the Raychaudhuri equation to obtain improved estimates for $\chi_{\bA}$, we would also need to carefully initiating the data for $u_\bA$ to have consistent bounds (see Section~\ref{sec:data.u}).

\subsubsection{$N$ dependence of the norms}\label{sec:intro.N.dependence}
In order to deal with the $N\to \infty$ difficulties, we introduce a norm which allows for exponential growth in $N$. Introduce a hierarchy of three\footnote{In the actual proof, it is convenient instead to introduce four constants so as to all for a different bootstrap constant for the time-differentiated quantity. We omit the technical considerations here in the introduction.} large constants depending on $N$: $C(N) \ll C_b(N) \ll A(N)$. Here, $C(N)$ is a universal function\footnote{Going through the proof gives that $C(N) = c^N N^N$ works for some large $c\geq 1$, though this is probably highly sub-optimal.} of $N$ and $C_b(N)$ can be thought of as a bootstrap constant. At the same time $A(N)$ is chosen to be much larger and in the bootstrap assumption, we allow the  error term quantities ($\wht \phi$, $\wht \varpi$, $\wht \mfg$, $\wht u_\bA$, $\wht \chi_\bA$) to grow like $e^{A(N)t}$. For instance, for $\wht \phi$ we make the bootstrap assumption that
\begin{equation}
\sum_{k\leq 3} \lambda^{k}\| \rd \wht \phi \|_{H^k}(t) \leq 2 \ep C_b(N) \lambda^2 e^{A(N)t}.
\end{equation}

The key point here is that when we carry out energy estimates for $\wht \phi$, there is a time integral, and so even if the nonlinear terms have extra large factors depending on $N$ (captured by $C(N)$ below), we will need to bound an integral 
$$\ep \lambda^2 \int_0^t C(N) C_b(N) e^{A(N) s} \, \ud s \leq \ep \lambda^2  \f{C_0(N) C_b(N)}{A(N)} e^{A(N) t},$$
which is much better than the bootstrap assumptions since $C_0(N) C_b(N) \ll A(N)$. (Notice also that since we have expanded the parametrix to sufficiently many terms, the main contribution would never come from quadratic terms in $\wht \phi$, $\wht \varpi$, $\wht \mfg$, $\wht u_{\bA}$, $\wht \chi_{\bA}$, as these terms would have extra smallness coming from $\lambda$. In particular, in terms on $N$ dependence for the borderline terms in $\lambda$, we never need to deal with e.g., $\int_0^t C(N) (C_b(N))^2 e^{2A(N) s} \, \ud s$, which would be worse than the $N$-dependence of the bootstrap assumptions.)

This almost handles all the issues with the $N$ dependence, except for the metric components, which are controlled using elliptic equations and thus the estimates do not exhibit an extra time integration. In order to deal with this, we use that $\Delta^{-1}: L^{\f 43} \to W^{1,4}$ (with weights that we suppress here), which is stronger than what we have for inverting the wave equation. This extra gain in the integrability in the Sobolev spaces allow us to use the smallness of $\|\partial \phi\|_{L^4}$ and $\|\partial \varpi\|_{L^4}$ (which comes from almost orthogonality, see Section~\ref{sec:intro.almost.orthogonality}) to obtain $N$-independent smallness and close the bootstrap argument.

\subsection{Related works}\label{sec:related}

We give a quick overview of some related results, but refer the reader to our recent survey \cite{HL.survey} for a more detailed discussion of other related works.

\subsubsection{Constructions of high-frequency spacetimes} \label{sec:related.HF}

High-frequency spacetimes are first constructed explicitly in symmetry classes \cite{Burnett, GW2, pHtF93, SGWK, SC.standing}. In plane symmetry, there is also a good well-posedness theory at low regularity consistent with high-frequency limits so that it could be used to construct examples of high-frequency limits \cite{LeFLeF2020, Lott1, Lott3, Lott2}. Using the local well-posedness theory in \cite{LR, LR2}, this was generalized in \cite{LR.HF} so that exact symmetries are not needed, but only high angular regularity is required. In all these examples, since the oscillations are limited to two null directions, the resulting limiting effective matter field can only have at most two families of null dusts.
 
Beyond the setting of two families of dust, examples were first constructed under $\mathbb U(1)$ symmetry \cite{HL.HF}. The work \cite{HL.HF} concerns the small data regime and allows for any finite number of families of null dusts (see Theorem~\ref{thm:old}).  

Finally, in more recent works, Touati constructed examples in generalized wave coordinates, where no symmetries are needed \cite{Touati2}. The use of generalized wave coordinates makes it possible to superpose different families of null dusts and obtain limiting effective matter field which corresponds to an arbitrary but finite number of families of null dusts \cite{touati2024reverse}.

As mentioned above, all these constructions only give null dust in the limit, where the number of families of dusts could be arbitrary but finite. The present paper gives the first construction of limiting Vlasov field beyond the setting of dusts. It is an interesting open problem whether this could also be carried out without symmetry, for instance by extending the work \cite{touati2024reverse}.

\subsubsection{Low-regularity results in $\mathbb U(1)$ symmetry} The construction high-frequency spacetimes can be viewed in the larger context of low-regularity solutions to the Einstein equations. 

In $(3+1)$ dimensions, the best known general low-regularity requires the initial data to be in $H^2$ \cite{L21}; see also \cite{bHjyC99b,bHjyC99a,sKiR2003,sKiR2005d,SmTa,dT2002}. There are also other special regimes of low-regularity results, for instance concerning impulsive gravitational waves \cite{Ang20, LR, LR2}.

Under $\mathbb U(1)$ symmetry, there are various works which obtain better low-regularity results than the case without symmetry. We refer the reader to \cite{sAntC2024, Huneau.mastersthesis, LVdM1, LVdM2, SmTa} for some results.

\subsection{Outline of the paper}\label{sec:outline} The remainder of the paper is structured as follows. In \textbf{Section~\ref{sec:setup}}, we introduce the geometric setup and the gauge that we use and recall the local existence results in this gauge. With the geometric background, we then state the precise version of our main theorem (Theorem~\ref{main.thm.2}) in \textbf{Section~\ref{sec:statement}}. In \textbf{Section~\ref{sec:dust}}, we approximate the solution to the Einstein--massless Vlasov system by solutions to the Einstein--null dust system with a large number of families of dust, and analyze the solutions to the Einstein--null dust system. In \textbf{Section~\ref{sec:main.theorem}}, we give the main result (Theorem~\ref{thm:dust.by.vacuum}) on approximating the solutions to the Einstein--null dust system by vacuum solutions with high frequency waves propagating in different directions.

The remainder of the paper is devoted to the proof of Theorem~\ref{thm:dust.by.vacuum} In \textbf{Section~\ref{sec:parametrix}}, we set up the parametrices for the wave and geometric quantities. In \textbf{Section~\ref{sec.id}}, we set up initial data which obey bounds consistent with the parametrices. In \textbf{Section~\ref{sec:bootstrap}}, we set up a bootstrap argument to control the error terms in the parametrices. Finally, the main a priori estimates are proved in \textbf{Section~\ref{sec:wave}} for the wave map fields $\phi$ and $\varpi$, in \textbf{Section~\ref{sec:eikonal}} for the eikonal functions $u_\bA$, in \textbf{Section~\ref{sec:metric}} for the metric, and in \textbf{Section~\ref{sec:chi}} for the null expansions $\chi_\bA$. Finally, we conclude the proof in \textbf{Section~\ref{sec:together}}.

\subsection*{Acknowledgements} J.~Luk is partially supported by a Terman fellowship and the NSF grant DMS-2304445. 

\section{Geometric setup, norms, and local well-posedness}\label{sec:setup}

\subsection{Elliptic gauge}\label{sec.elliptic.gauge}
Throughout the paper, we use an elliptic gauge that we now describe in this subsection. Under the symmetry reduction we work with, the wave dynamical quantities $\phi$, $\varpi$ can be separated from the geometric quantities. After introducing an elliptic gauge, we see that the geometric quantities satisfy elliptic equations, and thus obey stronger estimates as compared to $\phi$ and $\varpi$.

Given a $(2+1)$-dimensional metric $g$ on $\q M:=I\times \mathbb R^2$ in the form
\begin{equation}\label{g.form.0}
g=-n^2\ud t^2 + \bar{g}_{ij}(\ud x^i + \beta^i \ud t)(\ud x^j + \beta^j\ud t),
\end{equation}
we define $\Sigma_t:=\{(s,x^1,x^2): s=t\}$, $e_0:= \partial_t -\beta^i\rd_i$ (which is a future-directed normal to $\Sigma_t$), and
the second fundamental form of the embedding $\Sigma_t \subset \q M$
\begin{equation}\label{K}k_{ij}=-\frac{1}{2n}\q L_{e_0} \bar{g}_{ij}.
\end{equation}
We decompose $k$ into its trace and traceless parts.
\begin{equation}\label{K.tr.trfree} 
k_{ij}= H_{ij}+\frac{1}{2}\bar{g}_{ij}\tau.
\end{equation}
Here, $\tau:=\mbox{tr}_{\bar{g}} {k}$ and $H_{ij}$ is therefore traceless with respect to $\bar{g}$.

We introduce the elliptic gauge that we will use:
\begin{definition}\label{def:elliptic.gauge}
We say that a metric of the form \eqref{g.form.0} is in \textbf{elliptic gauge} if
\begin{enumerate}
	\item $\bar{g}$ is conformally flat, i.e., for some function $\gamma$,
	\begin{equation}\label{uniformized.g}
	\bar{g}_{ij}=e^{2\gamma}\delta_{ij};
	\end{equation}
	\item The constant $t$-hypersurfaces $\Sigma_t$ are maximal\footnote{The maximality of $\Sigma_t$ is technically not just a gauge condition, as in the context of an initial value problem, it requires the geometric condition that $\Sigma_0$ is maximal. Nonetheless, there is nothing fundamental about maximality per se; our argument works equally well when $\tau$ is a sufficiently regular and localized function.}
	$$\tau=0.$$
	We will later also fix the normalization of the function $t$ by imposing a condition on  $- n^{-2} = g^{-1}(\ud t, \ud t)$ as $|x|\to \infty$; see \eqref{eq:gamma.n.decomposition}.
\end{enumerate}
In particular, (1) implies that the metric takes the form 
\begin{equation}\label{g.form}
	g=-n^2 \ud t^2 + e^{2\gamma}\delta_{ij}(\ud x^i + \beta^i \ud t)(\ud x^j + \beta^j\ud t).
	\end{equation}
\end{definition}

	The following can be easily checked:
\begin{lemma}
	Given $g$ of the form \eqref{g.form}, the determinant of $g$ is given by
	\begin{equation}\label{g.det}
	\det(g)=-e^{4\gamma}n^2,
	\end{equation}
	and the inverse $g^{-1}$ is given by
	\begin{equation}\label{g.inverse}
	g^{-1}=\frac{1}{n^2}\left(\begin{array}{ccc}-1 & \beta^1 & \beta^2\\
	\beta^1 & n^2e^{-2\gamma}-\beta^1\beta^1 & -\beta^1 \beta^2\\
	\beta^2 & -\beta^1 \beta^2 & n^2e^{-2\gamma}-\beta^2\beta^2
	\end{array}
	\right).
	\end{equation}
\end{lemma}

In this gauge, $H$, $\gamma$, $n$ and $\beta$ satisfy elliptic equations. The following equations can be derived using the computations in \cite[Appendix~B]{HL.elliptic}.
\begin{lemma}
Define the following notation to describe the different cases simultaneously.
\begin{equation}
\mathcal R(Y_1,Y_2) = 
\begin{cases}
2\langle Y_1 U, Y_2 U\rangle & \hbox{in vacuum}, \\
2\langle Y_1 U, Y_2 U\rangle +  \sum_{\bA} F_{\bA}^2 (Y_1 u_\bA)(Y_2 u_\bA) & \hbox{for dust}, \\
2\langle Y_1 U, Y_2 U\rangle + \int_{\mathbb S^1} f^2(\om) (Y_1 u)(\om)(Y_2 u)(\om) \, \ud m(\om) & \hbox{for Vlasov}.
\end{cases}
\end{equation}
In each of these cases, the metric components in the elliptic gauge obey the following equations:
\begin{align}
{\de^{ik}}\partial_{{k}} H_{ij}=&\: -\frac{e^{2\gamma}}{n} \mathcal R(e_0, \partial_j), \label{elliptic.1}\\
\Delta \gamma = &\: -\f 12 \delta^{ij}\mathcal R(\rd_i, \partial_j) -\frac{e^{2\gamma}}{2 n^2}\mathcal R(e_0, e_0)- \frac{1}{2}e^{-2\gamma}|H|^2,\label{elliptic.2}\\
\Delta n =&\: ne^{-2\gamma}|H|^2+ \frac{e^{2\gamma}}{n}\mathcal R(e_0, e_0),\label{elliptic.3}\\
(\mathfrak L\beta)_{ij}=&\: 2Ne^{-2\gamma}H_{ij},\label{elliptic.4}
\end{align}
where $\Delta$ is the flat Laplacian $\Delta = \sum_{i=1}^2 \rd_{ii}^2$ and $\mathfrak L$ is the conformal Killing operator given by 
\begin{equation}\label{L.def}
(\mathfrak L\beta)_{ij}:=\delta_{j\ell}\rd_i\beta^\ell+\delta_{i\ell}\rd_j\beta^\ell-\delta_{ij}\rd_k\beta^k.
\end{equation}
\end{lemma}
The following simple computation will also be useful.
\begin{lemma}\label{lem:e0gamma}
$$2e_0\gamma = \rd_i \bt^i.$$
\end{lemma}

We also record the following useful computation. See for instance \cite[Lemma~2.6]{LVdM1}.
\begin{lemma}\label{lem:wave.operator}
For any function $h$, the wave operator is given as follows in coordinates:
$$\Box_g h  = -\f{e_0^2 h}{n^2} + e^{-2\gamma} \de^{ij}\rd^2_{ij} h + \f{e_0 n}{n^3} e_0 h + \f{e^{-2\gamma}}{n} \de^{ij} \rd_i n \rd_j h,$$
where $e_0 = \rd_t - \bt^i \rd_i$ as before.
\end{lemma}

\subsection{Norms}

\begin{definition}
Given a scalar function $u$, define
$$|\rd_x u|^2 = \sum_{i=1}^2 (\rd_i u)^2,\quad |\rd u|^2 = \sum_{\alp=0}^2 (\rd_\alp f)^2.$$
We make a similar definition for higher order tensor fields and for higher derivatives.
\end{definition}

We will use the following norms in this paper:
\begin{definition}
Let $V_t \subset \Sigma_t$ be an open connected subset with smooth boundary.
\begin{enumerate}
\item  For $k \in \mathbb Z_{\geq 0}$ and $p \in [1,\infty]$, define the $W^{k,p}(V_t)$ to be the standard Sobolev norms involving only spatial derivatives, i.e.,
$$ \| u \|_{W^{k,p}(V_t)} = \sum_{|\alp|\leq k} \| \rd_x^\alp u \|_{L^p(V_t)},$$
where $L^p(V_t)$ is taken with respect to the measure $\ud x = \ud x^1 \ud x^2$.
\item For $k \in \mathbb Z_{\geq 0}$, $p \in [1,\infty]$, $r \in \mathbb R$, define the weighted Sobolev norms by
$$\| u\|_{W^{k,p}_r(V_t)} =  \sum_{|\alp|\leq k} \| \langle x \rangle^{r+|\alp|} \rd_x^\alp u \|_{L^p(V_t)},$$
where $\langle x \rangle = (1+|x|^2)^{\f 12}$.
\item When $p = 2$, we denote $H^k(V_t) = W^{k,2}(V_t)$ and $H^k_r(V_t) = W^{k,2}_r(V_t)$.
\item When $k=0$, we denote the weighted $L^p$ spaces by $L^p_r(V_t) = W^{0,p}_r(V_t)$.
\end{enumerate}
\end{definition}

\subsection{Admissible free initial data}

Following \cite{Huneau.constraint, HL.elliptic}, we define a notion of admissible free initial data, which consist of (rescaled versions of) data for the matter field. At least in the small data regime, this is then sufficient to reconstruct the metric on the initial hypersurface. For simplicity, we restrict to the case where the data are smooth compactly supported, and where the level sets of the initial $u_{\bA}$ are given by coordinate planes.

\begin{df}[Admissible free initial data]\label{def.free.data}
	Define $\dot{U} = (\dot{\phi}, \dot{\varpi})$, $\nabla U = (\nabla \phi, \nabla \varpi)$, where $\dot{\phi}$, $\dot{\varpi}$ are defined as
	\begin{equation}\label{data.rescaled}
	\dot{\phi}=\frac{e^{2\gamma}}{n} (e_0 \phi),\quad \dot{\varpi}=\frac{e^{2\gamma}}{n} (e_0 \varpi),
	\end{equation}
	with $\gamma$ and $n$ as in \eqref{g.form}.
	
	
	For $R>0$, an \textbf{admissible free initial data set} with respect to the elliptic gauge consists of
	\begin{enumerate}
		\item $(\dot{U},\nabla U) \restriction_{\Sigma_0}\in C^\infty_c(B(0,R))$,
		\item $\begin{cases}
		\hbox{A finite set $\mathcal A$ in the null dust cases} \\
		\hbox{A probability measure $m$ on $\mathbb S^1$ in the Vlasov case},
		\end{cases}
		$
		\item $\begin{cases} 
		\breve{F}_{\bA} \restriction_{\Sigma_0} := F_\bA e^{\frac{\gamma}{2}}\restriction_{\Sigma_0}\in C^\infty_c(B(0,R)),\, \forall \bA\in \mathcal A \quad \hbox{in the null dust case}\\
		\breve{f}\restriction_{\Sigma_0} := f e^{\frac{\gamma}{2}}\restriction_{\Sigma_0} \in C^\infty_c(B(0,R)\times \mathbb S^1)\quad \hbox{in the Vlasov case},
		\end{cases}
		$
		\item $\begin{cases} 
		u_{\bA} \restriction_{\Sigma_0} = x^1 \cos \om_{\bA} + x^2 \sin \om_{\bA},\quad \om_{\bA} \in \mathbb S^1 ,\, \forall \bA\in \mathcal A\quad \hbox{in the null dust case}\\ 
		u(x,\om)\restriction_{\Sigma_0} = x^1 \cos \om + x^2 \sin \om\quad \hbox{in the Vlasov case},
		\end{cases}$
	\end{enumerate}
	where the following compatibility condition is satisfied:
	\begin{equation}\label{main.data.cond}
	\begin{cases}
	\int_{\mathbb R^2} 2\langle \dot{U},\rd_j U \rangle \, \ud x =0 \quad \hbox{in vacuum}, \\
	\int_{\mathbb R^2} \left(2\langle \dot{U},\rd_j U \rangle+ \sum_{\bA \in \mathcal A}\breve{F}_{\bA}^2|\nab u_{\bA} |\rd_j u_{\bA}\right) \, \ud x =0 \quad \hbox{in the null dust case}, \\
	\int_{\mathbb R^2} \left(2\langle \dot{U},\rd_j U \rangle+ \int_{\mathbb S^1} \breve{f}^2 |\nab u| \rd_j u \, \ud m(\om) \right) \, \ud x =0 \quad \hbox{in the Vlasov case}.
	\end{cases}
	\end{equation}
\end{df}

\subsection{Local existence and uniqueness results}\label{sec:easy.existence}

We now state the two main local existence and uniqueness results that we will use. These are results for regular initial data for the vacuum and/or the null dust case. In particular, when applied to high-frequency initial data, the time of existence is very small and is dependent on the frequency parameter. (Let us remark that we do not explicitly have or need a local existence and uniqueness result in the Einstein--massless Vlasov case, though in principle it could be proven along the lines of \cite{HL.elliptic, Touati.local}.)

To describe our local existence results, we first fix a cutoff function as follows.
\begin{definition}\label{def:cutoff}
Fix a cutoff function $\zeta:[0,\infty) \to [0,1]$ such that $\zeta(s) = 0$ for $s \leq 1$ and $\zeta(s) = 1$ for $s \geq 2$.
\end{definition}

The first result was established in our previous work \cite{HL.elliptic}, which applies for the Einstein--null dust system \eqref{eq:intro.back.2}. 

\begin{thm}[H.--L., Theorem~5.4 in \cite{HL.elliptic}]\label{lwp}
	Let $\alp_0 \in (0,\f 12)$, $k \geq 3$, $R>0$ and $\mathcal A$ be a finite set. Given an admissible free initial data set for the null dust case as in Definition \ref{def.free.data} such that
	\begin{equation}\label{smallness.fd}
	\|\dot{U}\|_{L^\infty(\Sigma_0)}+\|\nabla U\|_{L^\infty(\Sigma_0)} + \Big( \sum_{\bA \in \mathcal A} \|\breve{F}_{\bA}\|_{L^\infty(\Sigma_0)}^2 \Big)^{\f 12}\leq \ep,
	\end{equation}
	and 
	$$C_{high}:=\|\dot{U}\|_{H^k(\Sigma_0)}+\|\nabla U\|_{H^k(\Sigma_0)}+ \Big( \sum_{\bA \in \mathcal A} \|\breve{F}_{\bA}\|_{H^k(\Sigma_0)}^2 \Big)^{\f 12} <\infty.$$
	Then, there exists a constant $\ep_{low}=\ep_{low}(\alp_0,k,R)>0$ (independent of $C_{high}$ and $|\mathcal A|$) and a constant $T=T(C_{high}, \alp_0, k, R)>0$ (independent of $|\mathcal A|$) such that if $\ep<\ep_{low}$, there exists a unique solution to \eqref{eq:intro.back.2} in elliptic gauge on $[0,T]\times \m R^2 $.
	Moreover, {the following holds for some constant $C_h = C_h(C_{high},\alp_0, k,R)>0$:} 
	\begin{enumerate}
		\item The following estimates hold for $U$ and $F_{\bA}$ for all $t\in [0,T]$:
		\begin{subequations}
		\begin{empheq}{align}
		\|\nabla U\|_{H^k(\Sigma_t)} +\|\partial_t U \|_{H^k(\Sigma_t)} +\|\partial^2_t U\|_{H^{k-1}(\Sigma_t)}\leq &\:C_h,\label{eq:lwp.1} \\
		\sum_{\bA \in \mathcal A} \left(\|F_\bA \|_{H^k(\Sigma_t)}^2 + \|\partial_t F_{\bA}\|_{H^{k-1}(\Sigma_t)}^2+\|{\partial^2_t} F_{\bA}\|_{H^{k-2}(\Sigma_t)}^2\right)^{\f 12} \leq &\:C_h. \label{eq:lwp.2}
		\end{empheq}
		\end{subequations}
		\item The following estimates hold for the eikonal functions $u_{\bA}$ for all $\bA\in \mathcal A$ and for all $t\in [0,T]$:
		\begin{subequations}
		\begin{empheq}{align}
		\left(\min_{\bA} \inf_{x\in \mathbb R^2} |\nab u_{\bA}|(t,x)\right)^{-1}+\max_\bA\left(\|\nabla u_{\bA}-\overrightarrow{c_{\bA}} \|_{ H^{k}_{-\alp_0}(\Sigma_t)} + \Big\|\f{e^{\gamma}}n (e_0 u_{\bA}) - 1\Big\|_{H^{k}_{-\alp_0}(\Sigma_t)}\right)\leq &\: C_h, \label{eq:u.est.bkgd.1}\\
		\max_\bA \left(\|\partial_t\nabla u_{\bA} \|_{ H^{k-1}_{-\alp_0}(\Sigma_t)}+\Big\| \partial_t \Big(\f{{e^{\gamma}}}{n} e_0 u_\bA \Big) \Big\|_{ H^{k-1}_{-\alp_0}(\Sigma_t)}
		\right)\leq &\:C_h, \label{eq:u.est.bkgd.2}\\
		\max_\bA \left( \|\partial_t^2\nabla u_{\bA} \|_{ H^{k-2}_{-\alp_0}(\Sigma_t)}
		+ \Big\| \partial_t^2 \Big(\f{{e^{\gamma}}}{n} e_0 u_\bA \Big) \Big\|_{H^{k-2}_{-\alp_0}(\Sigma_t)}\right)\leq &\:C_h, \label{eq:u.est.bkgd.3}
		\end{empheq}
		\end{subequations}
		where $\overrightarrow{c_{\bA}} = (\cos \om_{\bA}, \sin \om_{\bA})$.
		\item The metric components $\gamma$ and $n$ can be decomposed as
		\begin{equation}\label{eq:gamma.n.decomposition}
		\gamma = \gamma_{a}\zeta(|x|)\log(|x|) + \gamma_r,\quad n = 1 +n_{a}(t)\zeta(|x|)\log(|x|) + n_r,
		\end{equation}
		with $\gamma_{a}\leq 0$ a constant, $n_{a}(t)\geq 0$ a function of $t$ alone and $\zeta(|x|)$ a cutoff function as in Definition~\ref{def:cutoff}.
		\item $\gamma$, $n$ and $\beta$ obey the following estimates for $t\in [0,T]$:
		\begin{subequations}
		\begin{empheq}{align}
		|\gamma_{a}| +\|\gamma_r\|_{ H^{k+2}_{-\alp_0}(\Sigma_t)} + \|\partial_t \gamma_r\|_{H^{k+1}_{-\alp_0}(\Sigma_t)} + \|\partial_t^2 \gamma_r \|_{H^k_{-\alp_0}(\Sigma_t)}\leq &\:C_h, \label{eq:g.est.bkgd.1} \\
		|n_{a}|(t) +|\partial_t n_{a}|(t)+|\partial^2_t n_{a}|(t)\leq &\:C_h,\\
		\|n_r\|_{H^{k+2}_{-\alp_0}(\Sigma_t)}+ \|\partial_t n_r\|_{H^{k+1}_{-\alp_0}(\Sigma_t)}+\|\partial^2_t n_r\|_{H^{k}_{-\alp_0}(\Sigma_t)}\leq &\:C_h,\\
		\|\beta \|_{H^{k+2}_{-\alp_0}} + \|\partial_t \beta\|_{H^{k+1}_{-\alp_0}(\Sigma_t)}+\|\partial^2_t \beta\|_{H^{k}_{-\alp_0}(\Sigma_t)} \leq &\:C_h.\label{eq:g.est.bkgd.4}
		\end{empheq}
		\end{subequations}
		\item The support of $U$ and $F_{\bA}$ satisfies 
		$$\mathrm{supp}(U,F_{\bA}) \subset J^+(\Sigma_0 \cap B(0, R)),$$
		where $J^{+}$ denotes the causal future.
	\end{enumerate}	
\end{thm}


Some remarks concerning the theorem are in order.
\begin{rk}
Strictly speaking, \cite[Theorem~5.4]{HL.elliptic} only deals with the polarized case, where $\varpi \equiv 0$. However, it is easy to check that the proof applies to the more general case in an identical manner.
\end{rk}

\begin{rk}
Note that as stated in \cite[Theorem~5.4]{HL.elliptic}, the constants are allowed to depend on $|\mathcal A|$, but in fact the proof shows that they only depend on $|\mathcal A|$ through   $\Big( \sum_{\bA} \|\breve{F}_{\bA}\|_{L^\infty}^2 \Big)^{\f 12}$ and $\Big( \sum_{\bA \in \mathcal A} \|\breve{F}_{\bA}\|_{H^k(\Sigma_0)}^2 \Big)^{\f 12}$; see \cite[Remark~5.5]{HL.elliptic}.
\end{rk}

\begin{rk}
We remark that in \cite{HL.elliptic}, there is an additional constant $C_{eik}$ capturing the deviation of the initial level sets of $u_{\bA}$ from planes. This is removed in the statement of Theorem~\ref{lwp} since we have prescribed $u_{\bA} \restriction_{\Sigma_0} = x^1 \cos \om_{\bA} + x^2 \sin \om_{\bA}$ in Definition~\ref{def.free.data}. 

In our case, we will later derive and use stronger estimates for $u_{\bA}(t,x)$; see Section~\ref{sec:angular.separation}.
\end{rk}

We will be using the following corollary of Theorem~\ref{lwp} for the null dust solutions, where $C_{high}$ is also small. In particular, in this case, it can be guaranteed that the solution at least exists up to $t = 1$.
\begin{cor}[H.--L., Corollary~5.7 in \cite{HL.elliptic}]\label{lwp.small}
	Suppose the assumptions of Theorem~\ref{lwp} hold. There exists $\ep_{small}=\ep_{small}(\alp_0,k,R)$ such that if $C_{high}$ and $\ep$ in Theorem \ref{lwp} both satisfy
	$$C_{high},\ep\leq \ep_{small},$$
	then the unique solution exists in $[0,1]\times \m R^2$. Moreover, there exists $C_0= C_0(\alp_0,k,R)$ such that all the estimates in Theorem~\ref{lwp} hold with $C_h$ replaced by $C_0\ep$.
	
\end{cor}

As we pointed out in Section~\ref{sec:intro.almost.orthogonality}, we need a strengthening of Theorem~\ref{lwp} in the vacuum case. This is because Theorem~\ref{lwp} requires an initial smallness condition on $L^\i(\Sigma_0)$ norms of $\dot{U}$ and $\nabla U$, and when constructing the highly oscillatory solutions, it is unclear that this need condition holds. Instead we will use the following result\footnote{Technically, \cite{Touati.local} only stated the version of the theorem with $k = 2$, but the propagation of higher norms are straightforward following the ideas in \cite{Touati.local}.} of Touati where only smallness in $L^4(\Sigma_0)$ is assumed.

\begin{thm}[Touati \cite{Touati.local}]\label{thm:Touati}
Let $\alp_0 \in (0, \f 12)$, $k \geq 2$ and $R>0$. 

Define $\dot{U} = (\dot{\phi},\dot{\varpi})$ as in \eqref{data.rescaled}.
\begin{equation}
\|\dot{U}\|_{L^4(\Sigma_0)}+\|\nabla U\|_{L^4(\Sigma_0)} \leq \ep,
\end{equation}
and
\begin{equation}
C_{high}:=\|\dot{U}\|_{H^k(\Sigma_0)}+\|\nabla U\|_{H^k(\Sigma_0)} <\infty.
\end{equation}
Then, there exists $\ep_{Touati} = \ep_{Touati}(\alp_0, k,R)>0$, $T = T(C_{high},\alp_0,k,R)>0$ and $C_h = C_h(C_{high},\alp_0, k,R) > 0$ such that if $\ep< \ep_{Touati}$, then the conclusions of Theorem~\ref{lwp} hold.
\end{thm}

To simplify the exposition, we fix the weight by defining
\begin{equation}\label{eq:alp0}
    \alp_0 = 10^{-10}.
\end{equation}

\section{Statement of main theorem}\label{sec:statement}

Having introduced the setup, we can now state a precise version of our main theorem: 
\begin{thm}[Main theorem, precise version]\label{main.thm.2}
	Let $R\geq 10$. Then there exists $\ep_0 = \ep_0(R)>0$ such that the following holds.
	
	Suppose $(\dot{U}(\cdot),\nabla U(\cdot), \breve f (\cdot,\omega), u(\cdot,\omega))\restriction_{\Sigma_0}$ is given and satisfies the following:
		\begin{enumerate}
	\item (Smoothness and support properties of $\dot{U}$, $\nabla U$ and $\breve{f}$) $(\dot{U},\nabla U) \restriction_{\Sigma_0}$ are smooth functions compactly supported in $B(0,R)$, and $\breve{f}(x,\omega) \restriction_{\Sigma_0}$ is smooth in $(x,\om)$ and the $x$-support is in $B(0,R)$.
	\item (Smallness condition) $$ \|\dot{U}\|_{H^{11}(\Sigma_0)}+\|\nabla U\|_{H^{11}(\Sigma_0)}+ \sup_{\om \in \mathbb S^1} \Big( \|\breve{f}\|_{H^{11}(\Sigma_0)}(\omg) +  \|\rd_\om \breve{f}\|_{H^{11}(\Sigma_0)}(\omg) \Big) \leq \ep.$$
	\item (Condition on the initial $u(x,\om)$) $$u\restriction_{\Sigma_0}(x,\om) = x^1\cos \om + x^2 \sin \om.$$
		\item (Orthogonality condition for the initial data)
		\begin{equation}\label{eq:ortho.main.thm}
		\int_{\m R^2} \left(-2\langle \dot{U},\partial_j U\rangle -\int_{\m S^1} \breve{f}^2|\nabla u|\partial_j u \,\ud m(\omega)\right) \ud x=0.
		\end{equation}
	\item (Genericity condition on initial data) $\phi$ is not identically zero.
	\end{enumerate}
	
	Then the following holds whenever $\ep \leq \ep_0$: 
	\begin{enumerate}
	\item There exists a (classical) solution $(g_0,\phi_0, f_0, u_0)$ to \eqref{vlasov} on $[0,1]\times \mathbb R^2$ in the gauge of Definition~\ref{def:elliptic.gauge} that arises from the given admissible free initial data set. 
	\item There exists a sequence of (classical) solutions $(g_{i},U_{i})$ to \eqref{sys}, which are all defined on $[0,1]\times \mathbb R^2$ in the gauge of Definition~\ref{def:elliptic.gauge}, such that for any $t\in [0,1]$, the following holds on $\Sigma_t$:
	$$(g_{i},U_{i}) \rightarrow (g_0,U_0)\quad\mbox{uniformly on compact sets}$$
	and 
	$$(\rd g_{i}, \rd U_{i})\rightharpoonup (\rd g_0,\rd U_0)\mbox{ { weakly in $L^2$}}$$ 
	with $\rd g_{i}, \rd U_{i}\in L^p_{loc}$ uniformly, for $2\leq p \leq 4$.
	\end{enumerate}
\end{thm}

The proof of Theorem~\ref{main.thm.2} will be given by a combination of Proposition~\ref{prop:construct.dust} (approximation of Vlasov by null dust) and Theorem~\ref{thm:dust.by.vacuum} (approximation of null dust by vacuum). We refer the reader to Section~\ref{sec:main.theorem}, where Theorem~\ref{main.thm.2} is proved assuming these results. 

\section{Einstein--null dust system}\label{sec:dust}

In this section, we construct solutions to the Einstein--null dust system with $N$ families. As $N\to \infty$, these solutions approximate the given solutions to the Einstein--massless Vlasov system. We then analyze the solutions to the Einstein--null dust system that we construct and prove properties that will be useful for the later parts of the paper.

\subsection{Approximation of Vlasov field by null dusts: general setup}

In this subsection, we give the setup for approximating solutions to the Einstein--massless Vlasov system by solutions to the Einstein--null dust system.


First, recall the well-known fact, which follows from the Krein--Milman theorem, that the set of convex combinations of Dirac measures is weak-* dense in the set of all probability measures. We will use a particular construction of a weak-* approximating sequence.
\begin{lm}\label{km}
	Let $m$ be a probability measure on $\mathbb S^1 := \mathbb R/(2\pi \mathbb Z)$. For all $N \in \mathbb Z_{>0}$, and $\bA = 0,1,\cdots,N-1$, we define $N$ separated points $\om_\bA^{(N)} = \f{2\pi \bA}{N}$ on $\m S^1$, and $N$ coefficients $\alp_\bA^{(N)} = m\Big([\f{2\pi \bA}{N}, \f{2\pi(\bA+1)}{N})\Big)$ (so that $\alp_\bA^{(N)} \geq 0$ and $\sum_{\bA=0}^{N-1} \alp_\bA^{(N)} =1$). It then holds that
	\begin{equation}\label{eq:m.limit}
	\sum_{\bA=0}^{N-1} \alp_\bA^{(N)} \delta_{\omega_{\bA}^N} \overset{\ast}{\rightharpoonup} m,
	\end{equation}
	in the weak-* topology as $N$ tends to infinity.
\end{lm}
\begin{proof}
	This is equivalent to checking that for any continuous function $h:\mathbb S^1\to \mathbb R$, $\sum_{\bA=0}^{N-1} \alp_\bA^{(N)} h(\omega_{\bA}^N) \to \int_{\mathbb S^1} h(\om) \, \ud m(\om)$, which is easy upon recalling that every continuous function on $\mathbb S^1$ is also uniformly continuous. \qedhere 
\end{proof}

From now on, we fix the parameters $\omega^{(N)}_{\bA}$ and $\alpha^{(N)}_{\bA}$ given by Lemma~\ref{km}.

We now approximate solutions to \eqref{vlasov} by $N$ dusts. We solve the following coupled system (see \eqref{eq:intro.back}):
\begin{subequations}\label{back.with.omega}
\begin{empheq}[left=\empheqlbrace]{align}
&R_{\mu \nu}(g^{(N)})= 2\langle \partial_\mu U^{(N)}, \partial_\nu U^{(N)}\rangle  + \sum_{{\bA}} \alpha_{\bA}^{(N)} (f^{(N)}_{\bA})^2 \partial_\mu u^{(N)}_{\bA} \partial_\nu u^{(N)}_{{\bA}}, \label{back.with.omega.1}\\
&\Box_{g^{(N)}} \phi^{(N)} + \frac{1}{2}e^{-4\phi^{(N)}}(g^{(N)})^{-1}(\ud\varpi^{(N)},\ud\varpi^{(N)})  = 0, \label{back.with.omega.2}\\
&\Box_{g^{(N)}} \varpi^{(N)} -4(g^{(N)})^{-1}(\ud\varpi^{(N)},\ud\phi^{(N)})=0, \label{back.with.omega.3}\\
&2 L^{(N)} f^{(N)} + (\Box_{g^{(N)}} u^{(N)}) f^{(N)}  = 0, \label{back.with.omega.4}\\
&(g^{(N)})^{-1} (\ud u^{(N)}, \ud u^{(N)})=0, \label{back.with.omega.6}
\end{empheq}
\end{subequations}
where $L^{(N)} = (g^{(N)})^{\alp\bt} \rd_\alp u^{(N)} \rd_\bt$, $f_{\bA}^{(N)} = f^{(N)}(\cdot,\om_{\bA}^{(N)})$ and $u_{\bA}^{(N)} = u^{(N)}(\cdot,\om_{\bA}^{(N)})$. Notice that $U^{(N)}= (\phi^{(N)},\varpi^{(N)})$, $g^{(N)}$, $f_{\bA}^{(N)}$, and $u_{\bA}^{(N)}$ are functions of $(t,x)$, while $f^{(N)}$ and $u^{(N)}$ are functions of $(t,x,\om)$.

Suppose we are given initial data to \eqref{vlasov} satisfying the assumptions of Theorem~\ref{main.thm.2}. We construct the approximate solutions by imposing the following initial conditions for the system \eqref{back.with.omega}:
\begin{equation}\label{eq:back.with.omega.data}
\begin{cases}
f^{(N)}\restriction_{\Sigma_0}(x,\om) = f\restriction_{\Sigma_0}(x,\omega), \\
u^{(N)}\restriction_{\Sigma_0}(x,\om) = x^1 \cos \om + x^2 \sin \om.
\end{cases}
\end{equation}
To specify the initial data for $U^{(N)} = (\phi^{(N)},\varpi^{(N)})$ and $\dot U^{(N)} = (\dot \phi^{(N)}, \dot \varpi^{(N)})$, we will need to deal with the constraint equations. This will be deferred to Section~\ref{sec:from.back.to.vlasov}. The data for $U^{(N)}$ will in particular be chosen to converge to the data for $U$ as $N\to \infty$.

%

After solving the constraint equations, our goal in this section will be to show that solutions to \eqref{back.with.omega} exist, and that as $N\to \infty$, solutions to \eqref{back.with.omega} with the data discussed above converge to those to \eqref{vlasov}.

\subsection{Rewriting the null dust system}

It will be convenient to rewrite \eqref{back.with.omega} in slightly different forms. We observe that even though in \eqref{back.with.omega}, $f^{(N)}$ and $u^{(N)}$ are defined as functions of $(t,x,\om)$, the equations for $g^{(N)}$ , $\phi^{(N)}$, $\varpi^{(N)}$ (i.e., \eqref{back.with.omega.1}--\eqref{back.with.omega.3}) depend only on $f^{(N)}$ and $u^{(N)}$ at $\om_\bA$. This motivates us to consider the following system (see \eqref{eq:intro.back.2}):
\begin{subequations}\label{preback}
\begin{empheq}[left=\empheqlbrace]{align}
&R_{\mu \nu}(g^{(N)})= 2\langle \partial_\mu U^{(N)}, \partial_\nu U^{(N)}\rangle  + \sum_{{\bA}} (F^{(N)}_{\bA})^2 \partial_\mu u^{(N)}_{\bA} \partial_\nu u^{(N)}_{{\bA}}, \label{preback.without.omega.1}\\
&\Box_{g^{(N)}} \phi^{(N)} + \frac{1}{2}e^{-4\phi^{(N)}}(g^{(N)})^{-1}(\ud\varpi^{(N)},\ud\varpi^{(N)})  = 0, \label{preback.without.omega.2}\\
&\Box_{g^{(N)}} \varpi^{(N)} -4(g^{(N)})^{-1}(\ud\varpi^{(N)},\ud\phi^{(N)})=0, \label{preback.without.omega.3}\\
&2 L^{(N)}_{\bA} F_{\bA}^{(N)} + (\Box_{g^{(N)}} u_{\bA}^{(N)}) F_{\bA}^{(N)} = 0 \quad \forall \bA, \label{preback.without.omega.4}\\
&(g^{(N)})^{-1} (\ud u_{\bA}^{(N)}, \ud u_{\bA}^{(N)})=0 \quad \forall \bA, \label{preback.without.omega.6}
\end{empheq}
\end{subequations}
where $L^{(N)}_\bA = (g^{(N)})^{\alp\bt} \rd_\alp u_\bA^{(N)} \rd_\bt$.

Notice in particular that after writing the equations in the form of \eqref{preback}, the local existence and uniqueness result in Theorem~\ref{lwp} immediately applies.

The following proposition clarifies the relation between \eqref{back.with.omega} and \eqref{preback}. The proof is largely straightforward and omitted. We just remark that in order to obtain 2b, we use the uniqueness of solutions for transport equations. 
\begin{proposition}\label{prop:dust.formulation.equivalent.1}
\begin{enumerate}
\item Suppose $(g^{(N)}, U^{(N)}= (\phi^{(N)},\varpi^{(N)}), f^{(N)},u^{(N)})$ is a smooth solution to \eqref{back.with.omega}. Then, setting 
\begin{equation}\label{eq:F.in.terms.of.f}
F_{\bA}^{(N)} = \sqrt{\alp_\bA^{(N)}} f^{(N)}(\cdot,\om_{\bA}^{(N)}),\quad u_{\bA}^{(N)} = u^{(N)}(\cdot,\om_{\bA}^{(N)}),
\end{equation}
it follows that $(U^{(N)}= (\phi^{(N)},\varpi^{(N)}), g^{(N)},\{ F_{\bA}^{(N)},u_{\bA}^{(N)}\}_{A=0}^{N-1})$ is a solution to \eqref{back.with.omega}.
\item Conversely, suppose $(g^{(N)}, U^{(N)}= (\phi^{(N)},\varpi^{(N)}),\{ F_{\bA}^{(N)},u_{\bA}^{(N)}\}_{A=0}^{N-1})$ is a smooth solution to \eqref{preback} satisfying
$$u_{\bA}^{(N)}\restriction_{\Sigma_0}(x) = x^1 \cos \om^{(N)}_{\bA} + x^2 \sin \om^{(N)}_{\bA}.$$ 
Solve for $(f^{(N)},u^{(N)})$ using the transport equations \eqref{back.with.omega.4}--\eqref{back.with.omega.6} (with the given $(U^{(N)}, g^{(N)})$) for all $\omg \in \mathbb S^1$ with initial data $f^{(N)}\restriction_{\Sigma_0}(x,\om)$ and $u^{(N)}\restriction_{\Sigma_0}(x,\om)$ satisfying
$$u^{(N)}\restriction_{\Sigma_0}(x,\om) = x^1 \cos \om + x^2 \sin \om,$$
and
$$f^{(N)}\restriction_{\Sigma_0}(x,\omega_{\bA}^{(N)}) = (\alp_\bA^{(N)})^{-\frac 12} F_{\bA}^{(N)}\restriction_{\Sigma_0}(x).$$ 
If a sufficiently regular solution exists, then the following holds: 
\begin{enumerate}
\item $(g^{(N)}, U^{(N)}= (\phi^{(N)},\varpi^{(N)}), f^{(N)},u^{(N)})$ solves \eqref{back.with.omega}.
\item $F_{\bA}^{(N)} = \sqrt{\alp_\bA^{(N)}} f^{(N)}(\cdot,\om_{\bA}^{(N)}),\quad u_{\bA}^{(N)} = u^{(N)}(\cdot,\om_{\bA}^{(N)}).$
\end{enumerate}
\end{enumerate}
\end{proposition}

We further rewrite \eqref{preback} into a different form by splitting $F_\bA^{(N)}$ into two parts. This choice is made in anticipation of capturing the interaction of the high-frequency $\phi$-waves and $\varpi$-waves for the vacuum approximation; see the discussions in Section~\ref{sec:ideas.approx.by.vacuum}. More precisely, consider the following system (see \eqref{eq:intro.back.split}):

\begin{subequations}\label{back}
\begin{empheq}[left=\empheqlbrace]{align}
&R_{\mu \nu}(g^{(N)})= 2\langle \partial_\mu U^{(N)}, \partial_\nu U^{(N)}\rangle  + \sum_{{\bA}} \Big( (F^{\phi,(N)}_{\bA})^2+\frac{e^{-4\phi^{(N)}}}{4}(F^{\varpi,(N)}_\bA)^2\Big)\partial_\mu u^{(N)}_{\bA} \partial_\nu u^{(N)}_{{\bA}}, \label{back.without.omega.1}\\
&\Box_{g^{(N)}} \phi^{(N)} + \frac{1}{2}e^{-4\phi^{(N)}}(g^{(N)})^{-1}(\ud\varpi^{(N)},\ud\varpi^{(N)})  = 0, \label{back.without.omega.2}\\
&\Box_{g^{(N)}} \varpi^{(N)} -4(g^{(N)})^{-1}(\ud\varpi^{(N)},\ud\phi^{(N)})=0, \label{back.without.omega.3}\\
&2 L^{(N)} F_{\bA}^{\phi,(N)} + (\Box_{g^{(N)}} u_{\bA}^{(N)}) F_{\bA}^{\phi,(N)} + e^{-4\phi} (g^{(N)})^{-1}( \ud \varpi^{(N)}, \ud u_{\bA}^{(N)}) F_{\bA}^{\varpi,(N)} = 0 \quad \forall \bA \label{back.without.omega.4}\\
\begin{split}
&2 L^{(N)} F_{\bA}^{\varpi,(N)} + (\Box_{g^{(N)}} u_{\bA}^{(N)}) F_{\bA}^{\varpi,(N)}  \\
&\qquad\qquad -4(g^{(N)})^{-1}(\ud \varpi^{(N)}, \ud u_{\bA}^{(N)}) F_{\bA}^{\phi,(N)}-4 (g^{(N)})^{-1}(\ud \phi^{(N)}, \ud u_{\bA}^{(N)}) F_{\bA}^{\varpi,(N)} = 0  \quad \forall \bA, \label{back.without.omega.5}
\end{split} \\
&(g^{(N)})^{-1} (\ud u_{\bA}^{(N)}, \ud u_{\bA}^{(N)})=0 \quad \forall \bA. \label{back.without.omega.6}
\end{empheq}
\end{subequations}

We will consider \eqref{back} with the initial condition
\begin{equation}\label{eq:F.split.initial}
    F_{\bA}^{\phi,(N)} \restriction_{\Sigma_0} = F_{\bA}^{(N)} \restriction_{\Sigma_0} = f \restriction_{\Sigma_0}(\cdot,\omg_\bA), \quad F_{\bA}^{\varpi,(N)} \restriction_{\Sigma_0} = 0.
\end{equation}


We will use the following relation between \eqref{preback} and \eqref{back}:
\begin{proposition}\label{prop:dust.formulation.equivalent.2}
Suppose $(g^{(N)}, U^{(N)} = (\phi^{(N)},\varpi^{(N)}), \{F_\bA^{(N)}, u_{\bA}^{(N)}\}_{A=0}^{N-1})$ solves \eqref{preback} on $[0,1]\times \mathbb R^2$.

Define $\{F_\bA^{\phi,(N)}, F_\bA^{\phi,(N)}\}_{\bA=0}^{N-1}$ by solving \eqref{back.without.omega.4} and \eqref{back.without.omega.5} (with the given $g^{(N)}$ and $\{u_{\bA}^{(N)}\}_{A=0}^{N-1}$) with the initial data given by \eqref{eq:F.split.initial}.

Then 
\begin{equation}\label{eq:reformulation.split}
    (F_\bA^{(N)})^2 = (F_\bA^{\phi,(N)})^2 + \f{e^{-4\phi^{(N)}}}4 (F_\bA^{\varpi,(N)})^2    
\end{equation}
and $(g^{(N)}, U^{(N)} = (\phi^{(N)},\varpi^{(N)}), \{F_\bA^{\phi,(N)},F_\bA^{\varpi,(N)}, u_{\bA}^{(N)}\}_{A=0}^{N-1})$ solves \eqref{back}.

\end{proposition}
\begin{proof}
    To simplify the notations, we omit the superscript $^{(N)}$. Let
    $$\overline{F}_\bA^2 = (F_\bA^\phi)^2 + \f{e^{-4\phi}}4 (F_\bA^\varpi)^2.$$
    Our goal is to show that $\overline{F}_\bA^2 = F_\bA^2$. Since they agree initially, it suffices to show that $\overline{F}_\bA$ obeys the same transport equation as $F_\bA$. For this, we compute
    \begin{equation}
\begin{split}
&\: 2 \overline{F}_\bA (g^{-1})^{\alpha \beta}\partial_{\alpha} u_\bA \partial_{\beta} \overline{F}_\bA + (\Box_g u_\bA) \overline{F}_\bA^2 \\
 = &\: 2 F_\bA^{\phi} (g^{-1})^{\alpha \beta}\partial_{\alpha} u_\bA \partial_{\beta} F_\bA^{\phi} + \frac{e^{-4\phi}}{4} F_\bA^{\varpi}  (g^{-1})^{\alpha \beta}\partial_{\alpha} u_\bA \partial_{\beta} F_\bA^{\varpi} - e^{-4\phi} F_\bA^{\varpi}(g^{-1})^{\alpha \beta}\partial_{\alpha} u_\bA \partial_{\beta} \phi F_\bA^{\varpi}  \\
&\: +  (\Box_g u_\bA) ((F_\bA^{\phi})^2 + \frac{e^{-4\phi}}{4} (F_\bA^\varpi)^2) = 0,
\end{split}
\end{equation}
which implies $2 (g^{-1})^{\alpha \beta}\partial_{\alpha} u_\bA \partial_{\beta} \overline{F}_\bA + (\Box_g u_\bA) \overline{F}_\bA = 0$ when $\overline{F}_\bA \neq 0$.\qedhere
\end{proof}

\subsection{The constraint equation for the initial data to the Einstein--null dust system}\label{sec:from.back.to.vlasov}

When solving the constraint equation, we will use the following result from \cite{LVdM1}.
\begin{lm}\label{lem:find.r1.r2}
Let $0 < R' < R''$. Suppose $\phi \not \equiv 0$ is a smooth function with $\mathrm{supp}(\phi) \subseteq B(0,R')$. Then there are smooth functions $r_1$, $r_2$ with $\mathrm{supp}(r_1),\, \mathrm{supp}(r_2) \subseteq B(0,R'')$ such that 
\begin{equation}\label{eq:def.r1.r2}
\det \Bigg[ \begin{array}{ll} \int_{\mathbb R^2} 2r_1 \rd_1 \phi  \,\ud x  & \int_{\mathbb R^2} 2r_2 \rd_1 \phi \, \ud x \\
 \int_{\mathbb R^2}2r_1 \rd_2 \phi \, \ud x &  \int_{\mathbb R^2} 2r_2 \rd_2 \phi \, \ud x \end{array}
\Bigg]   \neq 0.
\end{equation}
\end{lm}
\begin{proof}
This follows from \cite[Lemma~A.2, A.3]{LVdM1}. (The original result is stated only for $R'' = 2R'$, but the proof applies for more general $R' < R''$.) \qedhere
\end{proof}

We now solve the constraint equation for the initial data to \eqref{preback}.

\begin{lm}\label{lem:constraint.lower.bound}
	Suppose $(\dot{U}(\cdot),\nabla U(\cdot), \breve f (\cdot,\omega), u(\cdot,\omega))\restriction_{\Sigma_0}$ is given and satisfy the assumptions of Theorem~\ref{main.thm.2}. Define $r_1$, $r_2$ so that $\mathrm{supp}(r_1),\, \mathrm{supp}(r_2) \subseteq B(0,\f{3R}2)$ and \eqref{eq:def.r1.r2} holds, with the additional normalization conditions 
	\begin{equation}
	\|r_1 \|_{H^{10}(\Sigma_0)} = \|r_2 \|_{H^{10}(\Sigma_0)} = \ep
	\end{equation} 
	(which can be arranged after rescaling).

	Then for each $N \in \mathbb Z_{\geq 1}$, there exist parameters $\Omega^N_1, \Omega^N_2$ such that the following holds:
	\begin{enumerate}
	\item $\lim_{N\to \infty} (\Omega^N_1, \Omega^N_2) = (0,0)$.
	\item Consider the free initial data to \eqref{preback} given by 
	\begin{subequations}\label{eq:data.for.N.back}
        \begin{empheq}[left=\empheqlbrace]{align}
            &\phi^{(N)}\restriction_{\Sigma_0} = \phi\restriction_{\Sigma_0}, \quad  \varpi^{(N)} \restriction_{\Sigma_0}= \varpi\restriction_{\Sigma_0}, \quad \dot{\varpi}^{(N)} \restriction_{\Sigma_0} = \dot{\varpi}\restriction_{\Sigma_0}, \label{eq:data.for.N.back.1}\\
            &\dot{\phi}^{(N)}\restriction_{\Sigma_0} = \dot{\phi}\restriction_{\Sigma_0} + \Omega^N_1 r_1 + \Omg^N_2 r_2 , \label{eq:data.for.N.back.2}\\
            & \breve{F}_{\bA}^{(N)}\restriction_{\Sigma_0}(x) = \sqrt{ \alp_\bA^{(N)}}  \breve{f}\restriction_{\Sigma_0} (x,\om_{\bA}), \label{eq:data.for.N.back.3}\\
            & u_{\bA}^{(N)}\restriction_{\Sigma_0}(x) = u^{(N)} \restriction_{\Sigma_0} (x,\om_{\bA}) = x^1 \cos \om_\bA + x^2 \sin \om_\bA,\label{eq:data.for.N.back.5}
        \end{empheq}
    \end{subequations}    
	where $\omega_\bA \in \m S^1$ and $\alp_\bA^{(N)}$ given by Lemma \ref{km}. Then the free initial data set satisfies the compatibility condition \eqref{main.data.cond} for every $N \in \mathbb Z_{\geq 1}$:
	$$\int_{\mathbb R^2} \left(2\langle \dot{U}^{(N)},\rd_j U^{(N)} \rangle+ \sum_{\bA \in \mathcal A}(\breve{F}_{\bA}^{(N)})^2|\nab u_{\bA}^{(N)} |\rd_j u_{\bA}^{(N)}\right) \, \ud x =0 .$$
	\end{enumerate}
\end{lm}

\begin{proof}
	To keep the notation lean, we suppress the superscript $^{(N)}$. We write the condition \eqref{main.data.cond} for the given free initial data, i.e.,  
	$$	\int_{\mathbb R^2} \Big(2(\dot{\phi} + \Omega^N_1 r_1 + \Omega^N_2 r_2)\rd_j \phi + \frac{1}{2}e^{-4\phi}\dot{\varpi}\partial_j \varpi+ \sum_{\bA}\breve{F}_{{\bA}}^2|\nab u_{{\bA}}|\rd_j u_{{\bA}} \Big) \, \ud x =0,$$
as a system of linear equations for $\Omega^N_1,\Omega^N_2$, which in matrix form is given by
$$-\Bigg[ \begin{array}{ll} \int_{\mathbb R^2} 2r_1 \rd_1 \phi  \,\ud x  & \int_{\mathbb R^2} 2r_2 \rd_1 \phi \, \ud x \\
 \int_{\mathbb R^2}2r_1 \rd_2 \phi \, \ud x &  \int_{\mathbb R^2} 2r_2 \rd_2 \phi \, \ud x \end{array}
\Bigg] \Bigg[ \begin{array}{l} \Omega^N_1 \\ \Omega^N_2\end{array} \Bigg] = \Bigg[ \begin{array}{l} \int_{\mathbb R^2} \Big( 2 \dot{\phi}\rd_1 \phi+ \frac{1}{2}e^{-4\phi}\dot{\varpi}\partial_2 \varpi + {\sum_{\bA}}\breve{F}_{{\bA}}^2|\nab u_{{\bA}}|\rd_1 u_{{\bA}} \Big)  \,\ud x \\  \int_{\mathbb R^2} \Big( 2 \dot{\phi}\rd_2\phi+ \frac{1}{2}e^{-4\phi}\dot{\varpi}\partial_2 \varpi+ {\sum_{\bA}} \breve{F}_{{\bA}}^2|\nab u_{{\bA}}|\rd_2 u_{{\bA}} \Big)  \,\ud x \end{array} \Bigg].$$
 
Our choice of $r_1$, $r_2$ satisfies \eqref{eq:def.r1.r2} so that the matrix in the above equation is invertible, giving the desired $\Omg_1^N$, $\Omg_2^N$.
We then note that (by Lemma~\ref{km})
$$\int_{\mathbb R^2} \left(\int_{\m S^1} \breve{f}^2(x,\omega)|\nab_x u (x,\omega)|\rd_j u(x,\omega)\ud m(\omega)- {\sum_{\bA}} \breve{F}_{{\bA}}^2|\nab u_{{\bA}}|\rd_j u_{{\bA}} \right)  \ud x\rightarrow 0, \quad N \rightarrow \infty.$$
Consequently $\Omega^N_1$ and $\Omega^N_2$ tend to zero as $N$ tends to infinity.
\end{proof}

From now on, for any given $N \in \mathbb Z_{\geq 1}$, we consider \eqref{preback} with initial data given by Lemma~\ref{lem:constraint.lower.bound}. 
\begin{itemize}
    \item Alternatively, this can be also be viewed as initial data for \eqref{back.with.omega}; in this case \eqref{eq:data.for.N.back.1}, \eqref{eq:data.for.N.back.2} remain the same, but  \eqref{eq:data.for.N.back.3}--\eqref{eq:data.for.N.back.5} are replaced by \eqref{eq:back.with.omega.data}.
    \item This can be also be viewed as initial data for \eqref{back}; in this case \eqref{eq:data.for.N.back.1}, \eqref{eq:data.for.N.back.2}, \eqref{eq:data.for.N.back.5} remain the same, but \eqref{eq:data.for.N.back.3} is replaced by \eqref{eq:F.split.initial}.
\end{itemize}

\subsection{Construction of the approximations by solutions to the Einstein--null dust system}

In this subsection, we construct solutions to \eqref{back.with.omega} with data prescribed above tend to solutions to \eqref{vlasov} as $N\to \infty$.

It is convenient here as for the remainder of the paper to introduce the sets $K$ and $K_t$. These are spatially compact past set. Since we will only need estimates for the eikonal function on the support of the dust and the scalar field, we will only consider them on these sets.
\begin{definition}\label{def:K}
Define 
$$K= \{(t,x): 0 \leq t \leq 1,\, |x|< 2R+2 +2(1-t) \}.$$ 
In all the settings below, since since the metric is close to the Minkowski metric (see e.g., \eqref{lwp.small}), the boundary $\{(t,x): |x| = 2R+2+2(1-t)\}$ is spacelike and thus $\overline{K}$ is a compact past set. We denote also
\begin{equation}\label{eq:Kt.def}
K_t =B(0,2R+2+2(1-t))
\end{equation}
so that $K = \cup_{t \in [0,1]} K_t$.
\end{definition}

The following proposition is the main result of this subsection. Note that the proposition in particular implies the existence of solutions to \eqref{vlasov}. (Notice that existence for \eqref{vlasov} is not included in the local existence results in Section~\ref{sec:easy.existence}, as they are either for vacuum or for the null dust case, even though one expects that an analogous existence result can be proven by similar techniques as \cite{HL.elliptic}.)
\begin{proposition}\label{prop:construct.dust}
Fix $R\geq 10$. Given initial data $(\dot{U}, \nabla U,\breve f (x,\omega), u(x,\omega))$ satisfying the hypotheses (1)--(5) of Theorem~\ref{main.thm.2}, there exists $\ep_1 = \ep_1(R)>0$ such that the following holds whenever $\ep<\ep_1$: 
\begin{enumerate}
\item A solution $(g,U, f, u)$ to \eqref{vlasov} arising from the given admissible free initial data set exists on a time interval $[0,1]$. 
\item There exists a sequence of solutions $\{(g^{(N)},U^{(N)},f^{(N)},u^{(N)})\}_{N=1}^\infty$ to \eqref{back.with.omega} so that 
$$(g^{(N)},U^{(N)},u^{(N)}) \xrightarrow{N\to \infty} (g,U, u) \quad \hbox{ in $C^2_{loc,t,x}$}$$
and 
$$f^{(N)} \xrightarrow{N\to \infty} f \quad \hbox{in $C^2_{t,x} C^0_\omega$}.$$
\item There exists a sequence of solutions $\{(g^{(N)},U^{(N)},\{ F_\bA^{\phi,(N)},F_\bA^{\varpi,(N)}, u_\bA^{(N)}\}_{\bA=0}^{N-1})\}_{N=1}^\infty$ to \eqref{back} such that for any $N\in \mathbb Z_{\geq 1}$, $(g^{(N)},U^{(N)})$ agree with those in part (2) above, and 
\begin{equation}\label{eq:construct.dust.agree}
    (F^{\phi,(N)}_{\bA})^2+\frac{e^{-4\phi^{(N)}}}{4}(F^{\varpi,(N)}_\bA)^2 = \alpha_{\bA}^{(N)} f^2(\cdot,\om_\bA),\quad u^{(N)}_\bA = u^{(N)}(\cdot,\om_\bA).
\end{equation}
\item For $N \in \mathbb Z_{\geq 1}$, $(g^{(N)},U^{(N)},\{F^{\phi,(N)}_\bA, F^{\varpi,(N)}_\bA,u^{(N)}_\bA\}_{\bA = 1}^{N-1})$ in part (3) above satisfies the following properties:
\begin{enumerate}
    \item The estimate \eqref{eq:lwp.1} in (1) and points (2)--(5) in the conclusion of Proposition~\ref{lwp} hold with $k=11$ and with $C_h$ replaced by $C \ep$, where $C = C(R)>0$ is independent of $N$.
    \item The estimate \eqref{eq:lwp.2} in (1) of Proposition~\ref{lwp} holds for both $F_\bA^{\phi,(N)}$ and $F_\bA^{\varpi,(N)}$, i.e., for all $t \in [0,1]$,
    \begin{align}
        \sum_{\bA =0}^{N-1} \sum_{j=0}^1 \left(\|\rd_t^j F_\bA^{\phi,(N)} \|_{H^{11-j}(\Sigma_t)}^2  + \| \rd_t^j F_\bA^{\varpi,(N)} \|_{H^{11-j}(\Sigma_t)}^2 \right)^{\f 12} \leq C \ep.\label{eq:lwp.2.refined}
    \end{align}
    where $C = C(R)>0$.
    \item The following pointwise estimates hold for all $t \in [0,1]$:
    \begin{align}
        \sum_{\bA =0}^{N-1} \sum_{j=0}^1 \left( \| \rd_t^j F_\bA^{\phi,(N)}(\cdot,\omg) \|_{C^{9-j}(\Sigma_t)}^2 + \| \rd_t^j F_\bA^{\varpi,(N)}(\cdot,\omg) \|_{C^{9-j}(\Sigma_t)}^2 \right)^{\f 12} \leq C \ep, \\ 
         \sum_{j=0}^1 \Big(\| \rd_t^j \rd U \|_{C^{9-j}(\Sigma_t)} + |\rd_t^j \mfg_a| + \| \rd_t^j \mfg_r \|_{C^{11-j}(K_t)} + \|\rd_t^j \rd u_\bA \|_{C^{9-j}(K_t)} \Big) \leq C \ep
    \end{align}
    where we have written
    $$\mfg(t,x)=\mfg_{a}(t)\zeta(|x|)\log({|x|})+ \mfg_r(t,x).$$
    \item $\mathrm{supp}(U^{(N)}),\,\mathrm{supp}(F^{\phi,(N)}_\bA),\,\mathrm{supp}(F^{\varpi,(N)}_\bA) \subset [0,1]\times B(0, \f {3R}2+2)$.
\end{enumerate} 
\end{enumerate}
\end{proposition}

\begin{proof}
\pfstep{Step~1: Existence of solutions \eqref{back.with.omega} and existence of a limit} For every $N \in \mathbb Z_{\geq 1}$, solve \eqref{preback} with the initial data given by Lemma~\ref{lem:constraint.lower.bound}. By Corollary~\ref{lwp.small}, for every $N \in \mathbb Z_{>0}$, there exists a solution $(g^{(N)},U^{(N)},\{F_\bA^{(N)},u_\bA^{(N)}\}_{\bA=0}^{N-1})$ to \eqref{preback}.

In order to get a solution to \eqref{back.with.omega}, we use part 2 of Proposition~\ref{prop:dust.formulation.equivalent.1}, i.e., we fix $(g^{(N)},U^{(N)})$ as above, and solve for $(f^{(N)},u^{(N)})$ by the equations \eqref{back.with.omega.4}--\eqref{back.with.omega.6} with initial data \eqref{eq:back.with.omega.data}. Notice that since \eqref{eq:g.est.bkgd.1}--\eqref{eq:g.est.bkgd.4} hold with $C_h$ being a small constant (by Corollary~\ref{lwp.small}), the eikonal equation \eqref{back.with.omega.6} can be solved up to time $1$, where $u(x,\om)$ remains close to $t+ x^1\cos \om + x^2 \sin \om$ on the set $K$ (see Definition~\ref{def:K}). After that, the equation \eqref{back.with.omega.4} is linear and can be solved. Moreover, by Proposition~\ref{prop:dust.formulation.equivalent.1}, 
\begin{equation}\label{eq:construct.dust.agree.1}
    F_{\bA}^{(N)} = \sqrt{\alp_\bA^{(N)}} f^{(N)}(\cdot,\om_{\bA}^{(N)}),\quad u_{\bA}^{(N)} = u^{(N)}(\cdot,\om_{\bA}^{(N)}).
\end{equation}

Notice that the solutions we constructed are uniformly bounded in high regularity norms. Moreover, we can easily propagate a $C^1$ bound in $\omega$ after differentiating the transport equations \eqref{back.with.omega.4}--\eqref{back.with.omega.6} by $\rd_\om$ (since the $\rd_\om f$ and $\rd_\om u$ terms are linear). As a result, by compactness, we obtain the $C^2_{t,x} C^0_\omega$ convergence of a subsequence to a quintuple $(g^{(\infty)},U^{(\infty)},f^{(\infty)},u^{(\infty)})$.

\pfstep{Step~2: The limit solves \eqref{vlasov}} It is clear that the limit $(g^{(\infty)},U^{(\infty)},f^{(\infty)},u^{(\infty)})$ achieves the given free initial data, which in term implies that it achieves the given initial data. It remains to check that it satisfies \eqref{vlasov}. (In the statement of the theorem, we have removed the superscript $^{(\infty)}$.) 

For this purpose, first observe that the convergence is strong enough to easily justify that \eqref{eq:vlasov.2}--\eqref{eq:vlasov.5} hold for the limit. By uniform-in-$\omega$ convergence and \eqref{eq:m.limit}, it holds that
\begin{equation}
\begin{split}
&\: \lim_{N\to \infty} \sum_{\bA=0}^{N-1} \alpha_\bA^{(N)} f^{(N)}(t,x,\omega_\bA^{(N)})^2\partial_\mu u^{(N)}(t,x,\omega_\bA^{(N)})\partial_\nu u^{(N)}(t,x,\omega_\bA^N) \\ 
= &\: \int_{\m S^1}f^{(\infty)}(t,x,\omega)^2\partial_\mu u^{(\infty)}(t,x,\omega)\partial_\nu u^{(\infty)}(t,x,\omega)\,\ud m(\omega).
\end{split}
\end{equation}
Hence, \eqref{eq:vlasov.1} also holds for the limit.


\pfstep{Step~3: Construction of solutions to \eqref{back}} We start with the solution $(g^{(N)},U^{(N)},\{F_\bA^{(N)},u_\bA^{(N)}\}_{\bA=0}^{N-1})$ to \eqref{preback} already constructed in Step~1, and solve for $\{F_\bA^{\phi,(N)}, F_\bA^{\phi,(N)}\}_{\bA=0}^{N-1}$ using the transport equations according to Proposition~\ref{prop:dust.formulation.equivalent.2}. Since $(g^{(N)}, U^{(N)})$ are now fixed, the transport equations \eqref{back.without.omega.4}--\eqref{back.without.omega.5} are linear and can be solved. Finally, that \eqref{eq:construct.dust.agree} holds is a consequence of \eqref{eq:construct.dust.agree.1} and \eqref{eq:reformulation.split}. This proves point (3) of the proposition.

\pfstep{Step~4: Properties of the solutions} We now turn to point (4) of the proposition. Part (a) is immediate from the application of Corollary~\ref{lwp.small} above. For part (b), notice that Corollary~\ref{lwp.small} only gives a bound for $F_\bA^{(N)} = (F^{\phi,(N)}_{\bA})^2+\frac{e^{-4\phi^{(N)}}}{4}(F^{\varpi,(N)}_\bA)^2$. However, since \eqref{back.without.omega.4}--\eqref{back.without.omega.5} are linear equations given $(g^{(N)},U^{(N)})$, it is straightforward to propagate the bounds to obtain \eqref{eq:lwp.2.refined}.  Part (c) follows from parts (a) and (b) after using Sobolev embedding. Finally, by Corollary~\ref{lwp.small} and assumption of the initial support, we know that $\mathrm{supp}(U,F_{\bA}) \subset J^+(\Sigma_0 \cap B(0, \f{3R}2))$. Using the smallness estimates of the metric component, it is easy to see that $J^+(\Sigma_0 \cap B(0, \f{3R}2)) \subset [0,1]\times B(0, \f{3R}2+2)$, proving part (d). \qedhere

\end{proof}

\subsection{Angular separation of the eikonal functions}\label{sec:angular.separation}

We now show that for solutions to \eqref{back.with.omega} constructed in Proposition~\ref{prop:construct.dust}, we can control the angle between the level sets of $u(\om)$ and $u(\om')$. In particular, this gives an estimate for the angular separation of $u_{\bA}$.

For the rest of the subsection, fix $N$ and take a solution from Proposition~\ref{prop:construct.dust}. Our goal will be to obtain estimates which are independent of $N$. In the following, we will suppress the explicit dependence $^{(N)}$ in our notations. All implicit constants are independent of $N$.

In order to distinguish that derivatives in the $\om$ variable from the derivatives in the spacetime variables, we will use the notation $\srd_\om$ for the $\om$-derivative in the rest of this subsection.


\begin{lemma}\label{lem:srdu}
$$\sup_t \|\srd_\om u \|_{L^\i(K_t)} \ls 1.$$
\end{lemma}
\begin{proof}
We begin with 
$$\begin{cases}
\gi^{\alp\bt} \rd_\alp u \rd_\bt u = 0,\\
u(0,x,\om) = x^1 \cos \om + x^2 \sin \om.
\end{cases}$$
Differentiating, we obtain 
$$\begin{cases}
2 \gi^{\alp\bt} \rd_\alp u \rd_\bt (\srd_\om u) = 0,\\
(\srd_\om u)(0,x,\om) = -x^1 \sin \om + x^2 \cos \om.
\end{cases}$$
Integrating this transport equation yields the desired estimate.
\end{proof}

\begin{lemma}\label{lem:rdsrdu}
The following holds for $(t,x) \in K$:
\begin{align}
\label{eq:rdsrdu} |\rd_1 \srd_\om u (t,x,\om) + \sin \om |,\, |\rd_2 \srd_\om u (t,x,\om) - \cos \om |,\, | (e_0 \srd_\om u)(t,x,\om) | \ls \ep, \\
\label{eq:rds2rdu} |\rd_1 \srd_\om^2 u (t,x,\om) + \cos \om |,\, |\rd_2 \srd^2_\om u (t,x,\om) + \sin \om |,\, | (e_0 \srd^2_\om u)(t,x,\om) | \ls \ep.
\end{align}
\end{lemma}
\begin{proof}
We will prove the desired bounds by a bootstrap argument, assume that 
\begin{align}
\label{eq:rdsrdu.BA} |\rd_1 \srd_\om u (t,x,\om) + \sin \om |,\, |\rd_2 \srd_\om u (t,x,\om) - \cos \om |,\, | (e_0 \srd_\om u)(t,x,\om) | \leq \ep^{\f 34},\\
\label{eq:rds2rdu.BA} |\rd_1 \srd_\om^2 u (t,x,\om) + \cos \om |,\, |\rd_2 \srd^2_\om u (t,x,\om) + \sin \om |,\, | (e_0 \srd^2_\om u)(t,x,\om) | \leq \ep^{\f 34}.
\end{align}

We first prove \eqref{eq:rdsrdu}. Differentiating the eikonal equation by $\srd_\om \rd_\sigma$, we obtain
\begin{equation}\label{eq:transport.for.rdsrdu}
2\gi^{\alp\bt} \rd_\alp u \rd^2_{\bt \sigma} (\srd_\om u) + (\rd_\sigma \gi^{\alp\bt}) \rd_\alp u \rd_\bt (\srd_\om u) + 2\gi^{\alp\bt} \rd^2_{\alp \sigma} u \rd_\bt (\srd_\om u) = 0.
\end{equation}

Then by the estimates for $u$ in \eqref{eq:u.est.bkgd.1}--\eqref{eq:u.est.bkgd.2}, $g$ in \eqref{eq:g.est.bkgd.1}--\eqref{eq:g.est.bkgd.4} (with $C_h \ls \ep$, see Corollary~\ref{lwp.small}) (and Sobolev embedding) and the bootstrap assumptions \eqref{eq:rdsrdu.BA}--\eqref{eq:rds2rdu.BA}, we have $| (\rd_\sigma \gi^{\alp\bt}) \rd_\alp u \rd_\bt (\srd_\om u) |\ls \ep$ and $|\gi^{\alp\bt} \rd^2_{\alp \sigma} u \rd_\bt (\srd_\om u)| \ls \ep$. Therefore, integrating the transport equation \eqref{eq:transport.for.rdsrdu}, we obtain that $\rd_\sigma \srd_\om u$ is $O(\ep)$ close to its initial value, which gives \eqref{eq:rdsrdu} and improves the bootstrap assumption \eqref{eq:rdsrdu.BA}.

The proof of \eqref{eq:rds2rdu} is similar, except that we differentiate the eikonal equation once more with $\srd_\om$; we omit the details. \qedhere

\end{proof}


\begin{lemma}\label{lem:LL'}
For any $\om, \om' \in \mathbb S^1$, the following holds uniformly in $(t,x) \in K$:
$$-10 |\om - \om'|^2 \leq \gi^{\alp\bt} \rd_\alp u(t,x,\om) \rd_\bt u(t,x,\om') \leq -\f 1{10} |\om - \om'|^2,$$
where $|\om - \om'|$ is to be understood mod $2\pi$. 
\end{lemma}
\begin{proof}
We will suppress the explicit dependence on $(t,x)$. We first compute using \eqref{eq:u.est.bkgd.1}--\eqref{eq:u.est.bkgd.2} (with $C_h \ls \ep$, see Corollary~\ref{lwp.small}) that
\begin{equation}
\begin{split}
&\: \gi^{\alp\bt} \rd_\alp u(\om) \rd_\bt u(\om') = - \f 1{n^2} (e_0 u)(\om) (e_0 u)(\om') + e^{-2\gamma} \de^{ij} (\rd_i u)(\om) (\rd_j u)(\om') \\
= &\: e^{-2\gamma} (-1 + \cos \om \cos \om' + \sin \om \sin \om') + O(\ep) \\
= &\: e^{-2\gamma} (-1 + \cos (\om - \om')) + O(\ep).
\end{split}
\end{equation}
This already implies the desired result when, say, $|\om - \om'| \geq \f{\pi}4$.

For the case $|\om - \om'| <\f{\pi}{4}$, we need to compute the derivative. We compute
\begin{equation}\label{eq:srdom'guu}
\begin{split}
&\: \srd_{\om'}^2 [\gi^{\alp\bt} \rd_\alp u(\om) \rd_\bt u(\om')] = \gi^{\alp\bt} \rd_\alp u(\om) \srd_{\om'}^2 \rd_\bt u(\om') \\
= &\: - \f 1{n^2} (e_0 u)(\om) (e_0 \srd_{\om'}^2 u)(\om') + e^{-2\gamma} \de^{ij} (\rd_i u)(\om) (\rd_j \srd_{\om'}^2 u)(\om') \\
= &\: e^{-2\gamma} (- \cos \om \cos \om' - \sin \om \sin \om') + O(\ep) = - e^{-2\gamma} \cos (\om - \om') + O(\ep),
\end{split}
\end{equation}
where we have used \eqref{eq:u.est.bkgd.1}--\eqref{eq:u.est.bkgd.2} and Lemma~\ref{lem:rdsrdu}.

Integrating \eqref{eq:srdom'guu} with the initial condition $\gi^{\alp\bt} \rd_\alp u(\om) \rd_\bt u(\om) = 0$ and $\gi^{\alp\bt} \rd_\alp u(\om) \rd_\bt \srd_\om u(\om) = 0$  yields the desired conclusion when $|\om - \om'| < \f{\pi}4$. \qedhere
\end{proof}

The following is an immediate consequence of Lemma~\ref{lem:LL'} (after recalling $\om_\bA^{(N)} = \f{2\pi \bA}{N}$) and will be the main angular separation result that we use later in the paper.
\begin{cor}\label{cor:null.adapted}
Fix $N \in \mathbb Z_{>0}$ and let $(g^{(N)},U^{(N)},f^{(N)},u^{(N)})$ be the solution to \eqref{back.with.omega} as in Proposition~\ref{prop:construct.dust}. Denote $u_{\bA}(t,x) = u^{(N)}(t,x,\om_\bA^{(N)})$. Then the following holds uniformly in $K$:
$$\min_{\bA \neq \bB} \Big(- \gi^{\alp\bt} \rd_\alp u_\bA  \rd_\bt u_\bB \Big) \gtrsim N^{-2}.$$
\end{cor}

\subsection{Angularly separated, spatially adapted and null adapted eikonal functions}\label{sec.prepare}


In order to control the interaction of different waves, we need to set up our eikonal functions to be spatially adapted and null adapted. We now introduce these notions. For the definitions (and the remainder of the paper), $\nab$ denotes the Euclidean gradient, and $|\nab u_\bA|^2 = \sum_{i=1}^2 (\rd_i u_\bA)^2$.
\begin{df}\label{def.spatial}
	We say that a set of eikonal functions $\{u_{\bA}\}_{\bA\in \mathcal A}$ is {\bf spatially adapted} if  
	\begin{enumerate}
		\item $|\nab u_{\bA}|> 1, \quad \forall \bA\in\mathcal A$,
		\item $|\nab (u_{\bA_1}\pm u_{\bA_2})|> 1 \quad \forall \bA_1, \bA_2\in \mathcal A,\, \bA_1\neq \bA_2$,
		\item $|\nab (u_{\bA_1}\pm 2u_{\bA_2})|>1, \quad \forall \bA_1, \bA_2\in \mathcal A,\, \bA_1\neq \bA_2$,
		\item $|\nab (u_{\bA_1}\pm u_{\bA_2}\pm u_{\bA_3})|> 1, \quad \forall \bA_1, \bA_2,\bA_3\in \mathcal A,\, \bA_1\neq \bA_2,\bA_3, \bA_1\neq \bA_2$.
	\end{enumerate}
\end{df}

\begin{df}\label{def.null}
	We say that a set of eikonal functions $\{u_{\bA}\}_{\bA\in \mathcal A}$ is {\bf null adapted} if there exists a constant $\eta>0$ such that for any choice of the signs $\pm$, $\pm_1$, $\pm_2$, it holds that
	\begin{enumerate}
		\item $|(g^{-1})^{\alp\bt} \rd_\alp (u_{\bA_1}\pm u_{\bA_2})\rd_\beta(u_{\bA_1}\pm u_{\bA_2}))|> 1, \quad \forall \bA_1, \bA_2\in \mathcal A,\, \bA_1\neq \bA_2$,
		\item $|(g^{-1})^{\alp\bt} \rd_\alp (u_{\bA_1}\pm 2 u_{\bA_2})\rd_\beta(u_{\bA_1}\pm 2 u_{\bA_2}))|> 1, \quad \forall \bA_1, \bA_2\in \mathcal A,\, \bA_1\neq \bA_2$,
		\item $|(g^{-1})^{\alp\bt} \rd_\alp (u_{\bA_1}\pm_1 u_{\bA_2}\pm_2 u_{\bA_3})\rd_\beta(u_{\bA_1}\pm_1 u_{\bA_2}\pm_2 u_{\bA_3}))|> 1, \quad \forall \bA_1, \bA_2,\bA_3\in \mathcal A$ all different.
				\item $|(g^{-1})^{\alp\bt} \rd_\alp u_{\bA_1}\rd_\beta(u_{\bA_2}\pm u_{\bA_3}))|>1, \quad \forall \bA_1, \bA_2, \bA_3\in \mathcal A,\, \bA_2\neq \bA_3$,
		\item $|(g^{-1})^{\alp\bt} \rd_\alp u_{\bA_1}\rd_\beta(u_{\bA_2}\pm 2u_{\bA_3}))|> 1 \quad \forall \bA_1, \bA_2, \bA_3\in \mathcal A,\, \bA_2\neq \bA_3$,
		\item $|(g^{-1})^{\alp\bt} \rd_\alp u_{\bA_1}\rd_\beta( u_{\bA_2}\pm_1 u_{\bA_3} \pm_2 u_{\bA_4})|> 1 \quad \forall \bA_1, \bA_2,\bA_3,\bA_4\in \mathcal A$, $\bA_2 \neq \bA_3 \neq \bA_4 \neq \bA_2$.
	\end{enumerate}
\end{df}

We argue as \cite{HL.HF}, where we start with a family of ``angularly separated'' eikonal functions and show that after rescaling, the rescaled eikonal functions are both spatially adapted and null adpated. We recall the notion of angular separation from \cite{HL.HF}.
\begin{df}\label{def:angular.separation}
We say that $\{u_\bA\}_{\bA \in \mathcal A}$ are \textbf{angularly separated} if there exists $\eta' \in (0,1)$ such that 
$$\f{\de^{ij} \rd_i u_\bA \rd_j u_\bB }{|\nabla u_\bA| |\nabla u_\bB|} < 1 - \eta'$$
whenever $\bA \neq \bB$.
\end{df}

\begin{lm}\label{lem:angular.separation}
Define $\{u_\bA\}_{\bA=0}^{N-1}$ as in Corollary~\ref{cor:null.adapted}. Then $\{u_\bA\}_{\bA=0}^{N-1}$ are angularly separated on the set $[0,1]\times B(0,2R+4)$ with constant $\eta'$ in Definition~\ref{def:angular.separation} bounded below by $\eta' \gtrsim \f 1{N^2}$.
\end{lm}
\begin{proof}
Since $u_\bA$ is null, $e_0 u_\bA = n e^{-\gamma} |\nabla u_\bA|$. Similarly for $u_\bB$. Then 
$$-(g^{-1})^{\alp\bt} \rd_\alp u_\bA \rd_\bt u_\bB = \f 1{n^2} (e_0 u_\bA)(e_0 u_\bB) - e^{-2\gamma} \nabla u_\bA \cdot \nabla u_\bB = e^{-2\gamma} (|\nabla u_\bA| |\nabla u_\bB| - \nabla u_\bA \cdot \nabla u_\bB).$$ 
Dividing by $e^{-2\gamma} |\nabla u_\bA| |\nabla u_\bB|$ throughout, using the estimate in Corollary~\ref{cor:null.adapted} and recalling that $e^{2\gamma} \gtrsim 1$ on a spatially compact set, we obtain the desired bound. \qedhere
\end{proof}

\begin{lm}\label{lemma.adapt}
	Fix $N \in \mathbb Z_{>0}$ and let $(g^{(N)},U^{(N)},f^{\phi,(N)},f^{\varpi,(N)},u^{(N)})$ be the solution as in Proposition~\ref{prop:construct.dust}. Denote $u_{\bA}(t,x) = u^{(N)}(t,x,\om_\bA^{(N)})$ as in Corollary~\ref{cor:null.adapted}.
	Then there exists a collection of positive constants $\{a_\bA\}_{\bA=0}^{N-1}$ such that for $u'_\bA=a_\bA u_\bA$, $\{u'_{\bA}\}_{\bA\in \mathcal A}$ is spatially adapted and null adapted on the set $[0,1]\times B(0,2R+4)$.
Moreover, $\{a_\bA\}_{\bA=0}^{N-1}$ can be chosen so that $a_\bA \geq 1$ and
$$|a_\bA|\leq C(N)$$
for $C(N)$ a large enough constant depending on $N$.
\end{lm}

\begin{proof}
	The proof is an easy adaptation of the proof of Lemma 4.8 in \cite{HL.HF}. 
	
	More precisely, it is done by ordering $\bA_0,\bA_1,\cdots,\bA_{N-1}$, and choosing the $a_\bA$ bigger and bigger to ensure that there is a term in each expression which dominate the other. This is made possible thanks to the hypothesis of angular separation (see Lemma~\ref{lem:angular.separation}). Indeed, following \cite{HL.HF} and using Corollary~\ref{cor:null.adapted}, the choice (defined inductively)
	\begin{equation}\label{eq:choice.of.constants}
	a_{0} = 1,\quad a_{\bA} = c N^2 a_{\bA-1},\quad \bA = 1,\cdots, N-1.
	\end{equation}
	would work for some appropriately chosen $c \geq 1$. This proves the lemma. \qedhere
	
\end{proof}

\textbf{From now on, we define $a_\bA$ according to \eqref{eq:choice.of.constants} so that Lemma~\ref{lemma.adapt} holds.}

\section{Main approximation theorem by high-frequency vacuum spacetimes}\label{sec:main.theorem}

We now turn to the heart of the paper, which is to approximate the solutions to the Einstein--null dust system, constructed in the previous section, by high-frequency vacuum solutions. The following is our main result on this approximation.

\begin{theorem}\label{thm:dust.by.vacuum}
Fix $R \geq 10$. There exist $\ep_2>0$ and $N_0 \in \mathbb Z_{>0}$ such that for any odd $N \geq N_0$, there exists $\lambda_0^{(N)}>0$ such that the following holds. 

Suppose $(\dot{U}(\cdot),\nabla U(\cdot), \breve f (\cdot,\omega), u(\cdot,\omega))\restriction_{\Sigma_0}$ satisfies the assumption (1)--(5) of Theorem~\ref{main.thm.2} and suppose $(g^{(N)},U^{(N)},f^{\phi,(N)},f^{\varpi,(N)},u^{(N)})$ is the sequence of solutions to \eqref{back.with.omega} constructed in Proposition~\ref{prop:construct.dust}. Then if $\ep \leq \ep_2$ and $\lambda \in (0,\lambda_0^{(N)})$, there is a smooth solution $(g^{(N,\lambda)},U^{(N,\lambda)})$ to \eqref{sys} on $[0,1]\times \mathbb R^2$ so that the following bounds hold for all $t \in [0,1]$, with $B(N)$ depending only on $N$ and $R$:
\begin{subequations}
		\begin{empheq}{align}
\sum_{k\leq 5} \lambda^k \| g^{(N,\lambda)} - g^{(N)} \|_{H^k_{-\alp_0}(\Sigma_t)} + \sum_{k\leq 4} \lambda^{k+1} \| \rd_t (g^{(N,\lambda)} - g^{(N)}) \|_{H^k_{-\alp_0}(\Sigma_t)} \leq &\:B(N) \ep^2 \lambda^2,  \label{eq:dust.by.vacuum.est.1} \\
\sum_{k\leq 3} \lambda^k \Big\| U^{(N,\lambda)} - U^{(N)} - \lambda \sum_{\bA=0}^{N-1} (a_\bA^{(N)})^{-1} F^{U,(N)}_\bA \cos(\tfrac{a_\bA^{(N)} u^{(N,\lambda)}_\bA}{\lambda}) \Big\|_{H^{k+1}(\Sigma_t)} \leq &\: B(N) \ep \lambda^2, \label{eq:dust.by.vacuum.est.2} \\
\sum_{k\leq 3} \lambda^k \Big\| \rd_t U^{(N,\lambda)} - \rd_t U^{(N)} + \sum_{\bA=0}^{N-1} F^{U,(N)}_\bA \rd_t u^{(N,\lambda)}_\bA \sin(\tfrac{a_\bA u^{(N,\lambda)}_\bA}{\lambda}) \Big\|_{H^k(\Sigma_t)} \leq &\: B(N) \ep \lambda^2, \label{eq:dust.by.vacuum.est.3}
		\end{empheq}
\end{subequations}
for some phase function $u^{(N,\lambda)}_\bA$ satisfying the bound 
\begin{equation}
\sum_{k\leq 4} \lambda^k \| u^{(N,\lambda)}_\bA - u^{(N)}_\bA \|_{H^{k}(\Sigma_t \cap \mathrm{supp}(f^{U,(N)}_{\cdot}))} + \sum_{k\leq 3} \lambda^k \| \rd_t (u^{(N,\lambda)}_\bA - u^{(N)}_\bA) \|_{H^{k}(\Sigma_t \cap \mathrm{supp}(f^{U,(N)}_{\cdot}))} \leq B(N) \ep \lambda^2.
\end{equation}
Here, $a_\bA$ are constants as in Lemma~\ref{lemma.adapt}, and we have denoted $F^{U,(N)}_\bA  =(F^{\phi,(N)}_\bA , F^{\varpi,(N)}_\bA )$.
\end{theorem}

Notice that we have only given the estimates that are the easiest to state and that more precise bounds are available from (and indeed necessary for) the proof.

At this point, we show that our main theorem (Theorem~\ref{main.thm.2}) follows as a consequence of Theorem~\ref{thm:dust.by.vacuum}.
\begin{proof}[Proof of Theorem~\ref{main.thm.2} assuming Theorem~\ref{thm:dust.by.vacuum}] We choose our sequence $(g_i, U_i)$ to be $(g^{(N_i,\lambda_{N_i})}, U^{(N_i,\lambda_{N_i})})$ given by Theorem~\ref{thm:dust.by.vacuum}, with $N_i \geq N_0$ odd, $N_i \to \infty$, and $\lambda_{N_i} \in (0, \lambda_0^{(N_i)})$ chosen to be suitably small. We now show that this choice satisfies the conclusion of Theorem~\ref{main.thm.2}.

\pfstep{Step~1: $C^0$ convergence} For $g$, the $C^0_{\mathrm{loc}}$ convergence follows immediately from the estimates \eqref{eq:dust.by.vacuum.est.1} and the embedding $H^{2}_{-\alp_0} \hookrightarrow C^0_{\mathrm{loc}}$ in $2$ dimensions. For $\phi$, it follows from \eqref{eq:dust.by.vacuum.est.2}, \eqref{eq:dust.by.vacuum.est.3} and the bound $\| \cdot \|_{L^\i(\Sigma_t)} \ls \|\cdot \|_{L^2(\Sigma_t)}^{\f 12}\|\cdot \|_{H^2(\Sigma_t)}^{\f 12}$ in $2$ dimensions.

\pfstep{Step~2: Weak $W^{1,2}$ convergence} The estimates \eqref{eq:dust.by.vacuum.est.1} for the metric components (together with Sobolev embedding) are sufficient to show \emph{strong} convergence of the first derivatives of $g$. For $\rd_\mu\phi$, in order to check weak $L^2$ convergence, it suffices to check that for suitably small $\lambda_N \in (0, \lambda^{(N)}_0]$, 
\begin{equation}
\lim_{N\to \infty} \int_{\Sigma_t} \eta(x) \sum_{\bA=0}^{N-1} F^{\phi,(N)}_\bA(t,x) \rd_\mu u^{(N)}_\bA \sin(\tfrac{a_\bA^{(N)} u^{(N)}_\bA(t,x)}{\lambda_N}) \, \ud x = 0.
\end{equation}
This can in turn be established using 
$$\sin(\tfrac{a_\bA u^{(N)}_\bA(t,x)}{\lambda_N}) = -\f{\lambda^2 \Delta(\sin (\tfrac{a_\bA^{(N)} u^{(N)}_\bA(t,x)}{\lambda_N}))}{|a_\bA^{(N)}\nab u^{(N)}|^2 }+\f{\lambda a_\bA^{(N)} (\Delta u^{(N)}_\bA(t,x))\cos (\tfrac{a_\bA^{(N)} u^{(N)}_\bA(t,x)}{\lambda_N}))}{|a_\bA^{(N)}\nab u^{(N)}|^2 }$$
and integration by parts.

\pfstep{Step~3: Uniform boundedness up to $W^{1,4}$} As before, the estimate for $g$ is obviously strong enough. It thus suffices to bound $\sum_{\bA=0}^{N-1} F^{\phi,(N)}_\bA \rd_t u^{(N,\lambda)}_\bA \sin(\tfrac{a_\bA u^{(N,\lambda)}_\bA}{\lambda})$ uniformly in $L^4(\Sigma_t)$, which follows from Lemma~\ref{lem:orthogonality} and the estimates in Theorem~\ref{thm:dust.by.vacuum}. \qedhere

\end{proof}

It thus remains to prove Theorem~\ref{thm:dust.by.vacuum}. This will occupy the remainder of the paper; see Section~\ref{sec:together} for the conclusion of the proof. \textbf{For the rest of the paper, we will consider a solution to \eqref{back} given by Proposition~\ref{prop:construct.dust} for some odd $N \in \mathbb Z_{\geq 1}$. We will refer this as the \emph{background solution}. For suitably defined high-frequency perturbations, we will prove the estimates in Theorem~\ref{thm:dust.by.vacuum}.}

\textbf{In order to keep the notation lean, we will drop all the indices $^{(N)}$ associated with the background solution. Instead, we use $_0$ or $^0$ to denote the background solution, i.e., we write $g_0$, $\phi_0$, $u_{\bA}^0$, etc. Note that despite not writing $N$ explicitly, it is important to keep track of the dependence of the estimates on $N$.}

From now on, we use the convention that \textbf{$C>0$ will denote a constant which is independent of $N$, while $C(N) >0$ will denote a constant which is dependent on $N$.} These constants may change from line to line. We will also use the shorthand \textbf{$h_1 \ls h_2$ to denote $h_1 \leq C h_2$, where $C>0$ is a constant independent of $N$}, consistent with the convention above. (Later on we will also introduce $A(N)$, $C_b(N)$ and $C_b'(N)$; see Section~\ref{sec:bootstrap}.)

\section{The parametrix constuction}\label{sec:parametrix}

To approximate solutions to Einstein null dust system by solutions to Einstein vacuum equations, we want to construct well chosen high-frequency solutions, similarly to \cite{HL.HF}. We introduce an ansatz up to the second order in $\lambda$, in which all the amplitude functions are defined entirely in term of the background (which mean the solution to \eqref{back}), but where the phases $u_\bA$ are actual solutions to the eikonal equation corresponding to the metric of the solution to Einstein vacuum equations. 

For each of the quantities of interest, we will introduce a decomposition as a sum of a background term, some explicit highly oscillatory terms (with one or two powers of $\lambda$) which depend only on the background quantities, and an error term that is to be controlled with nonlinear estimates. To unify the notation, the error quantity will be denoted with a tilde, e.g., $\wht \phi$, $\wht \varpi$, $\wht \gamma$, etc.

Introduce the following decomposition of $\phi$ and $\varpi$:
\begin{equation}\label{eq:wave.para}
\phi = \phi_0 + \phi_1 + \phi_2 + \wht \phi,\quad \varpi = \varpi_0 + \varpi_1 + \varpi_2 + \wht \varpi, 
\end{equation}
where $(\phi_0,\varpi_0)$ are as in the background solutions, and 
\begin{subequations}\label{eq:wave.para.2}
		\begin{empheq}{align}
\phi_1 = &\: \lambda \sum_\bA  a_\bA^{-1}F^\phi_\bA \cos(\tfrac{a_\bA u_\bA}{\lambda})\\
\phi_2 = &\: \lambda^2 \Big( \sum_\bA a_\bA^{-1}\Big( F^{1,\phi}_\bA \sin(\tfrac{a_\bA u_\bA}{\lambda}) + F_\bA^{2,\phi}\cos(\tfrac{2a_\bA u_\bA}{\lambda})\Big)+\sum_{\pm,\bA,\bB:\bA\neq \bB} (\pm 1) F^\phi_{\bA,\bB} \cos (\tfrac{a_\bA u_\bA \pm a_\bB u_\bB}{\lambda }) \Big), \\
\varpi_1 = &\: \lambda \sum_\bA  a_\bA^{-1}F^\varpi_\bA \cos(\tfrac{a_\bA  u_\bA}{\lambda}),  \\
\varpi_2 = &\: \lambda^2 \Big( \sum_\bA a_\bA^{-1} \Big( F^{1,\varpi}_\bA \sin(\tfrac{a_\bA u_\bA}{\lambda}) +  F_\bA^{2,\varpi}\cos(\tfrac{2a_\bA u_\bA}{\lambda}) \Big) +\sum_{\pm,\bA,\bB:\bA\neq \bB} (\pm 1) F^\varpi_{\bA,\bB} \cos (\tfrac{a_\bA u_\bA \pm a_\bB u_\bB}{\lambda }) \Big),
		\end{empheq}
\end{subequations}
where $F_\bA^\phi$ and $F_\bA^\varpi$ are as in the background solution to \eqref{back}, while $F^{1,\phi}_\bA$, $F_\bA^{2,\phi}$, $F^\phi_{\bA,\bB}$ and their counterparts for $\varpi$ will be introduced in Section~\ref{sec.para.phi}.

For each metric component $\mfg$ (i.e., $\mfg\in \{\gamma,n,\beta^i\}$), introduce the following decomposition
\begin{equation}\label{eq:g.para}
\mfg=\mfg_0+\mfg_2+\wht{\mfg},
\end{equation}
where $\mfg_0$ is the corresponding metric component of the background solution and 
\begin{equation}\label{g1.def}
	\mfg_2= \lambda^2\Big( \sum_{\bA} \q G_{2,\bA}(\mfg) \cos(\tfrac{2a_\bA u_\bA}{\lambda})
	+\sum_{\bA} \q G_{1,\bA} (\mfg)\sin(\tfrac{a_\bA u_\bA}{\lambda})
	+\sum_{\pm, \bA,\bB: \bA \neq \bB} \q G_{\pm,\bA,\bB} (\mfg)\cos(\tfrac{a_\bA u_\bA \pm a_\bB u_{\gra B}}{\lambda})\Big)
\end{equation}
with $\q G_{2,\bA}(\mfg)$, etc.~to be introduced in Section~\ref{sec.para.g}.

Introduce the decomposition of the eikonal function $ u_\bA$ into
\begin{equation}\label{eq:u.para}
u_\bA = u^0_\bA + u_\bA^2+ \wht{u}_\bA,
\end{equation}
with $u^0_\bA$ the background eikonal function and $u^2_\bA$ is defined implicitly\footnote{The definition is implicit because $u_\bA$ on the right-hand side also depends on $u^2_\bA$. See Lemma~\ref{lem:u2.well.defined}.} by
\begin{equation}\label{u1.def}
	u^2_\bA = \lambda^2 a_{\bA}^{-1} \mathfrak v_\bA \sin(\tfrac{a_\bA u_\bA}{\lambda}),
\end{equation}
where $\mathfrak v_\bA$ is defined in Section \ref{sec.para.u}. 
 
In addition, $\chi_\bA = \Box_g u_\bA$ is decomposed into
\begin{equation}\label{eq:chi.para}
\chi_\bA = \chi_\bA^0 +\chi_\bA^1 +\wht \chi_\bA,
\end{equation}
with
\begin{equation}\label{chi1.def}
	\chi_\bA^1 = \lambda \Big(\sum_{\bB } \mathfrak X^{(\bA)}_{2,\bB} \sin ( \tfrac{2a_\bB u_\bB}{\lambda})
	+\sum_{\bB } \mathfrak X^{(\bA)}_{1,\bB} \cos ( \tfrac{a_\bB u_\bB}{\lambda})
	+ \sum_{\substack{\pm,\bB,\bC: \\\bB \neq \bA, \bC \neq \bA,\bB}}\mathfrak X^{(\bA)}_{\pm,\bB,\bC} \sin ( \tfrac{a_\bB u_\bB\pm a_\bC u_\bC}{\lambda})\Big).
\end{equation}
The terms $\mathfrak X^{(\bA)}$ will be defined in Section~\ref{sec.para.u}.

In the notations introduced above, the quantities $\phi_0,\varpi_0, \mfg_0, F_\bA^{\phi, \varpi}, u^0_\bA,\chi_\bA^0 $ are given by the solution of the background system \eqref{back}. The quantities $F_\bA^{\phi,\varpi,1,2}, F_{\bA,\bB}^{\phi, \varpi},\mfg_2,u_\bA^2,\chi_\bA^1$ are also background quantities in the sense that they can be computed directly from the background solution. The error quantities $\wht \phi$, $\wht \varpi$, $\wht \mfg$, $\wht u_\bA$ and $\wht{\chi}_\bA$ are coupled. They will be controlled using a bootstrap argument.

\subsection{Terms in the parametrix for $\phi$ and $\varpi$}\label{sec.para.phi}
Recall from \eqref{eq:wave.para}, \eqref{eq:wave.para.2} that we need to prescribe $F_\bA^{1,\phi}$, $F_\bA^{1,\varpi}$, $F^{2,\phi}_\bA$, $F^{2,\varpi}_\bA$ and $F^{\phi}_{\bA,\bB}$, $F^{\varpi}_{\bA,\bB}$ (for $\bA \neq \bB$). The first four quantities are prescribed by solving transport equations:
\begin{subequations}\label{eq:F1.F2.transport}
		\begin{empheq}[left=\empheqlbrace]{align}
	\label{whtfi}
	2 L_{\bA}^0  F^{1,\phi}_\bA + \chi_{\bA}^0 F^{1,\phi}_\bA+e^{-4\phi_0}(g_0^{-1})^{\alpha \beta}\partial_\alpha \varpi_0\partial_\beta u^0_\bA F_\bA^{1,\varpi}
	=&\:  - a_\bA^{-1} \Box_{g_0} F^\phi_\bA +A_\bA^{1,\phi}, \\
	\label{whtpi}
	2L_{\bA}^0  F^{1,\varpi}_\bA + \chi_{\bA}^0 F^{1,\varpi}_\bA -4g_0^{-1}(\ud \varpi_0, \ud u_\bA^0) F^{1,\phi}_\bA -4g_0^{-1}(\ud \phi_0, \ud u_\bA^0)  F^{1,\varpi}_\bA
		= &\: - a_\bA^{-1} \Box_{g_0} F^\varpi_\bA +A_\bA^{1,\varpi}, \\
	\label{whtfi2}2 L^0_\bA \partial_\beta  F^{2,\phi}_\bA + \chi_{\bA}^0 F^{1,\phi}_\bA F_\bA^{2,\phi} +e^{-4\phi_0}g_0^{-1}(\ud \varpi_0, \ud u^0_\bA) F_\bA^{2,\varpi}
	=&\: A_\bA^{2,\phi}, \\
	\label{whtpi2}2L_\bA^0 \partial_\beta F^{2,\varpi}_\bA + \chi_{\bA}^0 F^{1,\phi}_\bA F^{2,\varpi}_\bA - 4g^{-1}_0(\ud \varpi, \ud u^0_\bA) F^{2,\phi}_\bA -4 g_0^{-1}(\ud \phi, \ud u^0_\bA)  F^{2,\varpi}_\bA
	=&\: A_\bA^{2,\varpi},
	\end{empheq}
\end{subequations}
where
$L_{\bA}^0 = (g_0^{-1})^{\alp\bt}\rd_\alp u_{\bA} \rd_\bt$ and the $A_\bA$ terms are compactly supported in $[0,1]\times B(2R+2)$, depend only on the background (and can be explicitly written down, see Proposition~\ref{prop:wave.error}) and obey the following bounds for $t \in [0,1]$:
	\begin{equation}\label{A.bound.1}
		\|A_\bA\|_{C^{8}(\Sigma_t)}+\|\rd_tA_\bA\|_{C^{7}(\Sigma_t)}\leq C \ep^2.
	\end{equation}
We set zero initial data for all of the quantities $F_\bA^{1,\phi}$, $F_\bA^{1,\varpi}$, $F^{2,\phi}_\bA$, $F^{2,\varpi}_\bA$. Consequently, the following estimates hold for all $t\in [0,1]$:
\begin{equation}\label{eq:bkgd.bd:FA}
\|F_\bA^{1,U}\|_{C^7(\Sigma_t)}+\|F_\bA^{2,U}\|_{C^{8}(\Sigma_t)}
+\|F_{\bA}^{U}\|_{C^9(\Sigma_t)}+\|\rd_t F_\bA^{1,U}\|_{C^6(\Sigma_t)}+\|\rd_tF_\bA^{2,U}\|_{C^{7}(\Sigma_t)}
+\|\rd_tF_{\bA}^{U}\|_{C^8(\Sigma_t)} \leq C \ep^2,
\end{equation}
where we used $F_\bA^{1,U}$ to denote either $F_\bA^{1,\phi}$ or $F_\bA^{1,\varpi}$ (and similarly for other quantities).

The quantities $F^{\phi}_{\bA,\bB}$ and $F^{\varpi}_{\bA,\bB}$ (for $\bA \neq \bB$) are defined explicitly by
\begin{align}
\label{fabphi}	F^\phi_{\bA,\bB} = - \frac{e^{-4\phi_0} F_\bA^\varpi F_\bB^\varpi}{8 a_\bA a_\bB}, \quad F^\varpi_{\bA,\bB} = \f{F_\bA^\phi F_\bB^\varpi}{a_\bA a_\bB}.
	\end{align} 
Given the bounds on $F_\bA^U$ and $\phi_0$, the following thus holds for $t \in [0,1]$:
	\begin{equation}\label{eq:bkgd.bd:FAB}
	\| F_{\bA,\bB}^{U} \|_{C^9(\Sigma_t)} \leq C \ep^4.
	\end{equation}

\subsection{Terms in the parametrix for $g$}\label{sec.para.g}

For the metric $g$, we construct a parametrix for each metric component.  More precisely, we will define
$$n=n_0+n_2+\wht n,\quad \gamma=\gamma_0+\gamma_2+\wht \gamma,\quad \beta^{i}=\beta^{i}_0+\beta^{i}_2+\wht \beta^i.$$
We will use the convention that $\mfg$ denotes one of these metric components. In this notation, the decomposition above reads
\begin{equation}\label{g.para.def}
	\mfg=\mfg_0+\mfg_2+\wht \mfg.
\end{equation}
Before defining $\mfg_2$, we first define a notation to treat the elliptic equations satisfied by the metric components simultaneously. We consider the equations \eqref{elliptic.1}--\eqref{elliptic.4} in the vacuum case. After some algebraic manipulations (see \cite[Section~4.4.2]{HL.HF}), we obtain the following system of elliptic equations for the metric components:
\begin{align}
	&\Delta \gamma = -|\nabla \phi|^2-\frac{1}{4}e^{-4\phi}|\nabla \varpi|^2-\frac{e^{2\gamma}}{n^2} ((e_0 \phi)^2+\frac{1}{4}e^{-4\phi}(e_0 \varpi)^2)-\frac{e^{2\gamma}}{8n^2}|\mathfrak L\beta|^2,\label{elliptic.g.1}\\
	&\Delta n =\f{e^{2\gamma}}{4n}|\mathfrak L\beta|^2+ \frac{2e^{2\gamma}}{n}((e_0 \phi)^2+\frac{1}{4}e^{-4\phi}(e_0 \varpi)^2),\label{elliptic.g.2}\\
	& \Delta \beta^j={\delta^{ik}}\delta^{j\ell}\rd_k\left(\log(ne^{-2\gamma})\right)(\mathfrak L\beta)_{i\ell}-4 \delta^{ij}((e_0 \phi)( \partial_i \phi)+\frac{1}{4}e^{-4\phi}(e_0 \varpi)(\partial_i \varpi). \label{elliptic.g.3}
\end{align}

To handle \eqref{elliptic.g.1}--\eqref{elliptic.g.3} with uniform notations, for $\mfg \in\{ \gamma, n, \bt^j\}$, we write 
\begin{equation}\label{g.elliptic}
	\Delta \mfg= {\bf \Gamma}(\mfg)^{\mu\nu} \langle \rd_\mu U, \rd_\nu U \rangle +\Upsilon(\mfg),
\end{equation}
where
\begin{align}
	{\bf \Gamma}(\gamma)^{\mu\nu}:=& 
	\left(\begin{array}{ccc}-\f{e^{2\gamma}}{n^2} & \f{e^{2\gamma}}{n^2}\beta^1 & \f{e^{2\gamma}}{n^2}\beta^2\\
		\f{e^{2\gamma}}{n^2}\beta^1 & -1-\f{e^{2\gamma}}{n^2}(\beta^1)^2 & -\f{e^{2\gamma}}{n^2}\beta^1 \beta^2\\
		\f{e^{2\gamma}}{n^2}\beta^2 & -\f{e^{2\gamma}}{n^2}\beta^1 \beta^2 & -1-\f{e^{2\gamma}}{n^2}(\beta^2)^2
	\end{array}
	\right),\label{G.g}\\
	{\bf \Gamma}(n)^{\mu\nu}:=& 
	\left(\begin{array}{ccc}\f{2e^{2\gamma}}{n} & -\f{2e^{2\gamma}}{n}\beta^1 & -\f{2e^{2\gamma}}{n}\beta^2\\
	-	\f{2e^{2\gamma}}{n}\beta^1 & \f{2e^{2\gamma}}{n}(\beta^1)^2 & \f{2e^{2\gamma}}{n}\beta^1 \beta^2\\
	-	\f{2e^{2\gamma}}{n}\beta^2 & \f{2e^{2\gamma}}{n}\beta^1 \beta^2 & \f{2e^{2\gamma}}{n}(\beta^2)^2
	\end{array}
	\right),\label{G.N}\\
	{\bf \Gamma}(\beta^1)^{\mu\nu}:=& 
	\left(\begin{array}{ccc} 0 & -2 & 0\\
		-2 & 4\beta^1 & 2\beta^2\\
		0 & 2\beta^2 & -0
	\end{array}
	\right), \quad {\bf \Gamma}(\beta^2)^{\mu\nu}:= 
	\left(\begin{array}{ccc} 0 & 0 & -2\\
		0 & 0 & 2\beta^1\\
		-2 & 2\beta^1 & 4\beta^2
	\end{array}
	\right),\label{G.b}
\end{align}
and
\begin{align}
	\Upsilon(\gamma):=&-\frac{e^{2\gamma}}{8n^2}|\mathfrak L\beta|^2,\label{up1}\\
	\Upsilon(n):=&\f{e^{2\gamma}}{4n}|\mathfrak L\beta|^2,\label{up2}\\
	\Upsilon(\beta^i):=&{\delta^{{j}k}}\delta^{{i}\ell}\left((\rd_k\log n)(\mathfrak L\beta)_{{j}\ell}-2(\rd_k\gamma)(\mathfrak L\beta)_{{j}\ell}\right).\label{up3}
\end{align}
(Recall that $\langle \partial_\alpha U,\partial_\beta U \rangle = \partial_\alpha \phi \partial_\beta \phi + \frac{1}{4}e^{-4\phi}\partial_\alpha \varpi \partial_\beta \varpi$.)

Define also  $\bfG_0(\gamma)$, $\bfG_0(n)$, $\bfG_0(\beta^i)$,$\Upsilon_0(\gamma)$, $\Upsilon_0(N)$ and $\Upsilon_0(\beta^i)$ in a similar manner as above except that all of the metric components $\gamma$, $n$ and $\beta^i$ are replaced by their background value $\gamma_0$, $n_0$ and $\beta_0^i$. In particular, the background metric components, which are denoted by $\mfg_0$, satisfies a similar equation with an extra term
\begin{equation}\label{g0.elliptic}
	\Delta \mfg_0= {\bf \Gamma}_0(\mfg)^{\mu\nu} \langle \rd_\mu U_0, \rd_\nu U_0 \rangle +\f 12 \sum_{\bA} F^2_{\bA}{\bf \Gamma}_0(\mfg)^{\mu\nu}(\rd_\mu u_{\bA})(\rd_\nu u_{\bA})+\Upsilon_0(\mfg),
\end{equation}
where $F_{\bA}^2 = (F^{\phi}_{\bA})^2+\frac{e^{-4\phi}}{4}(F^{\varpi}_\bA)^2$.

With this notation, we now turn to the definition of $\mfg_2$ (see \eqref{g.para.def}), which is given by
\begin{equation}\label{eq:mfg2.form}
		\mfg_2= \lambda^2\sum_{\bA} \Big( \q G_{2,\bA}(\mfg) \cos(\tfrac{2a_\bA u_\bA}{\lambda})
		+ \q G_{1,\bA} (\mfg)\sin(\tfrac{a_\bA u_\bA}{\lambda}) \Big)
	+\lambda^2\sum_{\pm,\bA,\bB:\bB \neq \bA}\q G_{\pm,\bA,\bB}(\mfg) \cos(\tfrac{a_\bA u_\bA \pm a_\bB u_{\gra B}}{\lambda}),
	\end{equation}
	where
\begin{subequations}
		\begin{empheq}{align}
		 \q G_{2,\bA}(\mfg)& =-\f 18{\bf \Gamma}_0(\mfg)^{\mu\nu} \frac{ ((F^\phi_\bA)^2+\frac{1}{4}e^{-4\phi_0}(F_\bB^\varpi)^2)}{a_\bA^2|\nabla u^0_\bA|^2}  (\partial_\mu u^0_\bA)( \partial_\nu u^0_\bA), \label{eq:g2.def.1} \\
	  \q G_{1,\bA} (\mfg)&=-2{\bf \Gamma}_0(\mfg)^{\mu\nu}\frac{ F_\bA^\phi}{a_\bA^2|\nabla u^0_\bA|^2} (\partial_\mu \phi_0 )(\partial_\nu u^0_\bA) -\frac{1}{2}e^{-4\phi_0}{\bf \Gamma}_0(\mfg)^{\mu\nu}\frac{F_\bA^\varpi}{a_\bA^2 |\nabla u^0_\bA|^2} (\partial_\mu \varpi_0 )(\partial_\nu u^0_\bA), \label{eq:g2.def.2}\\
	  \q G_{\pm,\bA,\bB} &=-{\f 12}{\bf \Gamma}_0(\mfg)^{\mu\nu}\frac{(\mp 1)\cdot  (F^\phi_\bA F^\phi_{\gra B}+\frac{1}{4}e^{-4\phi_0}F^\varpi_\bA F^\varpi_\bB)}{|\nabla (a_\bA u^0_\bA \pm a_\bB u^0_{\gra B})|^2} (\partial_\mu u^0_\bA)( \partial_\nu u^0_{\gra B}).\label{eq:g2.def.3}
	  \end{empheq}
\end{subequations}
Recall that the (background) collection $\{a_\bA u^0_{\bA}\}_{\bA=0}^{N-1}$ is spatially adapted (see Lemma~\ref{lemma.adapt}) and therefore the above expression is well-defined. Moreover, each of the $\q G(\mfg)$ term is compactly supported in $B(2R+2)$ for each fixed $t$ and is bounded as follows for all $t \in [0,1]$:
\begin{equation}\label{eq:calG.est}
\| \q G(\mfg)\|_{C^{9}(\Sigma_t)}+\| \rd_t\q G(\mfg)\|_{C^{8}(\Sigma_t)} \leq C\ep^2.
\end{equation}

\subsection{Terms in the parametrix for $u_\bA$ and $\chi_\bA$}\label{sec.para.u}

The construction and estimates for the parametrix for $u_\bA$ and $\chi_\bA$ will only be relevant in the region $\cup_{t \in [0,1]} K_t$ (see \eqref{eq:Kt.def}), which is a causal region in which $\phi$ and $\varpi$ are supported. We will thus focus on this region.

We recall that $u_\bA$ is a solution to the equation $g^{\alpha \beta}\partial_\alpha u_\bA\partial_\beta u_\bA=0$, which can also be written as
$$-\frac{1}{n^2}(e_0 (u_\bA))^2 +e^{2\gamma}|\nabla u_{\bA}|^2 =0.$$

Recall that $u^2_\bA$ is defined by $u^2_\bA = \lambda^2 a_{\bA}^{-1} \mathfrak v_\bA \sin(\tfrac{a_\bA u_\bA}{\lambda})$ (see \eqref{eq:u.para}). The function $\mathfrak v_\bA$ is defined to be the solution to the transport equation
\begin{equation}\label{defu1}
2g^{\alpha \beta}_0\partial_\alpha u_\bA^0 \partial_\beta \mathfrak v_\bA = - \frac{1}{a_\bA}g_0^{\alpha \beta}(4F_\bA^\phi\partial_\alpha \phi_0 \partial_\beta u^0_\bA + e^{-4\phi_0}F_\bA^\varpi\partial_\alpha \varpi_0 \partial_\beta u^0_\bA)
\end{equation}
with zero initial data. As a consequence, the following estimates hold for $t \in [0,1]$: 
\begin{equation}\label{eq:v.est}
\| \mathfrak v_\bA \|_{C^{9}(K_t)} +\| \rd_t \mathfrak v_\bA \|_{C^{8}(K_t)} + \| \rd_t^2 \mathfrak v_\bA\|_{C^7(K_t)} \leq C a_\bA^{-1} \ep^2.
\end{equation}
It will be important (see proof of Proposition~\ref{phi.phi.diff}) to keep track of the factor $a_\bA^{-1}$ (instead of just bounding it above by $\leq 1$). We note also that the bounds from the $\rd_t$ derivatives in \eqref{eq:v.est} are derived using the transport equation \eqref{defu1} together with the fact that $2g^{\alpha t}_0\partial_\alpha u_\bA^0$ is uniformly bounded away from $0$.

We now turn to $\chi_\bA$. It satisfies the Raychaudhuri equation (see \cite[Appendix C]{HL.elliptic} for a derivation)
\begin{equation}\label{eq:Raychaudhuri}
L_\bA(\chi)+\chi^2=-R_{L_\bA L_\bA}=-(L_\bA\phi)^2-e^{-4\phi}(L_\bA\varpi)^2,
\end{equation}
where $L_\bA =g^{\alpha \beta}\partial_\alpha u_\bA \partial_\beta.$
Recall from \eqref{chi1.def} that $\chi_\bA^1$ is given by
$$\chi_\bA^1 = \lambda \sum_{\bB } \mathfrak X^{(\bA)}_{2,\bB} \sin ( \tfrac{2a_\bB u_\bB}{\lambda})
+\lambda \sum_{\bB } \mathfrak X^{(\bA)}_{1,\bB} \cos( \tfrac{a_\bB u_\bB}{\lambda})
+\lambda \sum_{\substack{\pm,\bB,\bC: \\ \bB\neq \bA, \bC \neq \bA,\bB}}\mathfrak X^{(\bA)}_{\pm,\bB,\bC} \sin ( \tfrac{a_\bB u_\bB\pm a_\bC u_\bC}{\lambda}).$$
The terms $\mathfrak X^{(\bA)}_{2,\bB}$ and $\mathfrak X^{(\bA)}_{1,\bB}$ are set in a slightly different manner: they are to be defined pointwise when $\bB \neq \bA$, while they are defined through transport equations when $\bB = \bA$. The distinction is made according to whether the oscillations are transversal to $L_\bA$. More precisely, we define
\begin{subequations}
		\begin{empheq}{align}
	\mathfrak X^{(\bA)}_{2,\bB}&=-\frac{\mathcal H^{0,(\bA)}_{2,\bB}}{2a_\bB (L_\bA^0(u^0_\bB))}  \qquad \text{ for }\bA \neq \bB, \label{eq:chi1.def.1.1}\\
	\mathfrak X^{(\bA)}_{1,\bB}&=\frac{\mathcal H^{0,(\bA)}_{1,\bB}}{a_\bB (L_\bA^0(u^0_\bB))}  \qquad \text{ for }\bA \neq \bB, \label{eq:chi1.def.1.2}\\
	\mathfrak X^{(\bA)}_{\pm,\bB,\bC}&=-\frac{\mathcal H^{0,(\bA)}_{\pm,\bB,\bC}}{L^0_\bA(a_\bB u_\bB^0 \pm a_\bC u_\bC^0)}, \label{eq:chi1.def.1.3}
		\end{empheq}
\end{subequations}
and $\mathfrak X^{(\bA)}_{1,\bA} $ and $\mathfrak X^{(\bA)}_{2,\bA}$ are defined  as solutions to the transport equation
\begin{subequations}
		\begin{empheq}{align}
	\label{eq:chi1.def.2.1}&	L_\bA^0 (	\mathfrak X^{(\bA)}_{1,\bA})  +2 \chi_\bA^0 	\mathfrak X^{(\bA)}_{1,\bA} = - \q H^{1,(\bA)}_{1,\bA}, \\
\label{eq:chi1.def.2.2}	&L_\bA^0 (	\mathfrak X^{(\bA)}_{2,\bA})  +2 \chi_\bA^0 	\mathfrak X^{(\bA)}_{2,\bA} = - \q H^{1,(\bA)}_{2,\bA},
		\end{empheq}
\end{subequations}
where $L_\bA^0= g_0^{\alpha \beta}\partial_\alpha u_\bA^0\partial_\beta$, and $\mathcal H^{0,(\bA)}_{2,\bB}, \cdots, \q H^{1,(\bA)}_{2,\bA}$ are background expressions which are to be defined in Lemma~\ref{lm:rll}. It will be proven in Lemma~\ref{lm:rll} that all the relevant\footnote{Note that in Lemma~\ref{lm:rll}, there are $\q H$ terms for which we only control $7$ derivatives, but the terms appearing here all obey the bound \eqref{eq:H.est.good}.} $\mathcal H$ quantities are supported in $[0,1]\times B(0,2R+2)$ and satisfy $\|\mathcal H\|_{C^8(K_t)} \leq C(N) \ep$. Thus, we have, for $t \in [0,1]$,
$$\| \mathfrak X \|_{C^8(K_t)} \leq C(N) \ep^2,$$
as long as these bounds are also satisfied by the initial data for $\mathfrak X^{(\bA)}_{1,\bA}$ and $\mathfrak X^{(\bA)}_{2,\bA}$ (which will be checked in Lemma~\ref{lem:chi.data}).

\section{The initial data}\label{sec.id}

In this section, we discuss the initial data and show that the initial data can be prescribed in a way that is consistent with the intended parametrix in Section~\ref{sec:parametrix}. In the three subsections of this section, we will discuss the data for the wave part, the eikonal functions $u_{\bA}$ and the second fundamental form $\chi_{\bA}$, respectively. Note that since the metric part satisfies elliptic equations, the elliptic estimates can be handled in the same way for every fixed time, and we do not need to separately deal with the initial data.

\subsection{Initial data for the wave part: almost orthogonality and the constraint equations}

\begin{lemma}\label{lem:orthogonality}
Suppose $\mathfrak F$ is compactly supported. Suppose also that $\{v_{\bA}\}_{\bA = 0}^{N-1}$ satisfies
$$\sup_\bA \| v_\bA \|_{C^2(\mathrm{supp}(\mathfrak F))}  \ls 1.$$
Then the following holds:
\begin{enumerate}
\item Suppose that on $\mathrm{supp}(\mathfrak F)$,
$$\bA \neq \bB \implies |\nabla (v_\bA \pm v_\bB)| \geq 1,$$
then
$$\| \sum_{\bA} \mathfrak F_{\bA} \sin(\tfrac{v_\bA}{\lambda}) \|_{L^2(\Sigma_t)} \ls \Big( \sum_{\bA} \| \mathfrak F_{\bA}\|_{L^2(\Sigma_t)}^2\Big)^{\f 12} + \lambda^{\f 12} N \sup_{\bA} \| \mathfrak F_{\bA}\|_{L^2(\Sigma_t)}^{\f 12}\| \mathfrak F_{\bA}\|_{H^{1}(\Sigma_t)}^{\f 12}.$$
\item Suppose that on $\mathrm{supp}(\mathfrak F)$,
\begin{equation}\label{eq:use.spatial.separation.for.orthogonality}
\bA \neq \bB, \bC,\bD \implies |\nabla(v_{\bA} \pm_1 v_\bB \pm_2 v_\bC \pm_3 v_\bD)|\geq 1 \hbox{ for any choice of signs $\pm_1,\pm_2,\pm_3$},
\end{equation} 
then 
$$\| \sum_{\bA} \mathfrak F_{\bA} \sin(\tfrac{v_\bA}{\lambda}) \|_{L^4(\Sigma_t)} \ls \Big( \sum_{\bA} \| \mathfrak F_{\bA}\|_{L^4(\Sigma_t)}^2\Big)^{\f 12} + \lambda^{\f 14} N \sup_{\bA} \| \mathfrak F_{\bA}\|_{L^4(\Sigma_t)}^{\f 34}\| \mathfrak F_{\bA}\|_{W^{1,4}(\Sigma_t)}^{\f 14} .$$
\end{enumerate}
\end{lemma}
\begin{proof}
We only give the proof of the $L^4$ bound which is harder. We expand the following expression and whenever one index is different from all the rest, we treat the term as an error, i.e.,
\begin{equation}
\begin{split}
&\:\| \sum_{\bA} \mathfrak F_{\bA} \sin(\tfrac{v_\bA}{\lambda}) \|_{L^4(\Sigma_t)}^4 \\
=&\: \int_{\Sigma_t} \sum_{\bA} \mathfrak F_{\bA}^4 \cos^4 (\tfrac{v_\bA}{\lambda}) \, \ud x + 6 \int_{\Sigma_t} \sum_{\bA,\,\bB: \bA \neq \bB} \mathfrak F_{\bA}^2 \mathfrak F_{\bB}^2 \sin^2 (\tfrac{v_\bA}{\lambda})\sin^2 (\tfrac{v_\bB}{\lambda}) \, \ud x \\
&\: + O\Big(\sum_{\pm_1,\pm_2,\pm_3,\bA,\bB,\bC,\bD: \bA \neq \bB, \bC,\bD} \Big| \int_{\Sigma_t} \mathfrak F_{\bA}\mathfrak F_{\bB}\mathfrak F_{\bC}\mathfrak F_{\bD}\cos(\tfrac{v_{\bA} \pm_1 v_\bB \pm_2 v_\bC \pm_3 v_\bD}{\lambda})\, \ud x \Big| \Big).
\end{split}
\end{equation}
The first two terms are clearly $\ls \Big( \sum_{\bA} \| \mathfrak F_{\bA}\|_{L^4(\Sigma_t)}^2\Big)^2$. For the last term, we use \eqref{eq:use.spatial.separation.for.orthogonality} to write 
\begin{equation}
\begin{split}
&\: \cos(\tfrac{v_{\bA} \pm_1 v_\bB \pm_2 v_\bC \pm_3 v_\bD}{\lambda}) \\
=&\: -\f{\lambda^2 \Delta(\cos (\tfrac{v_{\bA} \pm_1 v_\bB \pm_2 v_\bC \pm_3 v_\bD}{\lambda}))}{|\nabla (v_{\bA} \pm_1 v_\bB \pm_2 v_\bC \pm_3 v_\bD)|^2} - \f{\lambda (\Delta(v_{\bA} \pm_1 v_\bB \pm_2 v_\bC \pm_3 v_\bD)) \sin(\tfrac{v_{\bA} \pm_1 v_\bB \pm_2 v_\bC \pm_3 v_\bD}{\lambda})}{|\nabla(v_{\bA} \pm_1 v_\bB \pm_2 v_\bC \pm_3 v_\bD)|^2}
\end{split}
\end{equation}
Integrating by parts one of the derivatives away from the phase and then using H\"older's inequality, we obtain the desired result. \qedhere
\end{proof}

\begin{lm}\label{lmini}Under the assumptions of Theorem \ref{main.thm.2}, there exists $\lambda_0>0$ sufficiently small such that for every $\lambda\in (0,\lambda_0]$, there exist initial functions $U_{{\lambda}} $ and $\dot{U}_{{\lambda}}$ (for \eqref{sys}) which are compactly supported in $B(0,2R)$, such that the following holds:
	\begin{enumerate}
		\item 
		$U_\lambda$ and $\dot{U}_\lambda$ obey the following bounds:
		$$\left\|U_{{\lambda}} -U_0- \sum_\bA  \lambda a_\bA^{-1} F^U_\bA \cos(\tfrac{a_\bA u_\bA^0}{\lambda})   -\sum_{\pm, \bA,\bB:\bA\neq \bB} (\pm 1)\lambda^2F_{\bA,\bB}^U \cos(\tfrac{a_\bA u_\bA^0 \pm a_\bB u_\bB^0}{\lambda}) \right\|_{H^4(\Sigma_0)}\leq \ep\lambda^2C(N),$$
		and
\begin{align*}&\Big \| \dot{U}_{{\lambda}} -\dot{U}_0- \sum_\bA  e^{-\gamma_0} F^U_\bA|\nabla u_\bA^0| \sin(\tfrac{a_\bA u_\bA^0}{\lambda})  \\
&\: \qquad \qquad  -\sum_{\pm, \bA,\bB:\bA\neq \bB} (\pm 1) \lambda e^{-2\gamma_0}(a_\bA| \nabla u_\bA^0| \pm a_\bB |\nabla u_\bB^0| )F_{\bA,\bB}^U \sin(\tfrac{a_\bA u_\bA^0 \pm a_\bB u_\bB^0}{\lambda})  \Big\|_{H^3(\Sigma_0)}  \leq \ep\lambda^2C(N),
	\end{align*}
	where $F^U_\bA$ is as in the background solution and $F^U_{\bA,\bB}$ is as in \eqref{fabphi}.
		\item The following condition holds initially (see \eqref{eq:lara.def}, \eqref{main.data.cond}):
\begin{equation}\label{qty}\int_{\Sigma_0} \langle \dot{U}_{{\lambda}}, \partial_j U _{{\lambda}} \rangle \, \ud x=0.
	\end{equation}
	\end{enumerate}
\end{lm}
\begin{proof}
	The proof is similar in logic to the proof of \cite[Lemma 4.9]{HL.HF} but is slightly different as we use Lemma~\ref{lem:find.r1.r2}. Using the genericity assumption that $\phi_0$ is not identically zero, and using Lemma~\ref{lem:find.r1.r2} (and rescaling), we find $\wht r_{1}$, $\wht r_{2}$ such that $\mathrm{supp}(\wht r_{1}),\,\mathrm{supp}(\wht r_{2}) \subset B(0, 2R)$, $\|\wht r_1\|_{H^{10}(\Sigma_0)} = \|\wht r_2\|_{H^{10}(\Sigma_0)} = \ep$ and
	\begin{equation}\label{eq:def.tr1.tr2}
\det \Bigg[ \begin{array}{ll} \int_{\mathbb R^2} 2 \wht r_{1} \rd_1 \phi_0  \,\ud x  & \int_{\mathbb R^2} 2 \wht r_{2} \rd_1 \phi_0 \, \ud x \\
 \int_{\mathbb R^2}2 \wht r_{1} \rd_2 \phi_0 \, \ud x &  \int_{\mathbb R^2} 2 \wht r_{2} \rd_2 \phi_0 \, \ud x \end{array}
\Bigg]   \neq 0.
\end{equation}
	Notice that $\wht r_{1,\lambda}$, $\wht r_{2,\lambda}$ and the value in \eqref{eq:def.tr1.tr2} are all independent of $\lambda$.
	
We will look for $U_\lambda, \dot{U}_\lambda$ of the form
	\begin{align}
	U_\lambda=&\: U_0+ \sum_\bA  \lambda a_\bA^{-1} F^U_\bA \cos(\tfrac{a_\bA u_\bA^0}{\lambda})   +\sum_{\pm,\bA,\bB:\bA\neq \bB} (\pm 1)\lambda^2F_{\bA,\bB}^U \cos(\tfrac{a_\bA u_\bA^0 \pm a_\bB u_\bB^0}{\lambda}) ,\label{eq:U.in.osc.constraint} \\
\dot{U}_\lambda=&\: \dot{U}_0 + \sum_\bA  e^{-\gamma_0} F^U_\bA|\nabla u_\bA^0|  \sin(\tfrac{a_\bA u_\bA^0}{\lambda}) \nonumber \\
&\:  + \sum_{\pm,\bA,\bB:\bA\neq \bB} (\pm 1)\lambda e^{-2\gamma_0}(a_\bA| \nabla u_\bA^0| \pm a_\bB |\nabla u_\bB^0| )F_{\bA,\bB}^U \sin(\tfrac{a_\bA u_\bA^0 \pm a_\bB u_\bB^0}{\lambda}) + \wht \Omg_{1,\lambda}^U \wht r_{1} + \wht \Omg_{2,\lambda}^U \wht r_{2},\label{eq:dotU.in.osc.constraint}
	\end{align}
	where $\wht \Omg_i^U = (\wht \Omg_{i,\lambda}^\phi, \wht \Omg_{i,\lambda}^\varpi)$ for $i = 1,2$, with $(\wht \Omg_{1,\lambda}^\phi, \wht \Omg_{2,\lambda}^\phi)$ being the unknowns to be solved and $(\wht \Omg_{1,\lambda}^\varpi, \wht \Omg_{2,\lambda}^\varpi) = (0,0)$.

	We now impose \eqref{qty}. In a similar manner as Lemma \ref{lem:constraint.lower.bound}, this becomes an equation for $(\wht \Omg_{1,\lambda}^\phi, \wht \Omg_{2,\lambda}^\phi)$. We argue that all the terms on the right-hand side can be estimated by $\ep^2 C(N)\lambda^2$ :
	\begin{itemize}
		\item The $O(\ep^2 C(N))$ term without a high-frequency oscillatory phase cancels exactly due to \eqref{eq:ortho.main.thm}.
		\item The $O(\ep^2C(N))$ terms either cancel because of \eqref{main.data.cond} or are high frequency oscillatory terms, which become $O(\ep^2 C(N)\lambda^2)$ after integration by part twice.
		\item The remaining $O(\ep^2 \lambda C(N))$ terms have high-frequency oscillatory phases; they become $O(\ep^2 C(N) \lambda^2)$ after integration by part.
	\end{itemize}
	We refer the reader to \cite[Lemma 4.9]{HL.HF} for a similar argument. The upshot is that \eqref{qty} can be solved and $(\wht \Omg_{1,\lambda}^\phi, \wht \Omg_{2,\lambda}^\phi)$ obeys the bound $|\wht \Omg_{1,\lambda}^\phi|,\,|\wht \Omg_{2,\lambda}^\phi|\ls \ep^2 C(N) \lambda^2$. Plugging this estimate back into \eqref{eq:U.in.osc.constraint} and \eqref{eq:dotU.in.osc.constraint}, we obtain the desired estimates in part (1). \qedhere

\end{proof}
We note that the estimates we obtained in Lemma~\ref{lmini} are compatible with the parametrix for $U$ at $t = 0$; see \eqref{eq:wave.para.2}. In particular, recall that when solving the transport equations \eqref{eq:F1.F2.transport}, we have set the initial data for $F_\bA^{(1,U)}$ and $F_\bA^{(2,U)}$ to be zero. 

We now state the following important lemma, which allows us to apply the local existence result of Touati \cite{Touati.local} (Theorem~\ref{thm:Touati}) to obtain the existence of a solution $(U,g)$.
	
\begin{lemma}\label{smallnessl4}
	The initial data satisfies
	$$\| \partial U \|_{L^4(\Sigma_0)} \lesssim \ep.$$
\end{lemma}

\begin{proof}
Given that $\{a_\bA u_\bA\}$ are spatially adapted, the main term is of size $\ep^2$ in $L^4$ by Lemma~\ref{lem:orthogonality}. \qedhere
\end{proof}

We now apply the result Theorem~\ref{thm:Touati} of Touati to obtain a local solution. \textbf{From now on, $(\phi,\varpi,g)$ is the solution to Einstein vacuum equations in elliptic gauge, issuing from the free initial data of Lemma \ref{lmini}.} At this point, the time of existence given by Theorem~\ref{thm:Touati} could be very small (and dependent on $\lambda$ and $N$); we will prove a priori estimates showing that in fact the solution exists up to time $t = 1$.

Since we need also to solve for $u_\bA$ and $\chi_\bA$ using the eikonal equation and the Raychaudhuri equation, we need also to prescribe initial data for them; see Section~\ref{sec:data.u}. Note that the construction of the solution $(\phi,\varpi,g)$ will not require prescribing the data for $u_\bA$ and $\chi_\bA$, but they will be important in justifying the parametrix construction.


\subsection{Initial data for $\mfg$}

The initial data for $\mfg$ need not be prescribed since they satisfy elliptic equations. We record the estimates we have here. Since we have not prescribed $u_\bA$ yet, we describe a decomposition of $\mfg$ where the phases are given in terms of $u_\bA^0$ instead of $u_\bA$. This does not change any of the estimates since we will later prescribe $u_\bA$ so that $\sum_{k\leq 3} \|\rd (u_\bA^0- u_\bA)\|_{W^{k,4}(K_0)} \leq C(N) \ep \lambda^2$ (see Lemma~\ref{lem:chi.data} below), consistent with Section~\ref{sec.para.u} and that $\mathfrak v_\bA \restriction_{\Sigma_0} \equiv 0$. (Note that this is only possible for the initial data, as $\sum_{k\leq 3} \|\rd (u_\bA^0- u_\bA)\|_{W^{k,4}(K_t)}$ would only be $O(\lambda)$ for $t \in (0,1]$ due to the non-trivial $u_\bA^2$ terms.) The proof of the estimates is essentially the same as the estimates in Section~\ref{sec:metric} (which apply for any time $t$) and is omitted.

\begin{proposition}\label{prop:g.data}
For $\mfg \in \{ \gamma,\bt^i,n\}$, $\mfg$ admits a decomposition
$$\mfg = \mfg_0 + \underline{\mfg}_2 + \widetilde{\underline{\mfg}},$$
where $\underline{\mfg}_2$ is as in \eqref{eq:mfg2.form}, \eqref{eq:g2.def.1}--\eqref{eq:g2.def.3} except for using $u^0_\bA$ instead of $u_\bA$ in the phases. Then $\widetilde{\underline{\mfg}} = \widetilde{\underline{\mfg}}_a(t) \zeta(|x|) \log(|x|) + \widetilde{\underline{\mfg}}_r$, where $\widetilde{\underline{\mfg}}_a(t)$ are functions of $t$ alone and the following bounds hold:
	\begin{align}
|\widetilde{\underline{\mfg}}_{a}|(t)+\|\widetilde{\underline{\mfg}}_r\|_{W^{1,\f 43}_{-\alp_0 - \f 12}(\Sigma_t)} + \lambda|\rd_t \wht \mfg_{a}|(t)+\lambda\|\rd_t\widetilde{\underline{\mfg}}_r\|_{W^{1,\f 43}_{-\alp_0 - \f 12}(\Sigma_t)} \leq &\: C(N)\ep^2 \lambda^2 , \label{eq:g.L43.data}\\
\sum_{k\leq 3} \lambda^k \|\widetilde{\underline{\mfg}}_r\|_{W^{k+1,4}_{-\alp_0 + \f 12}(\Sigma_t)} + \sum_{k\leq 2} \lambda^{k+1} \|\rd_t\widetilde{\underline{\mfg}}_r\|_{W^{k+1,4}_{-\alp_0 + \f 12}(\Sigma_t)} \leq &\:  C(N)\ep^2 \lambda^2, \label{eq:g.L4.data}\\
	\sum_{k\leq 3} \lambda^k \|\widetilde{\underline{\mfg}}_r\|_{H^{k+2}_{-\alp_0}(\Sigma_t)} + \sum_{k\leq 2} \lambda^{k+1}\|\rd_t\widetilde{\underline{\mfg}}_r\|_{H^{k+2}_{-\alp_0}(\Sigma_t)}\leq &\: C(N)\ep^2 \lambda.  \label{eq:g.L2.data}
\end{align}
(Recall here that $\alp_0 = 10^{-10}$; see \eqref{eq:alp0}.)
\end{proposition}

\subsection{Initial data for $u_\bA$ and $\chi_\bA$}\label{sec:data.u} 

The initial data for $u_\bA$ and $\chi_\bA$ are to be prescribe simultaneously since $\chi_\bA$ is completely determined by $\chi_\bA = \Box_g u_\bA$ and the eikonal equation for $u_\bA$.

On $\Sigma_0$, we prescribe $u_\bA=u_\bA^0 + \wht u_{\bA}$ (i.e., $u_\bA^2\restriction_{\Sigma_0} = 0$ and $\mathfrak v_\bA\restriction_{\Sigma_0} = 0$ for $u_\bA^2$, $\mathfrak v_\bA$ in \eqref{eq:u.para}, \eqref{u1.def}), where $u_\bA^0 = x^1 \cos \om_\bA + x^2 \sin \om_\bA$, and 
\begin{equation}\label{eq:initial.utilde}
\wht u_{\bA} = \lambda^3 \Bigg( \sum_{\bB: \bB \neq \bA} \Big( \mathfrak v^{\mathrm{data},(\bA)}_{1,\bB} \cos (\tfrac{a_\bB u_\bB^0}\lambda) + \mathfrak v^{\mathrm{data},(\bA)}_{2,\bB} \sin (\tfrac{2 a_\bB u_\bB^0}\lambda) \Big) + \sum_{\pm,\bB, \bC} \mathfrak v^{\mathrm{data},(\bA)}_{\bB,\bC} \sin(\tfrac{a_\bB u_\bB^0 \pm a_\bC u_\bC^0}\lambda) \Bigg),
\end{equation}
where $\mathfrak v^{\mathrm{data},(\bA)}_{1,\bB}$, $\mathfrak v^{\mathrm{data},(\bA)}_{2,\bB}$ and $\mathfrak v^{\mathrm{data},(\bA)}_{\bB,\bC}$ are to be chosen below (see Lemma~\ref{lem:chi.data}).

 The initial data for $\chi_\bA$ are completely determined by the expression $\chi_\bA=\Box_g u_\bA$ and the choice of the initial data for $u_\bA$ through the following expression (see \cite[Lemma 7.5]{HL.elliptic}):
	\begin{equation}\label{Box.u.data}
	\begin{split}
		\Box_g u_{\bA} \restriction_{\Sigma_0}=& \frac{1}{n}e^{-\gamma} (e_0\gamma) |\nabla u_{\bA}|\restriction_{\Sigma_0}+\frac{1}{n e^{2\gamma}} \delta^{ij}\rd_i(n \rd_j u_{\bA})\restriction_{\Sigma_0}\\
		&-\frac{1}{n}e^{-\gamma}\left(\frac{1}{|\nabla u_{\bA}|}\delta^{ij}\rd_i u_{\bA} \rd_j (e^{-\gamma} n |\nabla u_{\bA}|) + \frac{1}{|\nabla u_{\bA}|}\delta^{ij}(\rd_i u_\bA)(\rd_j \beta^k)\partial_k u_{\bA}\right)\restriction_{\Sigma_0}.
	\end{split}
\end{equation}

The following lemma shows that $\mathfrak v^{\mathrm{data},(\bA)}_{1,\bB}$, $ \mathfrak v^{\mathrm{data},(\bA)}_{2,\bB}$ and $ \mathfrak v^{\mathrm{data},(\bA)}_{1,\bB}$ can be chosen so that the initial data for $u_\bA$ and $\chi_\bA$ obey the desired estimates that are to be propagated in evolution.
\begin{lemma}\label{lem:chi.data}
There exists a suitable choice of $\mathfrak v^{\mathrm{data},(\bA)}_{1,\bB}$, $ \mathfrak v^{\mathrm{data},(\bA)}_{2,\bB}$ and $ \mathfrak v^{\mathrm{data},(\bA)}_{1,\bB}$ in \eqref{eq:initial.utilde} such that the following all hold:
\begin{enumerate}
\item  $ \mathfrak v^{\mathrm{data},(\bA)}_{1,\bB}$, $ \mathfrak v^{\mathrm{data},(\bA)}_{2,\bB}$ and $ \mathfrak v^{\mathrm{data},(\bA)}_{1,\bB}$ satisfies the bound 
\begin{equation}\label{eq:frkv.def.est}
\|\mathfrak v^{\mathrm{data}} \|_{C^8(K_0)} \leq C(N) \ep^{2}.
\end{equation}
In particular,
\begin{equation}
\sum_{k\leq 3}\lambda^k \|\rd \wht u_\bA \|_{W^{k,4}(K_0)} \leq C(N) \ep^{2} \lambda^2. 
\end{equation}
\item $\chi_\bA\restriction_{\Sigma_0}$ admits the decomposition $\chi_{\bA}\restriction_{\Sigma_0} = \chi_{\bA}^0 \restriction_{\Sigma_0} + \chi_{\bA}^1 \restriction_{\Sigma_0} + \wht \chi_{\bA} \restriction_{\Sigma_0}$ with the following properties:
\begin{enumerate}
\item $\chi_{\bA}^0 \restriction_{\Sigma_0}$ is given by the background value.
\item $\chi_\bA^1\restriction_{\Sigma_0}$ is given by \eqref{chi1.def}, where $\mathfrak X^{(\bA)}_{2,\bB}$ ($\bA\neq \bB$), $\mathfrak X^{(\bA)}_{1,\bB}$ ($\bA\neq \bB$) and $\mathfrak X^{(\bA)}_{\pm,\bB,\bC}$ are given by \eqref{eq:chi1.def.1.1}--\eqref{eq:chi1.def.1.3}, and $\mathcal H$ are given in Lemma \ref{lm:rll}.
\item $\mathfrak X^{(\bA)}_{2,\bA}\restriction_{\Sigma_0}$, $\mathfrak X^{(\bA)}_{1,\bA}\restriction_{\Sigma_0}$ are prescribed to obey the bounds 
\begin{equation}
\| \mathfrak X^{(\bA)}_{2,\bA}\|_{C^{8}(K_0)} + \| \mathfrak X^{(\bA)}_{1,\bA}\|_{C^{8}(K_0)} \leq C(N) \ep.
\end{equation}
\item $\wht \chi_\bA \restriction_{\Sigma_0}$ satisfies the estimate
\begin{equation}
\sum_{k\leq 3} \lambda^k \| \wht \chi_\bA \|_{H^k(\Sigma_0)} \leq C(N) \ep \lambda^2.
\end{equation}
\end{enumerate}
\end{enumerate}
\end{lemma}
\begin{proof}
\pfstep{Step~1: Preliminaries} We will assume for now that the estimates \eqref{eq:frkv.def.est} are satisfied. We will verify this when the functions $\mathfrak v^{\mathrm{data},(\bA)}_{1,\bB}$, $ \mathfrak v^{\mathrm{data},(\bA)}_{2,\bB}$ and $ \mathfrak v^{\mathrm{data},(\bA)}_{1,\bB}$ are chosen.

We first write 
\begin{equation}
\chi_\bA\restriction_{\Sigma_0} = \chi_\bA^0 + \underline{\chi}_\bA^1 + \wht{\underline{\chi}}_\bA,
\end{equation}
where $\chi_\bA^0$ is as before, but
\begin{equation}
	\underline{\chi}_\bA^1 = \lambda \Big(\sum_{\bB } \mathfrak X^{(\bA)}_{2,\bB} \sin ( \tfrac{2a_\bB u^0_\bB}{\lambda})
	+\sum_{\bB } \mathfrak X^{(\bA)}_{1,\bB} \cos ( \tfrac{a_\bB u^0_\bB}{\lambda})
	+ \sum_{\substack{\pm,\bB,\bC: \\\bB \neq \bA, \bC \neq \bA,\bB}}\mathfrak X^{(\bA)}_{\pm,\bB,\bC} \sin ( \tfrac{a_\bB u^0_\bB\pm a_\bC u^0_\bC}{\lambda})\Big).
\end{equation}
When compared to \eqref{chi1.def}, notice that the amplitudes $\mathfrak X^{(\bA)}_{2,\bB}$, etc.~are the same as in \eqref{chi1.def} and it is only that $u_\bB$ etc.~are replaced by $u_\bB^0$ etc.~in the phases. Observe that as long as \eqref{eq:frkv.def.est} are satisfied, it suffices to control the $\mathfrak X^{(\bA)}$ quantities and $\wht{\underline{\chi}}_\bA$, since 
\begin{equation}
\sum_{k\leq 3} \lambda^k \Big\| \sin(\tfrac{2 a_\bA u_\bA}\lambda) - \sin(\tfrac{2 a_\bA u_\bA^0}\lambda) \Big\|_{H^k(\Sigma_0)} \leq C(N) \ep^2\lambda^2,
\end{equation}
(and similarly for other phases), which can be derived using \eqref{eq:initial.utilde}, \eqref{eq:frkv.def.est} and the pointwise bound 
\begin{equation}\label{eq:phase.change}
\Big| \sin(\tfrac{2a_\bA u_\bA}\lambda) - \sin(\tfrac{2 a_\bA u_\bA^0}\lambda) \Big| \ls |a_\bA| \f{|u_\bA - u_{\bA}^0|}{\lambda} \leq C(N) \ep^2\lambda^2.
\end{equation}

\pfstep{Step~2: Prescribing $\mathfrak v^{\mathrm{data},(\bA)}_{1,\bB}$, $ \mathfrak v^{\mathrm{data},(\bA)}_{2,\bB}$, $ \mathfrak v^{\mathrm{data},(\bA)}_{1,\bB}$, $\mathfrak X^{(\bA)}_{2,\bA}\restriction_{\Sigma_0}$, $\mathfrak X^{(\bA)}_{1,\bA}\restriction_{\Sigma_0}$} 
We now consider $\Box_g u_{\bA} - \Box_{g_0} u_{\bA}^0$, which can be written in two ways. On the one hand, $\Box_g u_{\bA} - \Box_{g_0} u_{\bA}^0 = \underline{\chi}_{\bA}^1 + \wht{\underline{\chi}}_\bA$. On the other hand, it can be expressed by expanding the right-hand side of \eqref{Box.u.data} with $u_\bA = u_\bA^0 + \wht u_\bA$ and $\mfg = \mfg_0 + \underline{\mfg}_2 + \widetilde{\underline{\mfg}}$ (for each metric component $\mfg$). We isolate two type of terms that arise in this process:
\begin{enumerate}
\item Two derivatives hit on $\wht u_\bA$ while all other terms come from background contribution:
\begin{equation}\label{eq:chi.data.2nd.u}
\begin{split}
&\: \frac{1}{e^{2\gamma_0}} \delta^{ij} \rd^2_{ij} \wht u_{\bA} - \frac{1}{e^{-2\gamma_0}} \delta^{ij} \de^{k\ell} \rd_i u^0_{\bA} \rd_\ell u^0_{\bA} \rd^2_{jk} \wht u_{\bA} \\
= &\: \frac{1}{e^{2\gamma_0}} \Big( \Delta \wht u_{\bA} -  (\cos \om_\bA \rd_{1} + \sin \om_\bA \rd_{2})^2 \wht u_{\bA} \Big),
\end{split}
\end{equation}
where as before, $\Delta$ denotes the Euclidean Laplacian.
\item The $u_\bA$ contribution comes from background, while the $\mfg$ contribution is linear in $\rd_x \underline{\mfg}_2$. (Note that by Lemma~\ref{lem:e0gamma}, $e_0\gamma$ can be written as spatial derivative of $\bt$.) Using the oscillations allowed here, we see that these terms take the form
\begin{equation}
\lambda \Big( \sum_{\bB} \mathcal Y^{(\bA)}_{1,\bB} \cos (\tfrac{a_\bB u_\bB^0}\lambda) + \sum_{\bB} \mathcal Y^{(\bA)}_{2,\bB} \sin (\tfrac{2 a_\bB u_\bB^0}\lambda) + \sum_{\pm,\bB,\bC: \bB \neq \bC} \mathcal Y^{(\bA)}_{\pm,\bB,\bC} \sin(\tfrac{a_\bB u_\bB^0 \pm a_\bC u_\bC^0}\lambda)\Big). 
\end{equation}
By Proposition~\ref{prop:g.data}, we have
\begin{equation}
\|\mathcal Y \|_{C^8(K_0)} \leq C(N) \ep^{2}.
\end{equation}
\end{enumerate}
We only isolate these terms because by Lemma~\ref{lmini}, Proposition~\ref{prop:g.data} and the assumed bound \eqref{eq:frkv.def.est}, all the other terms coming from $\hbox{RHS of \eqref{Box.u.data}} - \hbox{(RHS of \eqref{Box.u.data})}_0$ satisfy 
$\sum_{k\leq 3} \lambda^k\|\cdots \|_{H^k(\Sigma_0)} \leq \ep^2 \lambda^2 C(N)$.

We further expand the terms in \eqref{eq:chi.data.2nd.u} using \eqref{eq:initial.utilde}. For this it is helpful to note $\Delta = (\cos \om_\bA \rd_{1} + \sin \om_\bA \rd_{2})^2 + (\sin \om_\bA \rd_{1} - \cos \om_\bA \rd_{2})^2$. First,
\begin{equation}
\begin{split}
&\: \lambda^3 \Big( \Delta (\mathfrak v^{\mathrm{data},(\bA)}_{1,\bB} \cos (\tfrac{a_\bB u_\bB^0}\lambda) ) -  (\cos \om_\bA \rd_{1} + \sin \om_\bA \rd_{2})^2 (\mathfrak v^{\mathrm{data},(\bA)}_{1,\bB} \cos (\tfrac{a_\bB u_\bB^0}\lambda)) \Big) \\
= &\: - \lambda \mathfrak v^{\mathrm{data},(\bA)}_{1,\bB} \cos (\tfrac{a_\bB u_\bB^0}\lambda)   \Big( a_\bB^2 (\sin \om_\bA \cos \om_\bB - \cos \om_\bA \sin \om_\bB)^2)  \Big) + O(\lambda^2 \rd^{\leq 2} \mathfrak v^{\mathrm{data},(\bA)}_{1,\bB} ) \\
= &\: - \lambda \mathfrak v^{\mathrm{data},(\bA)}_{1,\bB} \cos (\tfrac{a_\bB u_\bB^0}\lambda)    a_\bB^2 \sin^2(\om_\bA-\om_\bB) + O(\lambda^2 \rd^{\leq 2}\mathfrak v^{\mathrm{data},(\bA)}_{1,\bB} ),
\end{split}
\end{equation}
where $O(\lambda^2 \rd^{\leq 2}\mathfrak v^{\mathrm{data},(\bA)}_{1,\bB} )$ denotes terms which are up to two derivatives of $\mathfrak v^{\mathrm{data},(\bA)}_{1,\bB}$ with a prefactor $\lambda^2$.
Similarly,
\begin{equation}
\begin{split}
&\: \lambda^3 \Big( \Delta (\mathfrak v^{\mathrm{data},(\bA)}_{2,\bB} \sin (\tfrac{2a_\bB u_\bB^0}\lambda) ) -  (\cos \om_\bA \rd_{1} + \sin \om_\bA \rd_{2})^2 (\mathfrak v^{\mathrm{data},(\bA)}_{2,\bB} \sin (\tfrac{2a_\bB u_\bB^0}\lambda)) \Big) \\
= &\: - 4 \lambda \mathfrak v^{\mathrm{data},(\bA)}_{2,\bB} \sin (\tfrac{2a_\bB u_\bB^0}\lambda)    a_\bB^2 \sin^2(\om_\bA-\om_\bB)  + O(\lambda^2  \rd^{\leq 2} \mathfrak v^{\mathrm{data},(\bA)}_{2,\bB}).
\end{split}
\end{equation}
Finally, 
\begin{equation}
\begin{split}
&\: \lambda^3 \Big( \Delta (\mathfrak v^{\mathrm{data},(\bA)}_{\bB,\bC} \sin (\tfrac{a_\bB u_\bB^0 \pm a_\bC u_\bC^0}\lambda) ) -  (\cos \om_\bA \rd_{1} + \sin \om_\bA \rd_{2})^2 (\mathfrak v^{\mathrm{data},(\bA)}_{\bB,\bC} \sin (\tfrac{a_\bB u_\bB^0 \pm a_\bC u_\bC^0}\lambda)) \Big) \\
= &\: - \lambda \mathfrak v^{\mathrm{data},(\bA)}_{\bB,\bC} \sin (\tfrac{a_\bB u_\bB^0 \pm a_\bC u_\bC^0}\lambda) \Big( a_\bB \sin(\om_\bA - \om_\bB)  \pm a_\bC \sin(\om_\bA- \om_\bC) \Big)^2 + O(\lambda^2 \rd^{\leq 2} \mathfrak v^{\mathrm{data},(\bA)}_{\bB,\bC} ).
\end{split}
\end{equation}

Therefore, the equation that we need to solve for is
\begin{equation}
\begin{split}
&\: \lambda \Big( \sum_{\bB } \mathfrak X^{(\bA)}_{2,\bB} \sin ( \tfrac{2a_\bB u_\bB}{\lambda})
+ \sum_{\bB } \mathfrak X^{(\bA)}_{1,\bB} \cos( \tfrac{a_\bB u_\bB}{\lambda})
+\sum_{\pm,\bB,\bC:\bB \neq \bC}\mathfrak X^{(\bA)}_{\pm,\bB,\bC} \sin ( \tfrac{a_\bB u_\bB\pm a_\bC u_\bC}{\lambda}) \Big) + \wht \chi_\bA \\
= &\: \lambda \Big( \sum_{\bB} \mathcal Y_{\bA,1,\bB} \cos (\tfrac{a_\bB u_\bB}\lambda) + \sum_{\bB} \mathcal Y_{\bA,2,\bB} \sin (\tfrac{2 a_\bB u_\bB}\lambda) + \sum_{\pm,\bB,\bC: \bB \neq \bC} \mathcal Y_{\pm,\bA,\bB,\bC} \sin(\tfrac{a_\bB u_\bB \pm a_\bC u_\bC}\lambda)\Big) \\
&\:  - \lambda \Big( \sum_{\bB: \bB \neq \bA} a_\bB^2 \sin^2(\om_\bA-\om_\bB) (\mathfrak v^{\mathrm{data},(\bA)}_{1,\bB} \cos (\tfrac{a_\bB u_\bB^0}\lambda)      +  4 \mathfrak v^{\mathrm{data},(\bA)}_{2,\bB} \sin (\tfrac{2a_\bB u_\bB^0}\lambda) ) \Big) \\
&\: - \lambda \sum_{\pm,\bB,\bC: \bB\neq \bC} \mathfrak v^{\mathrm{data},(\bA)}_{\bB,\bC} \sin (\tfrac{a_\bB u_\bB^0 \pm a_\bC u_\bC^0}\lambda) \Big( a_\bB \sin(\om_\bA - \om_\bB)  \pm a_\bC \sin(\om_\bA- \om_\bC) \Big)^2 + \cdots,
\end{split}
\end{equation}
where $\mathfrak X^{(\bA)}_{2,\bB}$ ($\bB \neq \bA$), $\mathfrak X^{(\bA)}_{1,\bB}$  ($\bB \neq \bA$), $\mathfrak X^{(\bA)}_{\pm,\bB,\bC}$, $\mathcal Y_{\bA,1,\bB}$, $\mathcal Y_{\bA,2,\bB}$ and $\mathcal Y_{\bA,\bB,\bC}$ ($\bB\neq \bC$) are given, while the initial value for $\mathfrak v^{\mathrm{data},(\bA)}_{1,\bB}$ ($\bB \neq \bA$), $\mathfrak v^{\mathrm{data},(\bA)}_{2,\bB}$ ($\bB \neq \bA$), $\mathfrak v^{\mathrm{data},(\bA)}_{\bB,\bC}$ ($\bB \neq \bC$), $\mathfrak X^{(\bA)}_{1,\bA}$, $\mathfrak X^{(\bA)}_{2,\bA}$, $\wht \chi_\bA$ are to be prescribed.

This naturally suggests the following choices of initial data for the terms making up $\wht u_\bA$:
\begin{align}
\mathfrak v^{\mathrm{data},(\bA)}_{1,\bB} = \f{\mathcal Y_{\bA,1,\bB} - \mathfrak X^{(\bA)}_{1,\bB}}{a_\bB^2 \sin^2(\om_\bA-\om_\bB)},\quad \mathfrak v^{\mathrm{data},(\bA)}_{2,\bB} = \f{\mathcal Y_{\bA,2,\bB} - \mathfrak X^{(\bA)}_{2,\bB}}{4 a_\bB^2 \sin^2(\om_\bA-\om_\bB)},\quad \bA\neq \bB, \label{eq:u.data.lower.bound.invert.1} \\
\mathfrak v^{\mathrm{data},(\bA)}_{\bB,\bC} = \f{\mathcal Y_{\bA,\bB,\bC} - \mathfrak X^{(\bA)}_{\pm,\bB,\bC}}{( a_\bB \sin(\om_\bA - \om_\bB)  \pm a_\bC \sin(\om_\bA- \om_\bC) )^2},\quad \bB \neq \bC, \label{eq:u.data.lower.bound.invert.2}
\end{align}
as well as the following choices of initial data for $\mathfrak X^{(\bA)}_{1,\bA}$ and $\mathfrak X^{(\bA)}_{2,\bA}$:
\begin{equation}
\mathfrak X^{(\bA)}_{1,\bA} = \mathcal Y_{\bA,1,\bA},\quad \mathfrak X^{(\bA)}_{2,\bA} = \mathcal Y_{\bA,2,\bA}.
\end{equation}
Finally, $\wht \chi_\bA$ can be prescribed to absorb all the remaining terms. It can be easily checked all the desired estimates hold, as long as we can lower-bound the coefficients that were divided in \eqref{eq:u.data.lower.bound.invert.1} and \eqref{eq:u.data.lower.bound.invert.2}. In other words, it remains to check the following claims:
\begin{align}
a_\bB^2 \sin^2(\om_\bA - \om_\bB) \gtrsim \f 1{N^2} \quad \hbox{whenever $\bA \neq \bB$}, \label{eq:chi.data.easy.lower.bound}\\
\Big( a_\bB \sin(\om_\bA - \om_\bB)  \pm a_\bC \sin(\om_\bA- \om_\bC) \Big)^2 \gtrsim \f 1{N^2} \quad \hbox{whenever $\bB \neq \bC$}.\label{eq:chi.data.funny.lower.bound}
\end{align} 
These can be established with the following observations:
\begin{itemize}
\item When $\bA \neq \bB$, we have $\Big| \sin(\om_\bA - \om_\bB) \Big| \gtrsim \f 1{N}$. This is because $|\om_\bA - \om_\bB|\geq \f{2\pi}{N}$ (recall $\om_\bA = \f{2\pi \bA}{N}$), and since $N$ is odd\footnote{We remark that this is the only place in the proof where we used that $N$ is odd.}, we also have $|\om_\bB - \om_\bC \pm \pi| \geq \f{\pi}{N}$. This proves \eqref{eq:chi.data.easy.lower.bound}.
\item We now consider \eqref{eq:chi.data.funny.lower.bound} when $\bA = \bB \neq \bC$. In this case, $\hbox{LHS of \eqref{eq:chi.data.funny.lower.bound}} = |a_\bC|^2 \sin^2(\om_\bB- \om_\bC)$, which can be bounded as in \eqref{eq:chi.data.easy.lower.bound}. The case $\bA = \bC$ is the same.
\item It remains to prove \eqref{eq:chi.data.funny.lower.bound} when $\bA \neq \bB \neq \bC \neq \bA$. By the argument for \eqref{eq:chi.data.easy.lower.bound}, we must have $\f 1{N^2} \ls |a_\bB|^2 \sin^2(\om_\bA - \om_\bB) \ls 1$ and $\f 1{N^2} \ls |a_\bC|^2 \sin^2(\om_\bA - \om_\bC) \ls 1$. Investigating again the choice of the constants $a_\bB$ and $a_\bC$ in \eqref{eq:choice.of.constants}, we see that the gap between the constants prohibits cancellations and thus \eqref{eq:chi.data.funny.lower.bound} holds. 
\end{itemize}
In particular, we can justify assuming \eqref{eq:frkv.def.est} from the beginning and this concludes the proof. \qedhere
\end{proof}

\section{Bootstrap assumptions and their consequences}\label{sec:bootstrap}

\subsection{The bootstrap assumptions}\label{sec:bootstrap.assumptions}

The local solution constructed in Section \ref{sec.id} exists a priori only on a time interval whose size depends on $\lambda$. To push the existence result on an interval of size $1$, we will use a continuity argument, based on bootstrap assumptions for $\wht \phi$, $\wht \varpi$, $\wht{g}$, $\wht u_\bA$ and $\wht \chi_\bA$ that we give in this subsection, and are assumed to hold true on a time interval $[0,T)$ for some $T \in (0,1)$.

In order to introduce our bootstrap assumptions, we need a hierarchy of constants:
\begin{equation}\label{eq;constant.hierarchy}
C(N)  \ll C_b(N) \ll C_b'(N) \ll A(N).
\end{equation}
All of these are allowed to depend on $N$. We use $C(N)$ to denote all constants that depend only on $N$ which appear in the proof. The other constants are chosen to be appropriately larger; $C_b(N)$ and $C_b'(N)$ are bootstrap constants with $C_b'(N)$ being used for time-differentiated quantities. The largest constant $A(N)$ is to be used as an exponentially growing factor.  We use the convention that \textbf{the constants $C(N)$ can increase from line to line, but $C_b(N)$, $C_b'(N)$ and $A(N)$ will always denote the same functions, which are to be chosen at the end of the proof}; see Section~\ref{sec:together}. We also continue to use the convention introduced at the end of Section~\ref{sec:main.theorem}, where \textbf{the $\ls$ notation are used when the implicit constants involved as \underline{independent} of $N$}.

We will moreover impose that $\ep_2$ is a universally small constant and that hierarchy of constants in $N$ will be chosen much larger that $\ep_2^{-1}$. In particular, we will also assume that
\begin{equation}\label{eq:N.vs.ep}
C(N) \leq C_b(N) \ep^3 \leq C_b'(N) \ep^9 \leq A(N) \ep^{12}.
\end{equation}

\subsubsection{Estimates for the wave map field.}
\begin{align}
&\sum_{k\leq 3} \lambda^k (\|\rd \wht \phi \|_{H^k(\Sigma_t)} + \|\rd \wht \varpi \|_{H^k(\Sigma_t)}) \leq 2 C_{b}(N)\lambda^2 \ep e^{A(N)t}. \label{BA:wave}
\end{align}
\subsubsection{Estimates for the metric.} For $\wht \mfg$, each of the metric component can be decomposed as 
$$\wht \mfg(t,x)=\wht \mfg_{a}(t)\zeta(|x|)\log({|x|})+ \wht {\mfg}_r(t,x),$$
where $\zeta$ is a cutoff function defined as Definition~\ref{def:cutoff}. Because of the local existence result (see Theorem~\ref{lwp} and Theorem~\ref{thm:Touati}), we know that $\gamma_a$ is independent of time and $\bt^j_a = 0$.

Recall that $\alp_0 = 10^{-10}$ (see \eqref{eq:alp0}). We introduce three different types of bootstrap assumptions. For up to the first derivatives of $\wht{\mfg}$ and $\rd_t \wht{\mfg}$, we assume that the following $L^{\f 43}$-based bootstrap assumptions hold:
\begin{align}
&|\wht \mfg_{a}|(t)+\|\wht{\mfg}_r\|_{W^{1,\f 43}_{-\alp_0 - \f 12}(\Sigma_t)}\leq 2 C_{b}(N)\ep^{\frac{3}{2}} \lambda^2e^{A(N)t} \label{BA:g.L43}\\
&|\rd_t \wht \mfg_{a}|(t)+\|\rd_t\wht{\mfg}_r\|_{W^{1,\f 43}_{-\alp_0 - \f 12}(\Sigma_t)}\leq 2C'_{b}(N)\ep^{\frac{3}{2}} \lambda e^{A(N)t}.\label{BA:dtg.L43}
\end{align}
For up to the fourth derivatives of $\wht{\mfg}$ and third derivatives of $\rd_t \wht{\mfg}$, we assume that the following $L^4$-based bootstrap assumptions hold:
\begin{align}
&\sum_{k\leq 3} \lambda^k \|\wht{\mfg}_r\|_{W^{k+1,4}_{-\alp_0 + \f 12}(\Sigma_t)}\leq 2 C_{b}(N)\ep^{\frac{3}{2}} \lambda^2e^{A(N)t} \label{BA:g.L4}\\
&\sum_{k\leq 2} \lambda^k \|\rd_t\wht{\mfg}_r\|_{W^{k+1,4}_{-\alp_0 + \f 12}(\Sigma_t)}\leq 2C'_{b}(N)\ep^{\frac{3}{2}} \lambda e^{A(N)t}.\label{BA:dtg.L4}
\end{align}
For up to the top derivatives, we assume that the following $L^2$-based bootstrap assumption holds:
\begin{align}
	&\sum_{k\leq 3} \lambda^k \|\wht{\mfg}_r\|_{H^{k+2}_{-\alp_0}(\Sigma_t)}\leq 2 C_{b}(N)\ep^{\frac{3}{2}} \lambda e^{A(N)t} \label{BA:g.L2}\\
	&\sum_{k\leq 2} \lambda^k\|\rd_t\wht{\mfg}_r\|_{H^{k+2}_{-\alp_0}(\Sigma_t)}\leq 2C'_{b}(N)\ep^{\frac{3}{2}} e^{A(N)t}. \label{BA:dtg.L2}
\end{align}

\subsubsection{Estimates for the eikonal functions $u_\bA$ and for $\chi_\bA$}
Since $u_\bA$ and $\chi_\bA$ are always multiplied by the amplitude functions $F$, which are supported in $[0,T]\times B(0,2R+2)$, it suffices to prove estimates for them on the set $K$ in Definition~\ref{def:K}.

We now introduce the bootstrap assumptions for $u_\bA$ and $\chi_\bA$:
\begin{align}
	\sum_{k \leq 3} \lambda^k\| \partial \widetilde{u}_\bA\|_{W^{k,4}(K_t)}\leq &\: 2C_{b}(N) \lambda^2 \ep e^{A(N)t}, \label{BA:u.L4}\\ 
	\| \partial \widetilde{u}_\bA\|_{H^{4}(K_t)}\leq &\: 2C_{b}(N) \lambda^{-2} \ep e^{A(N)t}, \label{BA:u.L2} \\
\sum_{k \leq 3}
\lambda^k\|\wht \chi_\bA\|_{H^{k}(K_t)}\leq &\: 2 C_b(N) \lambda^2 \ep e^{A(N)t}. \label{BA:chi}
\end{align}
Note that similarly to the estimates for the metric functions, the lower order bounds for $\wht u_\bA$ are $L^4$-based while the top order bound is $L^2$-based.

\subsection{Bootstrap assumptions are satisfied initially}

It is straightforward to check that the constructions in Section~\ref{sec.id} are exactly set up so that the bootstrap assumptions hold at time $t=0$. (In particular, using calculations such as \eqref{eq:phase.change}, we can pass between the oscillating phases in terms of $u_\bA$ and those in terms of $u_{\bA}^0$.) We summarize this as follows.
\begin{proposition}
For a suitable choice of $C_b(N)$ and $C_b'(N)$, all the bootstrap assumptions hold initially when $t = 0$.
\end{proposition}

\subsection{Immediate consequences of the bootstrap assumptions}

We have the following $L^\infty$ estimate:
\begin{lm}\label{lminfty}
The following $L^\infty$ estimates hold for all $t \in [0,T)$:
\begin{align*}
\sum_{k\leq 1} \lambda^k (\|\partial \widetilde{\phi} \|_{W^{k,\infty}(\Sigma_t)} + \|\partial \widetilde{\varpi} \|_{W^{k,\infty}(\Sigma_t)}) \lesssim &\: \ep \lambda^\frac{1}{2},\\
\sum_{k \leq 3} \lambda^k \| \wht \mfg_r \|_{W^{k,\infty}_{-\alp_0+1}(\Sigma_t)} + \sum_{k\leq 2} \lambda^{k+1} \|\partial_t \wht \mfg_r \|_{W^{k,\infty}_{-\alp_0+1}(\Sigma_t)} \lesssim &\:  \ep \lambda^{\frac{3}{2}},\\
\sum_{k\leq 2} \lambda^k \|\partial \wht u_\bA \|_{W^{k,\infty}(K_t)} \lesssim &\: \ep \lambda^{\frac{1}{2}},\\
\sum_{k\leq 1} \lambda^k \|\wht{\chi}_\bA\|_{W^{k,\infty}(K_t)} \lesssim &\: \ep \lambda^\frac{1}{2}.
\end{align*}
\end{lm}
\begin{proof} 
These are direct consequences of the bootstrap assumptions and Sobolev embedding. For $\wht U = (\widetilde{\phi}, \widetilde{\varpi})$,
	we use the estimate $\|h\|_{L^\infty(\mathbb R^2)} \leq \|f\|_{L^2(\mathbb R^2)}^\frac{1}{2}\|f\|_{H^2(\mathbb R^2)}^\frac{1}{2}$ to obtain using \eqref{BA:wave} that
$$	\|\rd \wht U \|_{L^\infty} \lesssim \ep \lambda C_b(N) e^{A(N)t} \lesssim \ep \lambda^\frac{1}{2},$$
where we have used $\lambda^\frac{1}{2}  C_b(N)e^{A(N)t} \lesssim 1$. $\wht u_\bA$ and $\wht \chi_\bA$ can treated similarly using \eqref{BA:u.L4}, \eqref{BA:u.L2}, \eqref{BA:chi}, after using a standard Sobolev extension theorem so that the embedding estimate applies on the set $K_t$.

Finally, for $\wht \mfg_r$, we use the following weighted Sobolev embedding: by Theorem~3.4.1.(a) in Appendix~I of \cite{CB.book}, we have
\begin{equation}\label{eq:weighted.Li.Sobolev}
\| h \|_{W^{k,\infty}_{-\alp_0+1}(\mathbb R^2)} \ls \| h \|_{W^{k+1,4}_{-\alp_0+\f 12}(\mathbb R^2)}.
\end{equation}
Thus the desired estimate follows from \eqref{BA:g.L4} and \eqref{BA:dtg.L4}. \qedhere
\end{proof}

We also observe that even though the bootstrap assumptions do not explicitly bound $\wht{U}$ itself (they only bound $\rd \wht{U}$), we can estimate $\wht{U}$ using that it is compactly supported:
\begin{lemma}\label{lem:wave.Poincare}
The following estimates hold for all $t \in [0,T)$:
    \begin{equation*}
        \begin{split}
            \| \wht \phi \|_{L^2(\Sigma_t)} + \| \wht \varpi \|_{L^2(\Sigma_t)} \ls &\: C_b(N) \ep \lambda^2 \ep e^{A(N)t}, \\ 
            \| \wht \phi \|_{L^\infty(\Sigma_t)} + \| \wht \varpi \|_{L^\infty(\Sigma_t)} \ls &\: \ep \lambda.
        \end{split}
    \end{equation*}
\end{lemma}
\begin{proof}
    Since $\wht U = (\widetilde{\phi}, \widetilde{\varpi})$ is supported in $B(0,R+2)$ (by part (5) of Proposition~\ref{lwp} and the control over the metric), we have
    $$\|\wht U \|_{L^2(\Sigma_t)} \ls \| \rd \wht U\|_{L^2(\Sigma_t)}$$
    by Poincar\'e's inequality. Thus the first inequality follows from \eqref{BA:wave}. The second inequality can then be proven in a similar manner as Lemma~\ref{lminfty}. \qedhere
\end{proof}

\section{Estimates for the wave map field}\label{sec:wave}

In this section, we prove the estimates for the error terms $\wht \phi$ and $\wht \varpi$. Since these estimates are very similar, we will mostly focus on the bounds for $\wht \phi$ and omit the details for $\wht \varpi$.

\subsection{Equation for the error term}

Expanding $\Box_g \phi$ using the expression of $\phi$ in \eqref{eq:wave.para}, \eqref{eq:wave.para.2}, and using the equations $\Box_{g} \phi + \f 12 e^{-4\phi} g^{-1}(\ud\varpi,\ud\varpi) = 0$, $\Box_{g_0} \phi_{0} + \f 12 e^{-4\phi_0} g_0^{-1}(\ud\varpi_0,\ud\varpi_0) = 0$, we obtain
\begin{equation}\label{eq:wave.first.equation}
\begin{split}
\Box_g \wht \phi =&\: - \sum_{\bA}\Box_{g} \Big(\lambda a_\bA^{-1} F^\phi_\bA \cos(\tfrac{a_\bA u_\bA}{\lambda}) \Big) - \sum_{\bA} \Box_g \Big(\lambda^2 a_{\bA}^{-1}F_\bA^{1,\phi} \sin(\tfrac{a_\bA u_\bA}{\lambda})\Big) \\
&\: - \sum_{\bA} \Box_g \Big(\lambda^2 a_{\bA}^{-1}F_\bA^{2,\phi} \cos(\tfrac{2a_\bA u_\bA}{\lambda})\Big) - \sum_{\pm,\bA,\bB:\bA\neq \bB}\Box_g \left(\lambda^2 F^\phi_{\bA,\bB} \cos (\tfrac{a_\bA u_\bA \pm a_\bB u_\bB}{\lambda }) \right) \\
&\: +  (\Box_{g_0} - \Box_{g}) \phi_0 +  \f 12 e^{-4\phi_0} g_0^{-1}(\ud\varpi_0,\ud\varpi_0) - \f 12 e^{-4\phi} g^{-1}(\ud\varpi,\ud\varpi).
\end{split}
\end{equation}

We will consider the terms on the right-hand side of \eqref{eq:wave.first.equation} in the next few lemmas. Our goal will be to identify the main terms up to errors that satisfy 
\begin{equation}\label{eq:error.general}
\sum_{k\leq 3}\lambda^k\|\cdots\|_{H^k(\Sigma_t)} \leq \ep \lambda^2 C(N)C'_{b}(N)e^{A(N)t}
\end{equation}
for all $t \in [0,T)$. See Proposition~\ref{prop:wave.error} for the equation we derive for $\wht \phi$. 

Before we proceed, we make some useful observations that will be used throughout.

\begin{remark}\label{rmk:background}(Background quantities)
We will refer to $g_0$, $\phi_0$, $\varpi_0$, $F^\phi_\bA$, $F^\varpi_\bA$ as ``background quantities.'' Any smooth functions of these quantities, which is independent of $\lambda$ is also referred to ``background quantities.''

With this convention, the terms $F^\phi_{\bA,\bB}$, $\mathcal G_{1,\bA}$, $\mathcal G_{2,\bA}$, $\mathcal G_{\pm,\bA,\bB}$, $\mathfrak v_\bA$, $\mathfrak X^{(\bA)}_{1,\bB}$ ($\bA\neq \bB$), $\mathfrak X^{(\bA)}_{2,\bB}$ and $\mathfrak X^{(\bA)}_{\pm,\bB,\bC}$ are all background quantities. (Note that $F^{1,U}_\bA$ and $F^{2,U}_\bA$ are not considered background quantities at this point and have to be tracked explicitly.)
\end{remark}
 
\begin{remark}\label{rmk:tilded}(Dealing with the tilded quantities)
In what follows, we will expand $\phi$, $\varpi$, $g^{-1}$, $u_{\bA}$ (and their first derivatives) and $\chi_{\bA}$ according to the parametrices in \eqref{eq:wave.para}, \eqref{eq:g.para}, \eqref{eq:u.para} and \eqref{eq:chi.para}. The following observations are useful for the tilded quantities:
\begin{enumerate}
\item For most terms that involve the tilded quantities linearly, we use the bootstrap assumption to show that are at worst $O(C(N)C_{b}'(N)\lambda^2\ep^2 e^{A(N)t})$ in the norm $\sum_{k\leq 3}\lambda^{-k}\|\cdot\|_{H^k(\Sigma_t)}$. 
\item There are only two exceptions for which some care is needed:
\begin{enumerate}
\item For the term $\Box_{g} \Big(\lambda a_\bA^{-1} F^\phi_\bA \cos(\tfrac{a_\bA u_\bA}{\lambda}) \Big)$, both $\rd_t \wht g$ and $\lambda \rd^2 \wht u_\bA$ may show up, and either term is one power of $\lambda$ short. Importantly, the actual terms combine to $\chi_\bA$ and we use that $\widetilde{\chi}_\bA$ obeys the better bound $\widetilde{\chi}_\bA = O(\lambda C_b(N)e^{A(N)t})$.
\item For the term $(\Box_{g_0} - \Box_{g}) \phi_0$, we get terms of the form $\rd_t \mfg \rd \phi_0$, which is again one power of $\lambda$ short. \textbf{This is the only type of tilded terms that we need to keep.} See Proposition~\ref{prop:wave.error}.
\end{enumerate}
\item Using in addition Lemma~\ref{lminfty} (and Lemma~\ref{lem:wave.Poincare}), we further see that for terms where the tilded quantities appear at least quadratically, the terms can be bounded are $O(C(N)(C'_{b}(N))^2 \lambda^{\f 52}\ep^3 e^{2A(N)t})$. Despite the bad dependence on $N$, the additional power in $\lambda^{\f 12}$ makes it an acceptable error.
\end{enumerate}
\end{remark}


%

\begin{lm}\label{lm.sf.1} For any $\bA$,
\begin{align*}
\Box_{g} \Big(\lambda a_\bA^{-1} F^\phi_\bA \cos(\tfrac{a_\bA u_\bA}{\lambda})\Big) 
=&-\Big(F^{\phi}_{\bA} \chi_\bA^0 + 2(g_0^{-1})^{\alpha \beta}\partial_\alpha u^0_{\bA}\partial_\beta F^\phi_{\bA}\Big)\sin(\tfrac{a_\bA u_{\bA}}{\lambda}) \\
& +\lambda a_{\bA}^{-1}\Box_{g_0} F_\bA^\phi \cos(\tfrac{a_\bA u_\bA}{\lambda}) +\lambda A^{\phi,\Box}_{2,\bA} \sin(\tfrac{2a_\bA u_\bA}{\lambda}) \\
&+\lambda\sum_{\pm, \bB: \bB \neq \bA} \Big( A^{\phi,\Box}_{\bA,2\bB} \cos(\tfrac{2a_\bB u_\bB \pm a_\bA u_\bA}{\lambda}) +
	A^{\phi,\Box}_{\bA,\bB} \sin(\tfrac{a_\bB u_\bB \pm a_\bA u_\bA}{\lambda}) \Big)\\
	&+\lambda\sum_{\substack{\pm_1,\pm_2, \bB,\bC:\\ \bB \neq \bA,\bC\neq \bB,\bA}}A^{\phi,\Box}_{\pm_1,\pm_2,\bA,\bB,\bC} \cos(\tfrac{a_\bA u_\bA \pm_1 a_\bB u_\bB \pm_2 a_\bC u_\bC}{\lambda})
	+\cdots,
\end{align*}
	where $A^{\phi,\Box}_{2,\bA} $, $A^{\phi,\Box}_{\bA,\bB}$, $A^{\phi,\Box}_{\bA,2\bB} $ and $A^{\phi,\Box}_{\pm_1,\pm_2,\bA,\bB,\bC}$ are background quantities supported in $[0,1]\times B(0,2R)$ and obey the following bound for all $t \in [0,1]$:
	$$\| A \|_{C^8(\Sigma_t)} + \| \rd_t A \|_{C^7(\Sigma_t)} \leq C(N) \ep,$$
	and the error term $\cdots$ satisfies \eqref{eq:error.general}.
\end{lm}
\begin{proof}
We first compute
\begin{equation}
	\begin{split}
&\Box_{g} (\lambda a_\bA^{-1} F^\phi_\bA \cos(\tfrac{a_\bA u_\bA}{\lambda}))\\
=& -\Big(F^{\phi}_{\bA} \chi_{\bA} + 2\gi^{\alpha \beta}\partial_\alpha u_{\bA}\partial_\beta F^\phi_{\bA}\Big)\sin(\tfrac{a_\bA u_{\bA}}{\lambda})+ \lambda  a_\bA^{-1} \Box_g F^\phi_{\bA} \cos(\tfrac{a_\bA u_\bA}{\lambda}),
\end{split}
\end{equation}
where we have used the fact that $u_{\bA}$ satisfy the eikonal equation $\gi^{\alp\bt} \rd_\alp u_\bA \rd_\bt u_\bA = 0$ and that $\Box_g u_\bA = \chi_\bA$.
Using the decomposition of $g$, $u_\bA $ and $\chi_\bA$, we write
\begin{equation}
	\begin{split}
&\Box_{g} (\lambda a_\bA^{-1} F^\phi_\bA \cos(\tfrac{a_\bA u_\bA}{\lambda}))\\
=& -\Big(F^{\phi}_{\bA} \chi^0_{\bA} + 2(g_0^{-1})^{\alpha \beta}\partial_\alpha u^0_{\bA}\partial_\beta F^\phi_{\bA}\Big)\sin(\tfrac{a_\bA u_{\bA}}{\lambda})\\ 
&+(-\chi_\bA^{1}F_\bA^\phi +2(g_0^{-1})^{\alpha \beta}\partial_\alpha u_\bA^2 \partial_\beta F_\bA^\phi)
\sin(\tfrac{a_\bA u_\bA}{\lambda})
+a_\bA^{-1}\lambda \Box_{g_{0}}F_\bA^\phi \cos(\tfrac{a_\bA u_\bA}{\lambda} )+\cdots,
\end{split}
\end{equation}
where the $\cdots$ are all acceptable terms according to Remark~\ref{rmk:tilded}.

Let us compute the $O(\lambda)$ term, which can be entirely expressed in term of background quantities. We have
\begin{align*}
	\chi_\bA^1  \sin( \tfrac{a_\bA u_\bA}{\lambda})
	=& \lambda \sum_{\bB:\bB \neq \bA} \mathfrak X^{(\bA)}_{2,\bB} \sin ( \tfrac{2a_\bB u_\bB}{\lambda}) \sin( \tfrac{a_\bA u_\bA}{\lambda})\\
	&+\lambda \sum_{\bB:\bB \neq \bA} \mathfrak X^{(\bA)}_{1,\bB} \cos ( \tfrac{a_\bB u_\bB}{\lambda}) \sin( \tfrac{a_\bA u_\bA}{\lambda})\\
	&+\lambda \sum_{\substack{\pm,\bB,\bC: \\ \bB\neq \bA, \bC \neq \bA,\bB}} \mathfrak X^{(\bA)}_{\pm,\bB,\bC} \sin ( \tfrac{a_\bB u_\bB\pm a_\bC u_\bC}{\lambda}) \sin( \tfrac{a_\bA u_\bA}{\lambda}),
\end{align*}
and we have
$$2(g_0^{-1})^{\alpha \beta}\partial_\alpha u_\bA^1 \partial_\beta F_\bA^\phi
\sin(\tfrac{a_\bA u_\bA}{\lambda})
=2\lambda\mathfrak v_\bA(g_0^{-1})^{\alpha \beta}\partial_\alpha u_\bA^0 \partial_\beta F_\bA^\phi \cos(\tfrac{a_\bA u_\bA}{\lambda})
\sin(\tfrac{a_\bA u_\bA}{\lambda})+\cdots$$
where $\cdots$ satisfies the desired estimates. This concludes the proof of Lemma \ref{lm.sf.1}.
\end{proof}

\begin{lm}\label{lm.sf.2}
	For any $\bA$,
\begin{align*}\Box_g \Big(\lambda^2 a_{\bA}^{-1}F_\bA^{1,\phi} \sin(\tfrac{a_\bA u_\bA}{\lambda})\Big)=&\: +\lambda\Big(F^{1,\phi}_{\bA} \Box_{g_0} u_{\bA}^0 + 2g_0^{-1}(\ud u^0_{\bA},\ud F^{1,\phi}_{\bA}) \Big)\cos(\tfrac{a_\bA u_{\bA}}{\lambda})+\cdots\\
	\Box_g \Big(\lambda^2 a_{\bA}^{-1}F_\bA^{2,\phi} \cos(\tfrac{2a_\bA u_\bA}{\lambda})\Big)=&\: -2\lambda\Big(F^{2,\phi}_{\bA} \Box_g u_{\bA}^0 + 2g_0^{-1}(\ud u^0_{\bA}, \ud F^{2,\phi}_{\bA})\Big)\sin(\tfrac{2 a_\bA u_{\bA}}{\lambda})+\cdots,
\end{align*}
where the error term $\cdots$ satisfies \eqref{eq:error.general}.
\end{lm}
\begin{proof}
These are similar to that in Lemma~\ref{lm.sf.1}, except that they are easier because we could discard all the $\lambda^2$ terms as error. We omit the details. \qedhere


\end{proof}

\begin{lm}\label{lem:wave.main.elliptic.terms}
For $\bA \neq \bB$ and $F^\phi_{\bA,\bB}$ prescribed as in \eqref{fabphi}, we have
\begin{align*}
	\Box_g \Big(\lambda^2 F^\phi_{\bA,\bB} \cos (\tfrac{a_\bA u_\bA \pm a_\bB u_\bB}{\lambda }) \Big)
	=&\: \mp \frac{e^{-4\phi_0} F_\bA^\varpi F_\bB^\varpi}{4} g^{-1}( \ud  u_\bA, \ud u_\bB ) \cos (\tfrac{a_\bA u_\bA \pm a_\bB u_\bB}{\lambda })\\
	&\:  +\lambda A_{\pm,\bA ,\bB}^{FAB} \sin(\tfrac{a_\bA u_\bA \pm a_\bB u_\bB}{\lambda }) + \cdots,
\end{align*}
where $A_{\pm,\bA,\bB}^{FAB}$ are background quantities supported in $[0,1]\times B(0,2R)$ and obey the following bound for all $t \in [0,1]$:
	$$\| A \|_{C^8(\Sigma_t)} + \| \rd_t A \|_{C^7(\Sigma_t)} \leq C(N) \ep,$$
	and the error term $\cdots$ satisfies \eqref{eq:error.general}.
\end{lm}
\begin{proof}
We compute
\begin{align*}
	&\:\Box_g \left((\pm 1) \lambda^2 F^\phi_{\bA,\bB} \cos (\tfrac{a_\bA u_\bA \pm a_\bB u_\bB}{\lambda }) \right)\\
	=&\:- F^\phi_{\bA,\bB}  \gi^{\alpha \beta}\partial_\alpha (a_\bA u_\bA \pm a_\bB u_\bB) \partial_\beta(a_\bA u_\bA \pm a_\bB u_\bB)\cos (\tfrac{a_\bA u_\bA \pm a_\bB u_\bB}{\lambda })\\
	&\:+\lambda A_{\pm,\bA,\bB}^{FAB} \sin(\tfrac{a_\bA u_\bA \pm a_\bB u_\bB}{\lambda }) +\cdots,
	\end{align*}
where 
$$A_{\pm,\bA,\bB}^{FAB} =\mp \left( 2(g_0^{-1})^{\alpha \beta}\partial_\alpha(a_\bA u^0_\bA \pm a_\bB u^0_\bB)\partial_\beta F^\phi_{\bA,\bB} 
+\Box_{g_0}(a_\bA u^0_\bA \pm a_\bB u^0_\bB)F^\phi_{\bA,\bB}\right).$$
Moreover, using the eikonal equation, it holds that
$$ \gi^{\alpha \beta}\partial_\alpha (a_\bA u_\bA \pm a_\bB u_\bB) \partial_\beta(a_\bA u_\bA \pm a_\bB u_\bB)= \pm 2 a_\bA a_\bB\gi^{\alpha \beta}\partial_\alpha u_\bA \partial_\beta u_\bB.$$
Recalling also that $F^\phi_{\bA,\bB} = \frac{e^{-4\phi_0} F_\bA^\varpi F_\bB^\varpi}{8 a_\bA a_\bB}$ (by \eqref{fabphi}), we obtain the desired conclusion. \qedhere
\end{proof}

\begin{lm}
	We have
\begin{align*}
	\Box_g \phi^{0} = &\: \Box_{g_0} \phi^0 + \frac{1}{n_0^3}\partial_t \wht n \partial_t \phi^0 + \frac{1}{n_0^2}\partial_t\wht \beta^i \partial_i \phi^0\\
&\: + \lambda \sum_{\bA} A_{2,\bA}^{\phi_0} \sin(\tfrac{2a_\bA u_\bA}{\lambda})
+\lambda\sum_{\bA}A_{1,\bA}^{\phi_0} \cos(\tfrac{a_\bA u_\bA}{\lambda})
+\lambda\sum_{\pm,\bA,\bB: \bB\neq \bA} A_{\pm,\bA,\bB}^{\phi_0} \sin(\tfrac{a_\bA u_\bA \pm a_\bB u_{\gra B}}{\lambda})+\cdots,
\end{align*}
where $ A_{2,\bA}^{\phi_0}$, $A_{1,\bA}^{\phi_0}$ and $A_{\pm,\bA,\bB}^{\phi_0}$ are background quantities supported in $[0,1]\times B(0,2R)$ and obey the following bound for all $t \in [0,1]$:
	$$\| A \|_{C^8(\Sigma_t)} + \| \rd_t A \|_{C^7(\Sigma_t)} \leq C(N) \ep,$$
	and the error term $\cdots$ satisfies \eqref{eq:error.general}.
\end{lm}
\begin{proof}
	In the expression of $\Box_g \phi^{0} -\Box_{g_0} \phi^0$ we isolate the term which are only $O(\lambda)$. They involved either a derivative of $g_2$ and can be expressed in term of background quantities, or are $\partial_t$ derivatives of $\wht n$ or $\wht \beta$ (since $\partial_t \wht \gamma$ can be expressed in term of space derivatives of other metric coefficients, and therefore has a better behaviour in $\lambda$, see Lemma~\ref{lem:e0gamma}). The detailed calculations can be found in Proposition 6.4 of \cite{HL.HF}.
\end{proof}

We now look at the quadratic nonlinearity.
\begin{lemma}\label{lem:wave.nonlinear}
	We have
	\begin{align*}
		&\frac{1}{2}e^{-4\phi}g^{-1}(\ud\varpi,\ud \varpi)=	\frac{1}{2}e^{-4\phi_0}g_0^{-1}(\ud\varpi_0,\ud \varpi_0)\\
	= &-	\sum_{\bA }e^{-4\phi_0} g_0^{-1}(\ud\varpi_0, \ud u_\bA^0)F_\bA^{\varpi}\sin(\tfrac{a_\bA u_\bA}{\lambda})+\sum_{\pm, \bA,\bB: \bA \neq \bB} \f{(\mp 1) e^{-4\phi_0} F_\bA^{\varpi}F_\bB^{\varpi}}{4} g^{-1}(\ud u_\bA, \ud u_\bB)\cos(\tfrac{a_\bA u_\bA\pm a_\bB u_\bB}{\lambda})\\
	&+	\lambda\sum_{\bA } \Big( e^{-4\phi_0} g_0^{-1}(\ud\varpi_0, \ud u_\bA^0)F_\bA^{1,\varpi} +A^{sl}_{1,\bA} \Big)	\cos(\tfrac{a_\bA u_\bA}{\lambda})\\
	&+\lambda\sum_{\bA } \Big(- 2e^{-4\phi_0} g_0^{-1}(\ud\varpi_0, \ud u_\bA^0)F_\bA^{2,\varpi} +A^{sl}_{2,\bA}\Big) \sin(\tfrac{2a_\bA u_\bA}{\lambda}) \\
	&+\lambda\sum_{\pm,\bA,\bB: \bA \neq \bB} \Big(A^{sl}_{\bA, \bB} \sin(\tfrac{a_\bA u_\bA \pm a_\bB u_\bB}{\lambda}) + A^{sl}_{2\bA, \bB} \cos(\tfrac{2a_\bA u_\bA \pm a_\bB u_\bB}{\lambda}) \Big)\\
&+ \lambda \sum_{\substack{\pm_1,\pm_2,\bA,\bB,\bC:\\\bB \neq \bA, \bC \neq \bA,\bB}} A^{sl}_{\pm_1,\pm_2, \bA, \bB, \bC} \cos(\tfrac{a_\bA u_\bA \pm_1 a_\bB u_\bB \pm_2 a_\bC u_\bC}{\lambda})+\cdots,
\end{align*}
where $A^{sl}_{1,\bA}$, $A^{sl}_{2,\bA}$, $A^{sl}_{\bA, \bB}$, $A^{sl}_{2\bA, \bB}$ and $A^{sl}_{\pm_1,\pm_2, \bA, \bB, \bC}$ are background quantities supported in $[0,1]\times B(0,2R)$ and obey the following bound for all $t \in [0,1]$:
$$\|A \|_{C^8(\Sigma_t)} + \|\rd_t A \|_{C^7(\Sigma_t)} \leq C(N) \ep^2,$$
and the error term $\cdots$ satisfies \eqref{eq:error.general}.
\end{lemma}
\begin{proof}

Using Remark~\ref{rmk:tilded}, all the tilded quantities can be absorbed into the error terms. We compute (using $\cdots$ to denote acceptable error terms)
\begin{align*}
	&\: \frac{1}{2}e^{-4\phi}g^{-1}(\ud\varpi,\ud \varpi)-	\frac{1}{2}e^{-4\phi_0}g_0^{-1}(\ud\varpi_0,\ud \varpi_0)\\
= &\: e^{-4\phi_0} \Big( g_0^{-1}(\ud \varpi_0, \ud \varpi_1) + \f 12 g^{-1}(\ud \varpi_1, \ud \varpi_1) +  g_0^{-1}(\ud \varpi_0, \ud \varpi_2) +  g^{-1}(\ud \varpi_1, \ud \varpi_2) \Big)\\
&\: + \f 12 (e^{-4\phi} - e^{-4\phi_0}) \Big(g_0^{-1}(\ud\varpi_0,\ud \varpi_0) + 2 g_0^{-1}(\ud\varpi_0,\ud \varpi_1) + g_0^{-1}(\ud\varpi_1,\ud \varpi_1)\Big) + \cdots
\end{align*}

We only need to keep track of up to $O(\lambda)$ terms. Hence, we write
\begin{align*}
\ud \varpi_1 = &\: -\sum_{\bA} F_{\bA}^\varpi (\ud u_{\bA}^0 + \ud u_{\bA}^2) \sin (\tfrac{a_{\bA} u_{\bA}}\lambda) + \lambda \sum_{\bA} a_{\bA}^{-1} \ud F_{\bA}^\varpi \cos (\tfrac{a_{\bA} u_{\bA}}\lambda) +\cdots \\
= &\: -\sum_{\bA} F_{\bA}^\varpi \ud u^0_{\bA} \sin (\tfrac{a_{\bA} u_{\bA}}\lambda) + \lambda  \sum_{\bA} \Big(a_{\bA}^{-1} \ud F_{\bA}^\varpi \cos (\tfrac{a_{\bA} u_{\bA}}\lambda) - \f 12  F_{\bA}^\varpi \mathfrak v_\bA \sin (\tfrac{2a_{\bA} u_{\bA}}\lambda)\Big)+\cdots, \\
\ud \varpi_2 = &\: \lambda \Big(\sum_{\bA} F_\bA^{1,\phi} \ud u_{\bA}^0 \cos(\tfrac{a_\bA u_\bA}{\lambda}) - 2 \sum_{\bA} F_\bA^{2,\phi} \ud u_{\bA}^0 \sin(\tfrac{2a_\bA u_\bA}{\lambda}) \\
&\: \qquad - \sum_{\pm, \bA,\bB:\bA\neq \bB} F^\phi_{\bA,\bB} (a_\bA \ud u_\bA^0 \pm a_\bB \ud u_\bB^0)\sin(\tfrac{a_\bA u_\bA \pm a_\bB u_\bB}{\lambda}) \Big) + \cdots, \\
e^{-4\phi} - e^{-4\phi_0} = &\: -4 \lambda e^{-4\phi_0} \sum_{\bA} a_{\bA}^{-1} F^\phi_{\bA} \cos(\tfrac{a_{\bA} u_{\bA}}\lambda) + \cdots.
\end{align*}

We can read off the terms which give $O(\lambda^0)$ contribution:
\begin{equation*}
\begin{split}
-	\sum_{\bA }e^{-4\phi_0} g_0^{-1}(\ud\varpi_0,\ud u_\bA^0)F_\bA^{\varpi}\sin(\tfrac{a_\bA u_\bA}{\lambda})+\sum_{\pm,\bA,\bB: \bA \neq \bB}(\mp 1) \f{e^{-4\phi_0}}{4}g_0^{-1}(\ud u_\bA^0,\ud u_\bB^0)F_\bA^{\varpi}F_\bB^{\varpi}\cos(\tfrac{a_\bA u_\bA^0\pm a_\bB u_\bB^0}{\lambda}).
\end{split}
\end{equation*}
Here, we used that $g_0^{-1}(\ud u_\bA^0,\ud u_\bA^0)$ is of order $\lambda$ since $g_0^{-1}(\ud u_\bA,\ud u_\bA) = 0$. For convenience of comparison with Lemma~\ref{lem:wave.main.elliptic.terms} later (see Proposition~\ref{prop:wave.est.final}), we then rewrite $g_0^{-1}(\ud u_\bA^0,\ud u_\bB^0)F_\bA^{\varpi}F_\bB^{\varpi}\cos(\tfrac{a_\bA u_\bA^0\pm a_\bB u_\bB^0}{\lambda})$ as $g^{-1}(\ud u_\bA,\ud u_\bB)F_\bA^{\varpi}F_\bB^{\varpi}\cos(\tfrac{a_\bA u_\bA \pm a_\bB u_\bB}\lambda)$, at the expense of acceptable terms with phases $\cos(\tfrac{a_\bA u_\bA}\lambda)$ or $\cos(\tfrac{2 a_\bA u_\bA\pm a_\bB u_\bB }\lambda)$ using \eqref{eq:u.para}--\eqref{u1.def}. For the remaining terms, it is straightforward to check that they have the desired form. \qedhere

%
%
\end{proof}

We summarize the computations so far. Notice that the $O(\lambda^0)$ terms are set up to cancel, and we can collect all the $O(\lambda)$ terms. We also give the corresponding statement for $\varpi$ which can be obtained using similar computations as above.

\begin{proposition}\label{prop:wave.error}
Define $F_\bA^{1,\phi}$, $F_\bA^{1,\varpi}$, $F_\bA^{2,\phi}$ and $F_\bA^{2,\varpi}$ by \eqref{whtfi}--\eqref{whtpi2}, where\footnote{Since the equations for $F_\bA^{1,\phi}$, $F_\bA^{2,\phi}$ couple with those for $F_\bA^{1,\varpi}$, $F_\bA^{2,\varpi}$, in order to solve for $F_\bA^{1,\phi}$, $F_\bA^{2,\phi}$, we also need to prescribe $A^{1,\varpi}_{\bA}$ and $A^{2,\phi}_\bA$. We will not spell that out that explicitly, but it is clear that they can be handled with the same argument.} 
$$A^{1,\phi}_{\bA}= -A^{\phi_0}_{1,\bA}-A^{sl}_{1,\bA}, \qquad A^{2,\phi}_\bA= \f 12A^{\phi,\Box}_{2,\bA} + \f 12 A^{\phi_0}_{2,\bA} + \f 12 A^{sl}_{2,\bA},$$
and similarly for $A^{1,\varpi}_{\bA}$ and $A^{2,\varpi}_\bA$.

Then
\begin{align*}
	\Box_g \wht \phi = &\lambda \sum_{\pm,\bA,\bB: \bB\neq \bA} E^\phi_{\pm,\bA \bB} \sin( \tfrac{a_\bA u_\bA \pm a_\bB u_\bB}{\lambda})
	+\lambda\sum_{\pm, \bA,\bB:\bB\neq \bA} E^\phi_{\pm,2\bA,\bB} \cos(\tfrac{2a_\bA u_\bA \pm a_\bB u_\bB}{\lambda})\\
	&+\lambda \sum_{\substack{\pm_1,\pm_2,\bA,\bB,\bC: \\ \bB\neq \bA, \bC \neq \bA,\bB}} E^\phi_{\pm_1,\pm_2,\bA,\bB,\bC}\cos(\tfrac{a_\bA u_\bA \pm_1  a_\bB u_\bB \pm_2 a_\bC u_\bC }{\lambda})\\
	&+ \frac{1}{n_0^3}\partial_t \wht n \partial_t \phi^0 + \frac{1}{n_0^2}\partial_t \wht \beta^i \partial_i \phi^0+\cdots,
\end{align*}
where
\begin{align*}
	E^\phi_{\pm,\bA \bB} &=-A^\Box_{\pm,\bA \bB} -A^{FAB}_{\pm,\bA \bB} -A^{\phi_0}_{\pm,\bA \bB} -A^{sl}_{\pm,\bA \bB}, \\
	E^\phi_{\pm,2\bA,\bB} &=-A^{\phi,\Box}_{\pm,2\bA,\bB}-A^{sl}_{\pm,2\bA,\bB},  \\
 E^\phi_{\pm_1,\pm_2,\bA,\bB,\bC}&= -A^\Box_{\pm_1,\pm_2,\bA,\bB,\bC}- A^{sl}_{\pm_1,\pm_2,\bA,\bB,\bC},
\end{align*}
and the error term $\cdots$ satisfies \eqref{eq:error.general}. Similarly for $\wht \varpi$.
\end{proposition}
\begin{proof}
Observe the following cancellations:
\begin{enumerate}
\item The $O(\lambda^0)$ terms with phase $\sin (\tfrac{a_\bA u_\bA}\lambda)$ combine to cancel thanks to the transport equation \eqref{back.without.omega.4} for $F^\phi_\bA$.
\item The $O(\lambda^0)$ terms with phases $\cos(\tfrac{a_\bA u_\bA\pm a_\bB u_\bB}{\lambda})$ in Lemma~\ref{lem:wave.nonlinear} are cancelled by the terms in Lemma~\ref{lem:wave.main.elliptic.terms}.
\item For the $O(\lambda^1)$ terms, the transport equations for $F^{1,\phi}_\bA$ and $F^{2,\phi}_\bA$ are exactly designed to cancel the terms with the phases $\sin(\tfrac{2 a_\bA u_\bA}{\lambda})$ and $\cos(\tfrac{a_\bA u_\bA}{\lambda})$.
\end{enumerate}
We are thus exactly left with the terms as stated. \qedhere
\end{proof}

\begin{prp}\label{prop:wave.est.final}
The following estimates hold:
\begin{align*}
\sum_{k \leq 3}\lambda^k \|\partial \wht \phi\|_{H^k(\Sigma_t)} + \sum_{k \leq 3}\lambda^k \|\partial \wht \varpi\|_{H^k(\Sigma_t)}\lesssim C(N) \lambda^2\ep+ \ep^{\f 54} \lambda^2 C_b(N)e^{A(N)t}.
\end{align*}

\end{prp}
 
 \begin{proof}
 	We focus on $\wht \phi$ as $\wht \varpi$ is similar. 
	
	We first define
\begin{align*}
 \wht \phi^{main}
	=&\:-\lambda^3 \sum_{\pm, \bA,\bB: \bA \neq \bB} \frac{E^\phi_{\pm,\bA \bB} \sin( \tfrac{a_\bA u_\bA \pm a_\bB u_\bB}{\lambda})}{g_0^{-1}(\ud(a_\bA u^0_\bA \pm a_\bB u^0_\bB),\ud(a_\bA u^0_\bA \pm a_\bB u^0_\bB))} \\
	&\:-\lambda^3\sum_{\pm, \bA,\bB: \bA \neq \bB}  \frac{E^\phi_{\pm,\bA,2\bB} \cos(\tfrac{a_\bA u_\bA \pm 2 a_\bB u_\bB}{\lambda}) }{g_0^{-1}(\ud(a_\bA u^0_\bA \pm 2a_\bB u^0_\bB),\ud(a_\bA u^0_\bA \pm 2a_\bB u^0_\bB))} \\
	&\:-\lambda^3\sum_{\substack{\pm_1,\pm_2, \bA,\bB,\bC:\\ \bB \neq \bA,\bC \neq \bA,\bB}}  \frac{E^\phi_{\pm_1,\pm_2,\bA,\bB,\bC}\cos(\tfrac{a_\bA u_\bA \pm_1  a_\bB u_\bB \pm_2 a_\bC u_\bC }{\lambda})}{g_0^{-1}(\ud (a_\bA u^0_\bA \pm_1 a_\bB u^0_\bB\pm_2 a_\bC u^0_\bC),\ud (a_\bA u^0_\bA \pm_1 a_\bB u^0_\bB\pm_2 a_\bC u^0_\bC))} ,
\end{align*}
where the terms $E^\phi_{\pm,\bA \bB}$, $E^\phi_{\pm,\bA,2\bB}$ and $E^\phi_{\pm_1,\pm_2,\bA,\bB,\bC}$ are as in Proposition~\ref{prop:wave.error}.
Notice that these are well-defined as $g_0^{-1} (\ud (a_\bA u^0_\bA \pm a_\bB u^0_\bB),\ud (a_\bA u^0_\bA \pm a_\bB u^0_\bB))$, etc.~are bounded away from $0$ by null adaptedness of the phases (see Lemma~\ref{lemma.adapt}).
 
We first observe that 
\begin{equation}\label{eq:bounds.for.phimain}
\sum_{k\leq 3} \lambda^k \|\rd \wht \phi^{main} \|_{H^k(\Sigma_t)} \leq C(N)\ep \lambda^2,
\end{equation}
which is much better than needed since $C(N) \ll C_b(N)\ep^{\f 14}$ (see \eqref{eq:N.vs.ep}) so that it remains to control $\wht \phi - \wht \phi^{main}$.
 Moreover, $\wht \phi^{main}$ is defined to remove the main terms in Proposition~\ref{prop:wave.error} so that
 	\begin{equation}\label{eq:wave-wave.main}
	\Box_g (\wht \phi- \wht \phi^{main})
=\frac{1}{n_0^3}\partial_t \wht n \partial_t \phi^0 + \frac{1}{n_0^2}\partial_t \wht \beta^i \partial_i \phi^0+\cdots,
\end{equation}
where as before $\cdots$ satisfies \eqref{eq:error.general}.
 
 
We control $\wht \phi- \wht \phi^{main}$ with energy estimate. The key is to integrate by parts to exchange $\partial_t$ derivatives on $ \wht \mfg$ with spatial derivative, using the wave equation satisfied by $\wht \phi$. This is similar to Proposition 6.21 of \cite{HL.HF}; we will give a sketch here and refer the reader to \cite{HL.HF} for details.

Consider first the $L^2(\Sigma_t)$ estimate for $\rd (\wht \phi- \wht \phi^{main})$, which is derived using $\rd_t (\wht \phi- \wht \phi^{main})$ as a multiplier. Our bootstrap assumptions give sufficient control of the geometry so that we have
\begin{equation}\label{eq:EE.details.1}
\begin{split}
 		&\: \|\partial(\wht \phi- \wht \phi^{main})\|_{L^2(\Sigma_t)}^2\\
 		\lesssim &\: C(N) \lambda^4\ep^2 +\Big( \int_0^t \|\cdots\|_{L^2(\Sigma_\tau)} \ud\tau \Big)^2 \\
		&\: + \Big| \int_0^t \int_{\Sigma_\tau}
 		\left(\frac{1}{n_0^3} \rd_t \wht n \partial_t \phi^0 + \frac{1}{n_0^2} \rd_t \wht \beta^i \partial_i \phi^0\right)\partial_t (\wht \phi- \wht \phi^{main}) \sqrt{|\det g|} \ud x \ud \tau \Big|.
\end{split}
\end{equation}
For the last term, we integrate by parts to remove the $\partial_t$ derivatives from $ \wht \mfg$. When the $\rd_t$ derivatives hits on $\rd_t (\wht \phi- \wht \phi^{main})$ we then use the wave equation (Lemma~\ref{lem:wave.operator}); whenever there are two derivatives, there must be one spatial derivative which we can integrate by parts back to $\wht \mfg$. Then note that by \eqref{BA:g.L4} the spatial derivative of $\wht\mfg$ can be considered an $\cdots$ error term. In other words, capturing also the boundary term from the integration by parts, we obtain
\begin{equation}\label{eq:EE.details.2}
\begin{split}
 		&\: \Big| \int_0^t \int_{\Sigma_\tau}
 		\left(\frac{1}{n_0^3} \rd_t \wht n \partial_t \phi^0 + \frac{1}{n_0^2} \rd_t \wht \beta^i \partial_i \phi^0\right)\partial_t (\wht \phi- \wht \phi^{main}) \sqrt{|\det g|} \ud x \ud \tau \Big|\\
 		\lesssim &\: C(N) \lambda^4\ep^2 + \Big| \int_{\Sigma_t}
 		\left(\frac{1}{n_0^3} \wht n \partial_t \phi^0 + \frac{1}{n_0^2} \wht \beta^i \partial_i \phi^0\right)\partial_t (\wht \phi- \wht \phi^{main}) \sqrt{|\det g|}\ud x \Big| \\
 	&\:	+\int_0^t \|\cdots \|_{L^2(\Sigma_\tau)} \|\rd (\wht \phi- \wht \phi^{main})\|_{L^2(\Sigma_\tau)} \ud\tau \\
	\ls &\: C(N) \lambda^4\ep^2 + \|\rd (\wht \phi- \wht \phi^{main})\|_{L^2(\Sigma_t)} \|\wht g\|_{L^2(\Sigma_t\cap B(0,2R+2))} + \int_0^t \|\cdots \|_{L^2(\Sigma_\tau)} \|\rd (\wht \phi- \wht \phi^{main})\|_{L^2(\Sigma_\tau)} \ud\tau.
\end{split}
\end{equation}

Combining \eqref{eq:EE.details.1} and \eqref{eq:EE.details.2}, using \eqref{BA:wave}, \eqref{BA:g.L4}, \eqref{eq:bounds.for.phimain} and the bounds \eqref{eq:error.general} for $\|\cdots\|_{L^2(\Sigma_t)}$, we thus obtain 
\begin{equation}\label{eq:EE.details.3}
\begin{split}
 		&\: \|\partial(\wht \phi- \wht \phi^{main})\|_{L^2(\Sigma_t)}^2\\
 		\lesssim &\: C(N) \lambda^4\ep^2 +\ep^2 \lambda^4 \Big( \int_0^t C(N) C_b(N) e^{A(N)\tau} \ud\tau \Big)^2 \\
		&\:+ \ep^{\f 52} \lambda^4 (C_b(N))^2 e^{2A(N)t} + \ep^2 \lambda^4\int_0^t C(N) (C_b(N))^2 e^{2A(N)\tau} \ud\tau \\
		\ls &\: C(N) \lambda^4\ep^2+\ep^{\f 52} \lambda^4 (C_b(N))^2 e^{2A(N)t},
\end{split}
\end{equation} 		
where we have used \eqref{eq;constant.hierarchy}, which in particular imply that the time integrals satisfy
$$\Big(\int_0^t C_b(N) e^{A(N)\tau} \ud \tau\Big)^2 ,\,\int_0^t (C_b(N))^2 e^{2A(N)\tau} \ud\tau \leq e^{2A(N)t}.$$
		
		The concludes the $k=0$ case of the estimates. The higher derivative estimates up to three derivatives can be handled similarly after noting that the $\cdots$ terms and the bootstrap assumptions lose a power of $\lambda$ for every additional derivative. \qedhere
 \end{proof} 		

 \begin{lemma}
The following estimate holds for all $t \in [0,T)$:
 	\begin{align*}
	\sum_{k\leq 2}\lambda^k \|\partial_t^2 \wht \phi \|_{H^k(\Sigma_t)}+ \sum_{k\leq 2}\lambda^k \|\partial_t^2 \wht \varpi \|_{H^k(\Sigma_t)}\lesssim 
  C(N) \lambda\ep+\lambda\ep^{\frac{5}{4}}C'_b(N)e^{A(N)t}.
	\end{align*}
 	\end{lemma}
\begin{proof}
	As before we focus on $\wht \phi$. We write $\wht \phi = (\wht \phi - \wht \phi^{main}) + \wht \phi^{main}$ as in Proposition~\ref{prop:wave.est.final}. It is straightforward to check that the term $\wht \phi^{main}$ satisfies the desired estimate. For $\wht \phi - \wht \phi^{main}$, we use the wave equation \eqref{eq:wave-wave.main} (see Lemma~\ref{lem:wave.operator}) to rewrite $\rd_t^2 (\wht \phi - \wht \phi^{main})$ as good terms plus terms involving either one derivative of $\wht \phi - \wht \phi^{main}$ or two derivatives of $\wht \phi - \wht \phi^{main}$ with at least one spatial derivative. Thus, using Proposition~\ref{prop:wave.est.final}, the bootstrap assumptions and the estimates on $\cdots$ in \eqref{eq:wave-wave.main}, we obtain the desired bound. Notice that the larger bootstrap constant $C'_b(N)$ appears because of the $\partial_t \wht \mfg$ terms in the equation \eqref{eq:wave-wave.main}; see \eqref{BA:dtg.L43}, \eqref{BA:dtg.L4}, \eqref{BA:dtg.L2}. \qedhere
	
\end{proof}

\section{Estimates for the solution to the eikonal equation}\label{sec:eikonal}

The aim of this section is to improve the estimates for $u_\bA$. Recall that $u_\bA = u_\bA^0 + u_\bA^2 + \widetilde{u}$ (see \eqref{eq:u.para}). The bounds for the background $u_\bA^0$ are given. We begin with the estimates for $u_\bA^2$. 

\begin{lm}\label{lem:u2.well.defined}
The functions $u^2_\bA$, given implicitly by
$$u^2_\bA = \lambda^2 a_\bA^{-1} \mathfrak v_\bA \sin(\tfrac{a_\bA u_\bA}{\lambda})$$
are well-defined. Moreover,
\begin{equation}\label{eq:dalp.u2}
\rd_\alp u^2_\bA = \lambda^2 \mathfrak b_{\bA,\alp}^{(2)} + \lambda \mathfrak b_\bA^{(1)} \rd_{\alp} (u_\bA^0 + \wht u_\bA),
\end{equation}
where $$\mathfrak b_{\bA,\alp}^{(2)} = \f{a_\bA^{-1} \rd_\alp \mathfrak v_\bA \sin(\tfrac{a_\bA u_\bA}{\lambda})}{1 - \mathfrak v_\bA\lambda \cos (\tfrac{a_\bA u_\bA}{\lambda})},\quad \mathfrak b_\bA^{(1)} = \f{\mathfrak v_\bA \cos (\tfrac{a_\bA u_\bA}{\lambda})}{1 - \mathfrak v_\bA\lambda \cos (\tfrac{a_\bA u_\bA}{\lambda})}.$$
Furthermore, the following estimates hold for $\rd u^2_\bA$:
\begin{equation}\label{eq:u2.est}
\| \rd u^2_\bA \|_{L^{4}(K_t)} \leq C \ep^2 \lambda,\quad \sum_{1\leq k\leq 3} \lambda^k \| \rd u^2_\bA \|_{W^{k,4}(K_t)} +  \lambda^4 \| \rd u^2_\bA \|_{W^{4,2}(K_t)} \leq C(N) \ep^2 \lambda.
\end{equation}
\end{lm}
\begin{proof}
After using the bounds for $\mathfrak v_\bA$ and the bootstrap assumptions for $\widetilde{u}_\bA$, the well-definedness is an immediate consequence of the implicit function theorem for $\lambda$ sufficiently small. The equation \eqref{eq:dalp.u2} simply follows from the definition of $u^2_\bA$. The bound \eqref{eq:u2.est} can then be derived from \eqref{eq:dalp.u2} after using the estimates for $\mathfrak v_\bA$, $\widetilde{u}_\bA$. \qedhere
\end{proof}

Next, we turn to the estimates for $\wht u_\bA$. We start with an easy lemma about transport equations.
\begin{lm}\label{lmtransport}
	Let $f$ be a solution to the transport equation
	$$X^\alpha \partial_\alpha f +hf = F.$$
	Assume that $X$ is $C^1$ and satisfies
	\begin{equation}\label{eq:X.basic}
	\frac{1}{2} \leq X^t \leq 4, \qquad |X^r| \leq 2X^t.
	\end{equation}
	Then for $p = 2,4$, $f$ satisfies the estimate
	$$\sup_{t\in[0,T)}\|f\|_{L^p(K_t)} \leq C\left(\|f\|_{L^p(K_0)} + \int_0^T \|F\|_{L^p(K_s)}ds\right) \exp \Big(C\int_0^T \Big(\|\partial_\alpha X^\alpha\|_{L^\infty(K_s)} +\|h\|_{L^\infty(K_s)}\Big) \ud s\Big).$$
	\end{lm} 
\begin{proof}
	We multiply the equation by $pf^{p-1}$ and integrate over $\displaystyle \cup_{s=0}^t K_t$. Integrating by parts, we obtain
$$
		\int_{K_t} X^t f^p -\int_{K_0} X^tf^p+\int_0^t \int_{\rd K_s} (2X^t - X^r) f^p d\omega
		=\int_0^t \int_{K_s} (f^p (\partial_\alpha X^\alpha - ph) +pF f^{p-1}).$$
		Since $2X^t -X^r \geq 0$, the term on $\rd K_s$ has a good sign and can be dropped. Thus, using $\f 12 \leq X^t \leq 4$ and absorbing suitable terms to the left-hand side, we obtain
		\begin{equation*}
		\begin{split}
		\sup_{s\in[0,t]} \|f\|_{L^p(K_s)}^p \leq &\: C \|f\|_{L^p(K_0)}^p + \Big(\int_0^t \|F\|_{L^p(K_s)} \, \ud s\Big)^p \\
		&\: + \int_0^t \|f\|_{L^p(K_s)}^p ( \|\partial_\alpha X^\alpha\|_{L^\infty(K_s)} +\|h\|_{L^\infty(K_s)}) \, \ud s.
		\end{split}
		\end{equation*}
		We conclude with Gr\"onwall's lemma. \qedhere

\end{proof}


\begin{lm}\label{lmdiuA} For $i \in \{1,2\}$, the functions
	$\partial_i \wht u_\bA$ satisfy the equation
	\begin{align*}
	X_\bA^\beta \partial_\beta \partial_i \wht u_\bA =&\: -\lambda^2 \sum_{\bB:\bB \neq \bA } \rd_i( R_{\bA}( \q G_{2,\bB}(\mfg)) \cos(\tfrac{2a_\bB u_\bB}{\lambda}))
		-\lambda^2 \sum_{\bB: \bB \neq \bA} \rd_i (R_{\bA}(\q G_{1,\bB} (\mfg)) \sin(\tfrac{a_\bB u_\bB}{\lambda})) \\
		&\: -\lambda^2 \sum_{\pm,\bB,\bC:\bB\neq\bC} \rd_i (R_{\bA}(\q G_{\pm,\bB,\bC}(\mfg)) \cos(\tfrac{a_\bB u_\bB \pm a_\bC u_{\gra C}}{\lambda})) + \cdots,
	\end{align*}
	where $X^\bt_{\bA} = 2 (1+ \lambda \mathfrak b_\bA^{(1)}) g^{\alp\bt} \rd_\alp u_\bA$ for $\mathfrak b_\bA^{(1)}$ as in Lemma~\ref{lem:u2.well.defined}, $R_\bA(\q G)$ have an explicit expression in term of background quantities satisfying
	\begin{equation}\label{eq:R.G.bound}
	\| R_{\bA,i}( \q G_{2,\bB}(\mfg))\|_{C^8(K_t)} + \| R_{\bA,i}(\q G_{1,\bB} (\mfg))\|_{C^8(K_t)} + \| R_{\bA,i}(\q G_{\pm,\bA,\bB}(\mfg))\|_{C^8(K_t)} \leq C(N) \ep^2,
	\end{equation}
and the error terms $\cdots$ satisfy the following estimates:
	\begin{align}
	\sum_{k\leq 3}\lambda^k \|\cdots\|_{W^{k,4}(K_t)} + \sum_{k\leq 2}\lambda^{k+1} \|\rd_t(\cdots)\|_{W^{k,4}(K_t)} \leq &\: \lambda^2\ep^{\frac{3}{2}}C(N)C_b(N)e^{A(N)t}, \label{eq:du.error.bound.1}\\
	\|\cdots\|_{H^4(K_t)} + \lambda \|\rd_t(\cdots)\|_{H^3(K_t)} \leq &\: \lambda^{-2} \ep^{\frac{3}{2}}C(N)C_b(N)e^{A(N)t}.\label{eq:du.error.bound.2}
	\end{align}
\end{lm}
\begin{proof}
	Differentiating $g^{\alp\bt} \rd_\alp u_\bA \rd_\bt u_\bA = 0$, we obtain
	\begin{equation}\label{eq:eikonal.diff}
	2 g^{\alp\bt} \rd_\alp u_\bA  \rd_\bt \rd_i u_\bA + \rd_i g^{\alp \bt} \rd_\alp u_\bA \rd_\bt u_\bA = 0.
	\end{equation}
	For the first term in \eqref{eq:eikonal.diff}, we expand $\rd^2_{\bt i} u_\bA$ using $u_\bA = u_\bA^0+u_\bA^2+\wht u_\bA$ and then use \eqref{eq:dalp.u2} to obtain
	\begin{equation}\label{eq:eikonal.diff.1}
	\begin{split}
	&\: 2 g^{\alp\bt} \rd_\alp u_\bA  \rd_\bt \rd_i u_\bA \\
	= &\: 2 g^{\alp\bt} \rd_\alp u_\bA  \rd_\bt \rd_i u_\bA^0 + 2 g^{\alp\bt} \rd_\alp u_\bA  \rd_\bt (\lambda^2 \mathfrak b_{\bA,i}^{(2)} + \lambda \mathfrak b_\bA^{(1)} \rd_i (u_\bA^0 + \wht u_\bA) ) + 2 g^{\alp\bt} \rd_\alp u_\bA  \rd_\bt \rd_i \wht u_\bA  \\
	= &\: (2 + 4 \lambda \mathfrak v_\bA \cos(\tfrac{a_\bA u_\bA}{\lambda})) g_0^{\alp\bt} \rd_\alp u_\bA^0  \rd_\bt \rd_i u_\bA^0 + 2 \lambda g_0^{\alp\bt} \rd_\alp u_\bA^0  \rd_\bt \mathfrak v_\bA \rd_i u_\bA^0 \cos(\tfrac{a_\bA u_\bA}{\lambda})\\
	&\: + 2 (1+ \lambda \mathfrak b_\bA^{(1)}) g^{\alp\bt} \rd_\alp u_\bA  \rd_\bt \rd_i \wht u_\bA + \cdots,
	\end{split}
	\end{equation}
	where in the second equality we used the bootstrap assumptions \eqref{BA:g.L4}, \eqref{BA:g.L2}, \eqref{BA:u.L4}, \eqref{BA:u.L2}, as well as
	\begin{align*}
	g^{\alp\bt} \rd_\alp u_\bA^2 \rd_\bt \rd_i u_\bA^0 = &\: \lambda g_0^{\alp\bt} \mathfrak v_\bA \rd_\alp u_\bA^0 \rd_\bt \rd_i u_\bA^0 \cos(\tfrac{a_\bA u_\bA}{\lambda}) + \cdots,\\
	\lambda g^{\alp\bt} \rd_\alp u_\bA \rd_\bt \mathfrak b_\bA^{(1)}  =&\: \lambda g_0^{\alp\bt} \rd_\alp u_\bA^0 \rd_\bt \mathfrak v_\bA \cos(\tfrac{a_\bA u_\bA}\lambda) + \cdots,\\
	\lambda^2 g^{\alp\bt} \rd_\alp u_\bA  \rd_\bt \mathfrak b_{\bA,i}^{(2)} = &\: \cdots,
	\end{align*}
	where both statements are derived with the help of the eikonal equation when the $\rd_\bt$ derivatives hit on the phase. We note that since we have $L^4$-based norms for the lower order estimates of $g$ and $u_\bA$ and $L^2$-based norms for the top order (see \eqref{BA:g.L4}, \eqref{BA:g.L2}, \eqref{BA:u.L4}, \eqref{BA:u.L2}), we likewise obtain such a hierarchy of estimates for the error terms.
	
	For the second term in \eqref{eq:eikonal.diff}, we expand using the paramatrices and obtain the following using the bootstrap assumptions:
	\begin{equation}\label{eq:eikonal.diff.2}
	\begin{split}
	&\: \rd_i g^{\alp \bt} \rd_\alp u_\bA \rd_\bt u_\bA \\
	= &\: \rd_i g_0^{\alp \bt} \rd_\alp u_\bA^0 \rd_\bt u_\bA^0 + 2\rd_i g_0^{\alp \bt} \rd_\alp u_\bA^2 \rd_\bt u_\bA^0 + \rd_i (g^{\alp\bt} - g^{\alp\bt}_0) \rd_\alp u_\bA^0 \rd_\bt u_\bA^0 + \cdots\\
	= &\: (1 + 2 \lambda \mathfrak v_\bA \cos(\tfrac{a_\bA u_\bA}\lambda)) \rd_i g_0^{\alp \bt} \rd_\alp u_\bA^0 \rd_\bt u_\bA^0 + \rd_i (g^{\alp\bt} - g^{\alp\bt}_0) \rd_\alp u_\bA^0 \rd_\bt u_\bA^0  + \cdots.
	\end{split}
	\end{equation}
	
	Differentiate now the eikonal equation for the background by $\rd_i$ to obtain
	\begin{equation}\label{eq:eikonal.bkgd.diff}
	2 g_0^{\alp\bt} \rd_\alp u_\bA^0  \rd_\bt \rd_i u_\bA^0 + \rd_i g_0^{\alp \bt} \rd_\alp u_\bA^0 \rd_\bt u_\bA^0 = 0.
	\end{equation}
	Thus, plugging \eqref{eq:eikonal.diff.1} and \eqref{eq:eikonal.diff.2} into \eqref{eq:eikonal.diff} and using \eqref{eq:eikonal.bkgd.diff}, we obtain
	\begin{equation}\label{eq:utilde.eqn.almost}
	\begin{split}
	&\: 2 (1+ \lambda \mathfrak b_\bA^{(1)}) g^{\alp\bt} \rd_\alp u_\bA  \rd_\bt \rd_i \wht u_\bA \\ 
	= &\: - 2 \lambda g_0^{\alp\bt} \rd_\alp u_\bA^0  \rd_\bt \mathfrak v_\bA \rd_i u_\bA^0 \cos(\tfrac{a_\bA u_\bA}{\lambda}) - \rd_i (g^{\alp\bt} - g^{\alp\bt}_0) \rd_\alp u_\bA^0 \rd_\bt u_\bA^0 + \cdots \\
	= &\: - 2 \lambda g_0^{\alp\bt} \rd_\alp u_\bA^0  \rd_\bt \mathfrak v_\bA \rd_i u_\bA^0 \cos(\tfrac{a_\bA u_\bA}{\lambda}) - \rd_i ((g^{\alp\bt} - g^{\alp\bt}_0) \rd_\alp u_\bA^0 \rd_\bt u_\bA^0) + \cdots,
	\end{split}
	\end{equation}
	where we used that $(g^{\alp\bt} - g^{\alp\bt}_0)\rd_i (\rd_\alp u_\bA^0 \rd_\bt u_\bA^0)$ satisfies \eqref{eq:du.error.bound.1}.
	
	We now expand $(g^{\alp\bt} - g^{\alp\bt}_0) \rd_\alp u_\bA^0 \rd_\bt u_\bA^0$ using $\mfg = \mfg_0 + \mfg_2 + \widetilde{\mfg}$. Notice that we only need to keep the $\mfg_2$ terms here, as the contributions from $\wht \mfg$ can be treated as error. 
	Using the ansatz for $\mfg_2$ (see \eqref{eq:mfg2.form}, \eqref{eq:g2.def.1}--\eqref{eq:g2.def.3}) we can write
	\begin{equation}\label{eq:main.R.computation.in.g}
	\begin{split}
	 (g^{\alp\bt} - g^{\alp\bt}_0) \rd_\alp u_\bA^0 \rd_\bt u_\bA^0 
	= &\: \lambda^2 \sum_{\bB } R_{\bA}( \q G_{2,\bB}(\mfg)) \cos(\tfrac{2a_\bB u_\bB}{\lambda})
		+\lambda^2 \sum_{\bB} R_{\bA}(\q G_{1,\bB} (\mfg)) \sin(\tfrac{a_\bB u_\bB}{\lambda})\\
		&+\lambda^2 \sum_{\pm,\bB,\bC:\bB\neq\bC} R_{\bA}(\q G_{\pm,\bB,\bC}(\mfg)) \cos(\tfrac{a_\bB u_\bB \pm a_\bC u_{\gra C}}{\lambda}) + \cdots,
	\end{split}
	\end{equation}
	where $R_{\bA}(\mathcal G)$ satisfies the estimates \eqref{eq:R.G.bound}.

	Finally, we isolate the $\bB = \bA$ contribution in the first two terms on the right-hand side of \eqref{eq:main.R.computation.in.g}. We relegate the computation to Lemma \ref{lmgprim} below, according to which we have 
	\begin{equation}
	\begin{split}
	&\: \rd_i ((g^{\alp\bt} - g^{\alp\bt}_0) \rd_\alp u_\bA^0 \rd_\bt u_\bA^0 )\\
	= &\: \lambda^2 \sum_{\bB:\bB \neq \bA } \rd_i( R_{\bA}( \q G_{2,\bB}(\mfg)) \cos(\tfrac{2a_\bB u_\bB}{\lambda}))
		+\lambda^2 \sum_{\bB: \bB \neq \bA} \rd_i (R_{\bA}(\q G_{1,\bB} (\mfg)) \sin(\tfrac{a_\bB u_\bB}{\lambda})) \\
		&\: +\lambda^2 \sum_{\pm,\bB,\bC:\bB\neq\bC} \rd_i (R_{\bA}(\q G_{\pm,\bB,\bC}(\mfg)) \cos(\tfrac{a_\bB u_\bB \pm a_\bC u_{\gra C}}{\lambda})) \\
		&\: + \frac{\lambda^2}{a_\bA^2} \rd_i \Big( g_0^{\alpha \beta}(4F_\bA^\phi\partial_\alpha \phi_0 \partial_\beta u^0_\bA + e^{-4\phi_0}F_\bA^\varpi\partial_\alpha \varpi_0 \partial_\beta u^0_\bA)\sin(\tfrac{a_\bA u_\bA}{\lambda}) \Big)+ \cdots.
	\end{split}
	\end{equation}
	Plugging this back into \eqref{eq:utilde.eqn.almost}, and dropping acceptable error terms, we obtain 
	\begin{equation}
	\begin{split}
	&\: 2 (1+ \lambda \mathfrak b_\bA^{(1)}) g^{\alp\bt} \rd_\alp u_\bA  \rd_\bt \rd_i \wht u_\bA \\ 
	= &\: - \lambda \rd_i u_\bA^0 \Big(2  g_0^{\alp\bt} \rd_\alp u_\bA^0  \rd_\bt \mathfrak v_\bA + \frac{1}{a_\bA}g_0^{\alpha \beta}(4F_\bA^\phi\partial_\alpha \phi_0 \partial_\beta u^0_\bA + e^{-4\phi_0}F_\bA^\varpi\partial_\alpha \varpi_0 \partial_\beta u^0_\bA) \Big) \cos(\tfrac{a_\bA u_\bA}{\lambda})  \\
	&\: -\lambda^2 \sum_{\bB:\bB \neq \bA } \rd_i( R_{\bA}( \q G_{2,\bB}(\mfg)) \cos(\tfrac{2a_\bB u_\bB}{\lambda}))
		-\lambda^2 \sum_{\bB: \bB \neq \bA} \rd_i (R_{\bA}(\q G_{1,\bB} (\mfg)) \sin(\tfrac{a_\bB u_\bB}{\lambda})) \\
		&\: -\lambda^2 \sum_{\pm,\bB,\bC:\bB\neq\bC} \rd_i (R_{\bA}(\q G_{\pm,\bB,\bC}(\mfg)) \cos(\tfrac{a_\bB u_\bB \pm a_\bC u_{\gra C}}{\lambda})) + \cdots,
	\end{split}
	\end{equation}
	which gives the desired conclusion after using the transport equation \eqref{defu1}. \qedhere
\end{proof}

We now complete the computation that was used in the proof of the previous lemma.
\begin{lm}\label{lmgprim}
	We have $R_{\bA}( \q G_{2,\bA})=0$ and 
	$$R_{\bA}(\q G_{1,\bA})=\frac{4}{a_\bA^2}  g_0^{\alpha \beta}(F_\bA^\phi\partial_\alpha \phi_0 \partial_\beta u^0_\bA + \f 14 e^{-4\phi_0}F_\bA^\varpi\partial_\alpha \varpi_0 \partial_\beta u^0_\bA).$$
\end{lm}
\begin{proof}

We use the following computation (see (4.24), (6.8) in \cite{HL.HF}):
\begin{equation*}
	\begin{split}
	 (g^{\alp\bt} - g^{\alp\bt}_0) \rd_\alp u_\bA^0 \rd_\bt u_\bA^0 
	= \frac{2 n_2}{n_0^3}(e_0 u^0_\bA)^2-2 \gamma_2 e^{-2\gamma_0}|\nabla u^0_\bA|^2+\frac{2}{n_0^2} \beta_2^j \partial_j u_\bA^0 (e_0 u^0_\bA) + \cdots.
	\end{split}
	\end{equation*}
	Now, we plug in the expressions in \eqref{eq:mfg2.form}, \eqref{eq:g2.def.1}--\eqref{eq:g2.def.3} (with $\Gamma$ given in \eqref{G.g}, \eqref{G.N}, \eqref{G.b}) for $n_2$, $\gamma_2$ and $\bt_2$. We then obtain
	\begin{align*}
		&\: R_\bA(\q G_{2,\bA})\\
		=&\: -\frac{1}{8}\frac{(F_\bA^\phi)^2+\frac{1}{4}e^{-4\phi_0}(F_\bA^\varpi)^2}{a_\bA^2|\nabla u_\bA^0|^2}\Bigg(\frac{2\mathbf{\Gamma}_0(n)^{\mu\nu}\partial_\mu u^0_\bA \partial_\nu u_\bA^0}{n_0^3}(e_0( u^0_\bA))^2-2\mathbf{\Gamma}_0(\gamma)^{\mu\nu}\partial_\mu u^0_\bA \partial_\nu u_\bA^0e^{-2\gamma_0}|\nabla u^0_\bA|^2\\
		&\qquad\qquad\qquad\qquad+\frac{2}{n_0^2}\mathbf{\Gamma}_0(\beta^j)^{\mu\nu} \partial_\mu u^0_\bA \partial_\nu u_\bA^0 \partial_j u_\bA^0 (e_0( u^0_\bA))\Bigg)\\
		=&\: -\frac{1}{8}\frac{(F_\bA^\phi)^2+\frac{1}{4}e^{-4\phi_0}(F_\bA^\varpi)^2}{a_\bA^2|\nabla u_\bA^0|^2}\Bigg(\frac{2}{n_0^3}(\frac{2e^{2\gamma_0}}{n_0}(e_0 u_\bA^0)^2)(e_0 u^0_\bA)^2\\
		&\: \qquad \qquad \qquad -2(-|\nabla u_\bA^0|^2-\frac{e^{2\gamma_0}}{n_0^2}(e_0 u_\bA^0)^2)e^{-2\gamma_0}|\nabla u^0_\bA|^2 +\frac{2}{n_0^2}(-4e_0(u_\bA^0)\partial^i u_\bA^0) \partial_i u_\bA^0 (e_0( u^0_\bA))\Bigg)\\
		=&\: 0,
		\end{align*}
	where we have used the notation $e_0 =\partial_t-\beta^i_0 \partial_i$, and the eikonal equation for $u_\bA^0$, which says $|\nabla u_\bA^0|^2 = \frac{e^{2\gamma_0}}{n_0^2}(e_0(u_\bA^0))^2$. 
	
	For $R_\bA(\q G_{1,\bA})$, we compute similarly using $|\nabla u_\bA^0|^2 = \frac{e^{2\gamma_0}}{n_0^2}(e_0(u_\bA^0))^2$ to obtain
	\begin{equation*}
	\begin{split}
	&\: R_\bA(\q G_{1,\bA}) \\
	= &\: - \frac{ 2 F_\bA^\phi}{a_\bA^2|\nabla u^0_\bA|^2} \Big( \f{2}{n_0^3}\f{2 e^{2\gamma_0}}{n_0} (e_0 \phi_0) (e_0 u_\bA^0)(e_0 u_\bA^0)^2 - 2(- \de^{ij} \rd_i \phi_0 \rd_j u_\bA^0 - \frac{e^{2\gamma_0}}{n_0^2} (e_0 \phi_0) (e_0 u_\bA^0))e^{-2\gamma_0}|\nabla u^0_\bA|^2 \\
	&\: \qquad\qquad\qquad\qquad + \f{2}{n_0^2} \de^{ij} (-2 e_0 \phi_0 \rd_j u_\bA^0 - 2 \rd_j \phi_0 e_0 u_\bA)\partial_i u_\bA^0 (e_0 u^0_\bA) \Big) \\
	&\:  - \frac{ F_\bA^\varpi e^{-4\phi_0}}{2 a_\bA^2|\nabla u^0_\bA|^2} \Big( \f{2}{n_0^3}\f{2 e^{2\gamma_0}}{n_0} (e_0 \varpi_0) (e_0 u_\bA^0)(e_0 u_\bA^0)^2 - 2(- \de^{ij} \rd_i \varpi_0 \rd_j u_\bA^0 - \frac{e^{2\gamma_0}}{n_0^2} (e_0 \varpi_0) (e_0 u_\bA^0))e^{-2\gamma_0}|\nabla u^0_\bA|^2 \\
	&\: \qquad\qquad\qquad\qquad + \f{2}{n_0^2} \de^{ij} (-2 e_0 \varpi_0 \rd_j u_\bA^0 - 2 \rd_j \varpi_0 e_0 u_\bA)\partial_i u_\bA^0 (e_0 u^0_\bA) \Big) \\
	= &\: - \frac{ 2 (e_0 u_\bA^0)^2}{a_\bA^2|\nabla u^0_\bA|^2} \Bigg( \f{2 e^{2\gamma_0}}{n_0^4} \Big(F_\bA^\phi(e_0 \phi_0) - \f {e^{-4\phi_0}}4 F_\bA^\varpi(e_0 \varpi_0)\Big) (e_0 u_\bA^0) - \f{2}{n_0^2} \de^{ij} \Big(F_\bA^\phi\rd_i \phi_0 - \f {e^{-4\phi_0}}4 F_\bA^\varpi\rd_i \varpi_0 \Big) \rd_j u_\bA^0\Bigg) \\
	= &\: 4 a_\bA^{-2}   \Big(F_\bA^\phi g_0^{\alp\bt} \rd_\alp \phi_0 \rd_\bt u_\bA^0 + \f {e^{-4\phi_0}}4 F_\bA^\varpi g_0^{\alp\bt} \rd_\alp \varpi_0 \rd_\bt u_\bA^0\Big),
	\end{split}
	\end{equation*}
	where at the very last step we recalled \eqref{g.inverse}. \qedhere 
	
	\end{proof}

\begin{proposition}\label{prpu}
	The following estimate holds for all $t \in [0,T)$:
	$$\sum_{k \leq 3 }\lambda^k\| \partial_i \wht u_\bA\|_{W^{k,4}(K_t)} + \lambda^4 \| \partial_i \wht u_\bA\|_{H^{4}(K_t)} \leq C(N) \ep^{\frac{3}{2}}\lambda^2e^{A(N)t}.$$
\end{proposition}

\begin{proof}
	
	Let 
	\begin{align*}
		\wht u_\bA^{main}=&\:  -\lambda^3\sum_{\bB: \bB \neq \bA}\frac{ R_\bA( \q G_{2,\bB}(\mfg))}{2 a_\bB(g_0^{-1})^{\alpha\beta}\partial_\alpha u_\bA^0 \partial_\beta u_\bB^0} \cos(\tfrac{2a_\bB u_\bB}{\lambda})
		- \lambda^3\sum_{\bB: \bB \neq \bA} \frac{R_\bA(\q G_{1,\bB} (\mfg))}{a_\bB(g_0^{-1})^{\alpha\beta}\partial_\alpha u_\bA^0 \partial_\beta u_\bB^0}\sin(\tfrac{a_\bB u_\bB}{\lambda})\\
		&\: - \lambda^2 \sum_{\pm,\bB,\bC:\bB\neq\bC} \frac{R_\bA(\q G_{\pm,\bB,\bC}(\mfg))}{(g_0^{-1})^{\alpha\beta}\partial_\alpha u_\bA^0 \partial_\beta (a_\bB u_\bB^0 \pm a_\bC u_{\gra C}^0)} \cos(\tfrac{a_\bB u_\bB \pm a_\bC u_{\gra C}}{\lambda}). 
	\end{align*}
	It is easy to check that $\wht u_\bA^{main}$ satisfies the desired estimate. Moreover, 
	\begin{equation}\label{eq;}
	X^\beta_\bA \partial_\beta \partial_i (\wht u_\bA-\wht u_\bA^{main})=\cdots,
	\end{equation}
	with $\cdots$ satisfying \eqref{eq:du.error.bound.1}.
	
	Recall that $X^\bt_{\bA} = 2 (1+ \lambda \mathfrak b_\bA^{(1)}) g^{\alp\bt} \rd_\alp u_\bA$ with $ \mathfrak b_\bA^{(1)} = \f{\mathfrak v_\bA \cos (\tfrac{a_\bA u_\bA}{\lambda})}{1 - \mathfrak v_\bA\lambda \cos (\tfrac{a_\bA u_\bA}{\lambda})}$. By the assumptions on the background, the $L^\i$ estimates in Lemma~\ref{lminfty}, and \eqref{eq:v.est}, the bound \eqref{eq:X.basic} holds and, moreover,
	\begin{equation}\label{eq:div.X.bound}
	\|\partial_\alpha X^\alpha_\bA\|_{L^\infty(K_t)} \lesssim \|\Box_g u_\bA\|_{L^\infty(K_t)} +\|\partial g \partial u_\bA\|_{L^\infty(K_t)} + a_\bA \| \mathfrak v_\bA\|_{W^{1,\infty}(K_t)} \lesssim \ep.
	\end{equation}
	We remark that it is important that the right-hand side of \eqref{eq:div.X.bound} is independent of $N$ (which relies on \eqref{eq:v.est}).
	This allows us to apply Lemma \ref{lmtransport},  which yields
	$$\|\partial_i (\wht u_\bA-\wht u_\bA^{main})\|_{L^4(K_t)} \leq C(N) \ep^2 \lambda^2 + \int_0^t  \lambda^2\ep^{\frac{3}{2}} (C(N)C_b(N))e^{A(N)s} \ud s \leq  \lambda^2 \ep^{\frac{3}{2}}e^{A(N)t},$$
	where we used $A(N)\geq C(N)C_b(N)$, and Lemma~\ref{lem:chi.data} to estimate the initial data. This implies the desired bound when $k=0$.

	The estimates for the derivatives are similar after differentiating the equation.  \qedhere
	
\end{proof}

To obtain estimates for $\partial_t \wht u_\bA$ and $\partial_t^2 \wht u_\bA$,  we use the eikonal equation.
\begin{proposition}\label{prpdt}
	The following estimates hold for all $t \in [0,T)$:
	\begin{align}
	\sum_{k\leq 3} \lambda^k \|\partial_t \wht u_\bA\|_{W^{k,4}(K_t)} + \lambda^4 \|\partial_t \wht u_\bA\|_{H^{4}(K_t)} \leq &\:  \lambda^2 C(N) C_b(N) \ep^{\frac{3}{2}} e^{A(N)t}, \label{eqdtuA}\\
	\sum_{k\leq 2} \lambda^k \|\partial^2_t \wht u_\bA\|_{W^{k,4}(K_t)} + \lambda^3 \|\partial^2_t \wht u_\bA\|_{H^3(K_t)} \leq & \:  \lambda C(N)C'_b(N)\ep^{\frac{3}{2}} e^{A(N)t}. \label{eqdtdtuA}
		\end{align}
	Moreover, at the lowest order, we have the following improved estimates for all $t \in [0,T)$:
	\begin{equation}\label{eq:dtu.improved}
	\|\partial_t \wht u_\bA\|_{L^{4}(K_t)} \ls \lambda^2 C_b(N) \ep^{\frac{3}{2}} e^{A(N)t}, \quad  \|\partial^2_t \wht u_\bA\|_{L^{4}(K_t)} \ls \lambda C_b'(N) \ep^{\frac{3}{2}} e^{A(N)t}.
	\end{equation}
\end{proposition}

\begin{proof}
By the eikonal equation, we have
$$((\rd_t - \bt^i \rd_i) (u_\bA^0 + u_\bA^2 + \wht u_\bA))^2 - n^2 e^{-2\gamma} |\nab u_\bA|^2=0.$$
Solving this as a quadratic equation in $(\rd_t - \bt^i \rd_i)  \wht u_\bA$, we obtain
\begin{equation}\label{eq:dtu.expression}
\begin{split}
(\rd_t - \bt^i \rd_i) \wht u_\bA = &\: -(\rd_t - \bt^i \rd_i)(u_\bA^0 + u_\bA^2) \pm \sqrt{n^2 e^{-2\gamma} |\nab u_\bA|^2} \\
= &\: -(\rd_t - \bt^i \rd_i)(u_\bA^0 + u_\bA^2) \pm n e^{-\gamma} |\nab u_\bA|.
\end{split}
\end{equation}
By the $L^\i$ smallness condition given by Lemma~\ref{lminfty}, we must take the $+$ root above. Using the eikonal equation for the background, we have $(\rd_t - \bt^i_0 \rd_i)u_\bA^0 = n_0 e^{-\gamma_0} |\nab u_\bA^0|$. Also, by definition of $u_\bA^2$, we have
\begin{equation}\label{eq:dtu.expression.dt.u2}
\begin{split}
(\rd_t - \bt_0^i \rd_i) u^2_\bA = &\: \lambda (\rd_t - \bt_0^i \rd_i) u_\bA^0 \mathfrak v_\bA \cos(\tfrac{a_\bA u_\bA}{\lambda})+ \cdots \\
= &\: \lambda n_0 e^{-\gamma_0} |\nab u_\bA^0| \mathfrak v_\bA \cos(\tfrac{a_\bA u_\bA}{\lambda})+ \cdots,
\end{split}
\end{equation}
as well as
\begin{equation}
\begin{split}
|\nab u_\bA| =&\: |\nab u_\bA^0 + \lambda \mathfrak v_\bA \cos(\tfrac{a_\bA u_\bA}{\lambda}) \nab u_\bA^0| + \cdots \\
=&\: (1+  \lambda \mathfrak v_\bA \cos(\tfrac{a_\bA u_\bA}{\lambda})) | \nab u_\bA^0| + \cdots,
\end{split}
\end{equation} 
where we used the background eikonal equation in \eqref{eq:dtu.expression.dt.u2}, and in both instances, the bootstrap assumptions and \eqref{eq:v.est} imply that $\cdots$ satisfy 
\begin{align}
\sum_{k\leq 3} \lambda^k \|\cdots \|_{W^{k,4}(K_t)} + \lambda^4 \| \cdots \|_{H^4(K_t)} \leq &\: C(N) C_b(N) \ep^{\f 32} \lambda^2 e^{A(N)t}, \label{eq:dtu.dots}\\
\sum_{k\leq 2} \lambda^k \|\rd_t(\cdots) \|_{W^{k,4}(K_t)} + \lambda^3 \| \rd_t (\cdots) \|_{H^3(K_t)} \leq &\: C(N) C_b'(N) \ep^{\f 32} \lambda e^{A(N)t},\label{eq:dtu.dtdots}
\end{align}
as well as the improved lowest order estimates
\begin{equation}\label{eq:dtu.dots.lowest}
\| \cdots \|_{L^{4}(K_t)} \ls \lambda^2 C_b(N) \ep^{\frac{3}{2}} e^{A(N)t}, \quad  \|\partial_t (\cdots)\|_{L^{4}(K_t)} \ls \lambda C_b'(N) \ep^{\frac{3}{2}}e^{A(N)t}.
\end{equation}

Thus, plugging all these back into \eqref{eq:dtu.expression}, we obtain
\begin{equation}\label{eq:dtu.expression.final}
\begin{split}
(\rd_t - \bt^i \rd_i) \wht u_\bA 
= &\: -(\rd_t - \bt^i \rd_i)(u_\bA^0 + u_\bA^2) + n e^{-\gamma} |\nab u_\bA| \\
= &\: -(\bt^i_0 - \bt^i)\rd_i (u_\bA^0 + u_\bA^2) - (n_0 e^{-\gamma_0} - n e^{-\gamma}) |\nab u_\bA| + \cdots = \cdots,
\end{split}
\end{equation}
where $\cdots$ continues to denote terms satisfying \eqref{eq:dtu.dots}, \eqref{eq:dtu.dtdots} and \eqref{eq:dtu.dots.lowest}, and in the last equality we used \eqref{eq:calG.est}, \eqref{BA:g.L4}--\eqref{BA:dtg.L2}.

Using \eqref{eq:dtu.dots} and the estimates we obtained in Lemma~\ref{lmdiuA}, we thus obtain \eqref{eqdtuA}. We then differentiate \eqref{eq:dtu.expression.final} by $\rd_t$, use \eqref{eq:dtu.dtdots} and the bound \eqref{eqdtuA} (that we just obtained) to deduce \eqref{eqdtdtuA}. Finally, the improved lowest order bounds \eqref{eq:dtu.improved} hold thanks to \eqref{eq:dtu.dots.lowest}. \qedhere

\end{proof}


\section{Estimates for the metric components}\label{sec:metric}

We now prove estimates for the metric components. We will derive an elliptic equation for $\wht g$ (see Proposition~\ref{prop:whtg.eqn}) and obtain estimates for $\wht g$ that improve the bootstrap assumptions (see Proposition~\ref{prop:g.est}). In the process, we will use $\cdots$ to catch all error terms that obey the following bounds for all $t \in [0,T)$:

\begin{subequations}
\begin{empheq}{align}
				\sum_{k\leq 3}\lambda^k\|\cdots\|_{W^{k,2}_{-\alp_0+2}(\Sigma_t)} \leq &\: \ep^2 \lambda^2 C(N)C_b(N)e^{A(N)t}, \label{eq:g.error.1} \\
				\sum_{k\leq 2}\lambda^{k+1}\|\rd_t(\cdots) \|_{W^{k,2}_{-\alp_0+2}(\Sigma_t)} \leq &\:\ep^2 \lambda C(N)C'_b(N)e^{A(N)t}, 				\label{eq:g.error.2} \\
		\|\cdots \|_{W^{0,\frac{4}{3}}_{-\alp_0+\f 52}(\Sigma_t)} \lesssim &\:\ep^2 \lambda^2 C_b(N)e^{A(N)t},	\label{eq:g.error.3}\\
		\|\rd_t(\cdots)\|_{W^{0,\frac{4}{3}}_{-\alp_0+\f 52}(\Sigma_t)} \lesssim &\:\ep^2 \lambda
		 C'_b(N)e^{A(N)t}. 	\label{eq:g.error.4}
		\end{empheq}
		\end{subequations}
We emphasize that we have better $N$-dependence at the lowest order in \eqref{eq:g.error.3} and \eqref{eq:g.error.4}. This is important for actually improving the bootstrap assumptions. 

To derive the equation for $\mfg$, we use the parametrix given in \eqref{g.para.def}, \eqref{eq:mfg2.form}, the equation \eqref{g.elliptic} as well as its counterpart for $\mfg_0$. For the $\mfg_0$ equation, we start with equations \eqref{elliptic.1}--\eqref{elliptic.4}, but taking into account the matter field so as to obtain
$$\Delta \mfg_0 = {\bf \Gamma}(\mfg_0)^{\mu\nu} \langle \rd_\mu U_0, \rd_\nu U_0 \rangle + \f 12 \sum_{\bA} {\bf \Gamma}(\mfg_0)^{\mu\nu} F_\bA^2 (\rd_\mu u_\bA^0 )(\rd_\nu u_\bA^0)  + \Upsilon(\mfg_0),$$
where as before $F_\bA^2 = (F^\phi_\bA)^2 + \f 14 e^{-4\phi_0} (F^\varpi_\bA)^2$. Hence,
\begin{equation}\label{eq:main.wht.g.prep}
\begin{split}
\Delta \wht \mfg 
=&\:  \Delta \mfg - \Delta \mfg_0 - \Delta \mfg_2 \\
= &\: {\bf \Gamma}(\mfg)^{\mu\nu} \langle \rd_\mu U, \rd_\nu U \rangle - {\bf \Gamma}(\mfg_0)^{\mu\nu} \langle \rd_\mu U_0, \rd_\nu U_0 \rangle -  \f 12 \sum_{\bA} {\bf \Gamma}(\mfg_0)^{\mu\nu} F_\bA^2 (\rd_\mu u_\bA^0 )(\rd_\nu u_\bA^0) \\
&\: +\Upsilon(\mfg) - \Upsilon(\mfg_0)- \Delta \mfg_2.
\end{split}
\end{equation}
We will analyze these terms below. 

In order to analyze the term $ {\bf \Gamma}(\mfg)^{\mu\nu} \langle \rd_\mu U, \rd_\nu U\rangle$ in \eqref{eq:main.wht.g.prep}, we first consider $\langle \partial_\mu U, \rd_\nu U \rangle - \langle \partial_\mu U_0, \rd_\nu U_0 \rangle$ in the proposition below.
\begin{proposition}\label{phi.phi.diff}
	The difference $\langle \partial_\mu U, \rd_\nu U \rangle - \langle \partial_\mu U_0, \rd_\nu U_0 \rangle$ can be expanded as follows:
	\begin{equation}\label{g.expand}
		\begin{split}
			\langle \partial_\mu U, \rd_\nu U \rangle - \langle \partial_\mu U_0, \rd_\nu U_0 \rangle
			= \Theta_{\mu\nu}^{(main)}+\Theta_{\mu\nu}^{(1)} + \Theta_{\mu\nu}^{(2)} + \cdots,
		\end{split}
	\end{equation}
	where
	\begin{equation}\label{g.main.expand}
		\Theta_{\mu\nu}^{(main)}:=\f 12\sum_{\bA} F_\bA^2 (\partial_\mu u_\bA^0)( \partial_\nu u_\bA^0 ),
	\end{equation}
	and
	\begin{equation}\label{g.O1.expand}
		\begin{split}
			\Theta_{\mu\nu}^{(1)}
			:=&-2\sum_{\bA}\left(F_\bA^\phi \partial_{(\mu} \phi_0 +\frac{1}{4}e^{-4\phi_0}F_\bA^\varpi  \partial_{(\mu} \varpi_0 \right)\partial_{\nu)} u^0_\bA \sin(\tfrac{a_\bA u_\bA}{ \lambda})\\
			&	+ \f 12 \sum_{\bA} \left((F^\phi_\bA)^2 +\frac{1}{4}e^{-4\phi_0} (F_\bA^\varpi)^2\right)\partial_{{\mu}} u^0_\bA \partial_{{\nu}} u^0_\bA \cos(\tfrac{2a_{\bA}u_\bA}{ \lambda}) \\
			&+ \f 12\sum_{\pm,\bA,\bB: \bA\neq \bB}(\mp 1)\cdot \left(F^\phi_\bA F^\phi_\bB +\frac{1}{4}e^{-4\varpi_0}F^\phi_\bA F^\varpi_\bB \right) {\rd_\mu u^0_{\bA} \rd_\nu u^0_{\bB} }\cos(\tfrac{a_{\bA}u_\bA \pm a_{\bB}u_\bB}{ \lambda}),
		\end{split}
	\end{equation}
	and
	\begin{equation}\label{g.Ol.expand}
		\begin{split}
			\Theta_{\mu\nu}^{(2)}
			:=& \lambda\sum_{\bA} \left( (	\Theta_{\mu\nu}^{(2)})_{1,\bA} \cos(\tfrac{a_\bA u_\bA}{\lambda}) +(	\Theta_{\mu\nu}^{(2)})_{2\bA} \sin(\tfrac{2a_\bA u_\bA}{\lambda}) + (\Theta_{\mu\nu}^{(2)})_{3\bA} \cos(\tfrac{3a_\bA u_\bA}{\lambda}) \right)\\
			&+\lambda \sum_{\pm,\bA,\bB: \bB \neq \bA } \left( (	\Theta_{\mu\nu}^{(2)})_{\pm,\bA,\bB} \sin(\tfrac{a_\bA u_\bA \pm a_\bB u_\bB }{\lambda}) +(	\Theta_{\mu\nu}^{(2)})_{\pm,2\bA,\bB} \cos(\tfrac{2a_\bA u_\bA\pm a_\bB u_\bB }{\lambda})  \right)\\
			&+\lambda \sum_{\substack{\pm_1,\pm_2,\bA,\bB,\bC\\
			\bB\neq \bA,\,\bC \neq \bA,\bB}}(\Theta_{\mu\nu}^{(2)})_{\pm_1,\pm_2,\bA,\bB,\bC} \cos(\tfrac{a_\bA u_\bA \pm_1 a_\bB u_\bB\pm_1 a_\bC u_\bC }{\lambda})
		\end{split}
	\end{equation}
	with each of the coefficients in $\Theta_{\mu\nu}^{(2)}$ is supported in $[0,1]\times B(2R+2)$ and satisfies
	$$\|(\Theta_{\mu\nu}^{(2)})_{\cdots}\|_{C^6(\Sigma_t)} + \|\rd_t ((\Theta_{\mu\nu}^{(2)})_{\cdots})\|_{C^5(\Sigma_t)} \leq \ep^2 C(N)$$
	for all $t \in [0,T)$, and the error term $\cdots$ satisfies  \eqref{eq:g.error.1}--\eqref{eq:g.error.4}.

\end{proposition}

\begin{proof}
	The term $	\Theta_{\mu\nu}^{(1)}$ and $	\Theta_{\mu\nu}^{(2)}$ are obtained with simple computation for products involving only the background functions $\phi_0$ and $\varpi_0$, and the parametrices up to second order. This is very similar to Proposition 7.1 in \cite{HL.HF} and we refer the reader there for details. We only note two differences from \cite{HL.HF}:
	\begin{enumerate}
	\item Compared to \cite{HL.HF}, our paramatrix for $\phi$ has an extra term $\lambda^2 F^\phi_{\bA,\bB} \cos (\tfrac{a_\bA u_\bA \pm a_\bB u_\bB}{\lambda })$, but does not have a term  $\lambda^2 F^{3,\phi}_\bA \sin (\tfrac{3a_\bA u_\bA}{\lambda})$. Thus, when computing $\rd_\mu \phi$, there is a contribution of 
	$$\lambda F^\phi_{\bA,\bB} \rd_\mu u_\bA \sin (\tfrac{a_\bA u_\bA \pm a_\bB u_\bB}{\lambda }),$$ 
	with $\bA \neq \bB$. For the order-$\lambda$ contribution, this could be multiplied by a term with no phase (from $\rd \phi_0$), which gives $(	\Theta_{\mu\nu}^{(2)})_{\pm,\bA,\bB}$, or otherwise multiplied by a term with a phase $\sin(\tfrac{a_\bC u_\bC}{\lambda})$ (from $\rd \phi_1$), which gives $(	\Theta_{\mu\nu}^{(2)})_{1,\bA} $, $(\Theta_{\mu\nu}^{(2)})_{\pm,2\bA,\bB}$, or $ (\Theta_{\mu\nu}^{(2)})_{\pm_1,\pm_2,\bA,\bB,\bC}$ terms. Similarly for $\varpi$. 
	\item We also have additional terms coming from the expansion of $u_\bA$, i.e., in $\rd_\mu \phi$, there are additional terms of the form 
	\begin{equation*}
	\begin{split}
	\sum_\bA F^\phi_\bA \rd_\mu u_\bA^2 \sin(\tfrac{a_\bA u_\bA}{\lambda}) = &\: \lambda \sum_\bA  F^\phi_\bA \mathfrak v_\bA \rd_\mu u_\bA \cos(\tfrac{a_\bA u_\bA}{\lambda}) \sin(\tfrac{a_\bA u_\bA}{\lambda}) + \cdots \\
	= &\: \f 12 \lambda \sum_\bA  F^\phi_\bA \mathfrak v_\bA \rd_\mu u_\bA \sin(\tfrac{2a_\bA u_\bA}{\lambda})  + \cdots
	\end{split}
	\end{equation*} 
	Now, in order for this term to contribute to an order-$\lambda$ term, this is either multiplied by something without a phase, which gives the $(\Theta_{\mu\nu}^{(2)})_{2,\bA}$ term, or else it is multiplied by something with a phase $\sin(\tfrac{a_\bB u_\bB}{\lambda})$ (from $\rd \phi_1$), which gives either the $(	\Theta_{\mu\nu}^{(2)})_{3,\bA}$ or the $(	\Theta_{\mu\nu}^{(2)})_{\pm,2\bA,\bB}$ term. Similarly for $\varpi$. 
	\end{enumerate}
	
	Let us study the remainder term. It is easy to check that all these terms obey the bounds \eqref{eq:g.error.1}--\eqref{eq:g.error.2}. For \eqref{eq:g.error.3} and \eqref{eq:g.error.4}, we need to be more careful to not have extra factors of $C(N)$, for otherwise we could not improve the bootstrap assumptions. For this purpose, it is important to make use of the fact that we only need to control an $L^{\f 43}$-based Sobolev norm (as opposed to $L^2$-based). Orthogonality argument will also play an important role. 
	
	We first make two observations about the remainder term: 
	\begin{enumerate}
	\item If the contribution is independent of the tilded quantities $\wht \phi$, $\wht \varpi$ or $\wht u_{\bA}$, then one has a bound $\ep^2 \lambda^2 C(N)$, which is better than what we need since $C(N) \ll C_b(N)$.
	\item If a term is at least quadratic in the tilded quantities $\wht \phi$, $\wht \varpi$ or $\wht u_{\bA}$, then there are extra $\lambda$'s to spare after using the $L^\infty$ estimates in Lemma~\ref{lminfty}. The extra $\lambda$'s can dominate extra factors of $C_b(N) e^{A(N)t}$.
	\end{enumerate}
	
	We now prove \eqref{eq:g.error.3}. Because of the discussion above, it remains to track the terms which are of order $\lambda^2$, and depend on the bootstrapped quantities $\wht \phi$, $\wht \varpi$ or $\wht u_{\bA}$ linearly. There are three sources of such terms:
	\begin{enumerate}
	\item The terms involving $\rd\wht \phi$ or $\rd\wht \varpi$ multiplied by an $O(1)$ contribution from $\rd\phi_0$, $\rd\phi_1$, $\rd \varpi_0$ or $\rd\varpi_1$, i.e., 
	\begin{equation}\label{eq:error.in.g.dphi2.1}
	\| \rd \wht \phi \rd \phi_0 \|_{L^{\f 43}(\Sigma_t)}, \quad \| \rd \wht \phi \rd \phi_1 \|_{L^{\f 43}(\Sigma_t)},\quad \| \rd \wht \varpi \rd \varpi_0 \|_{L^{\f 43}(\Sigma_t)}, \quad \| \rd \wht \varpi \rd \varpi_1 \|_{L^{\f 43}(\Sigma_t)}.
	\end{equation}
	\item Contribution of $\wht \phi$ in $e^{-4\phi} - e^{-4\phi_0}$ multiplied by the $O(1)$ contribution in $(\rd \varpi)^2$. More precisely, we need to control
	\begin{equation}\label{eq:error.in.g.dphi2.2}
	\| e^{-4\phi_0} \widetilde \phi \rd (\varpi_0 + \varpi_1) \rd (\varpi_0 + \varpi_1)\|_{L^{\f 43}(\Sigma_t)}.
	\end{equation}
	\item Terms involving $\wht u_{\bA}$. More precisely, defining
	\begin{equation}
	\rd_\mu \phi_{1,\wht u} := \sum_{\bA} F_\bA^{\phi} \partial_\mu \wht u_\bA \sin(\tfrac{a_\bA u_\bA}{\lambda}),\quad \rd_\mu \varpi_{1,\wht u} := \sum_{\bA} F^{\varpi}_\bA \partial_\mu \wht u_\bA \sin(\tfrac{a_\bA u_\bA}{\lambda}),
	\end{equation}
	we need to control the terms
	\begin{equation}\label{eq:error.in.g.dphi2.3}
	\| \rd \phi_{1,\wht u} \rd \phi_0 \|_{L^{\f 43}(\Sigma_t)}, \quad \| \rd \phi_{1,\wht u} \rd \phi_1 \|_{L^{\f 43}(\Sigma_t)}, \quad\| \rd \varpi_{1,\wht u} \rd \varpi_0 \|_{L^{\f 43}(\Sigma_t)}, \quad \| \rd \varpi_{1,\wht u} \rd \varpi_1 \|_{L^{\f 43}(\Sigma_t)}.
	\end{equation}
	\end{enumerate}
	
	The key to handling these terms will be orthogonality. First, we have
	\begin{equation}\label{eq:orthogonality.again}
	\| \rd \phi_1 \|_{L^4(\Sigma_t)},\, \| \rd \varpi_1 \|_{L^4(\Sigma_t)} \ls \ep,
    	\end{equation}
	using part (2) of Lemma~\ref{lem:orthogonality}, Lemma~\ref{lemma.adapt}, and the background estimates in Proposition~\ref{prop:construct.dust}. We immediately obtain
	\begin{equation}
	\begin{split}
	\hbox{\eqref{eq:error.in.g.dphi2.1}} \ls &\: (\|\rd \phi_0 \|_{L^{4}(\Sigma_t)} + \|\rd \phi_1 \|_{L^{4}(\Sigma_t)} + \|\rd \varpi_0 \|_{L^{4}(\Sigma_t)} + \|\rd \varpi_1 \|_{L^{4}(\Sigma_t)})(\| \rd \wht \phi \|_{L^2(\Sigma_t)} + \| \rd \wht \varpi \|_{L^2(\Sigma_t)}) \\
	\ls &\: \ep^2 \lambda^2 C_b(N) e^{A(N)t}.
	\end{split}
	\end{equation}
	The term \eqref{eq:error.in.g.dphi2.2} can be estimated similarly, using also Sobolev embedding and Lemma~\ref{lem:wave.Poincare}:
	\begin{equation}
	\begin{split}
	\hbox{\eqref{eq:error.in.g.dphi2.2}} \ls &\: (\|\rd \varpi_0 \|_{L^{4}(\Sigma_t)} + \|\rd \varpi_1 \|_{L^{4}(\Sigma_t)}) (\|\rd \varpi_0 \|_{L^{4}(\Sigma_t)} + \|\rd \varpi_1 \|_{L^{4}(\Sigma_t)}) \| \wht \phi \|_{L^4(\Sigma_t)} \\
	\ls &\: \ep^3 \lambda^2 C_b(N) e^{A(N)t}.
	\end{split}
	\end{equation}

	For \eqref{eq:error.in.g.dphi2.3}, we first need an improved estimate for $\rd_\mu \phi_{1,\wht u}$. By part (1) of Lemma~\ref{lem:orthogonality}, we have
	\begin{align*}
		\| \rd_\mu \phi_{1,\wht u} \|_{L^2} 
		\lesssim &\: \Big( \sum_{\bA} \| F_{\bA}^\phi \rd \wht u_\bA \|_{L^2(\Sigma_t)}^2\Big)^{\f 12} + C(N) C_b(N) \ep^2 \lambda^{\f 52} \\
		\ls &\: \Big( \sum_{\bA} \| F_{\bA}^\phi \|_{L^4(\Sigma_t)}^2\Big)^{\f 12} \Big(\sup_{\bA} \| \rd \wht u_\bA \|_{L^4(\Sigma_t)} \Big) + C(N) C_b(N) \ep^2 \lambda^{\f 52} e^{A(N)t}\\
		\ls &\: C_b(N) \ep^{\f 52} \lambda^2 e^{A(N)t},
	\end{align*}
	where we have used the pointwise bound $\| \sum_{\bA} (F_\bA^\phi)^2\|_{L^\i(\Sigma_t)} \ls \ep^2$ (see Proposition~\ref{prop:construct.dust}) and the improved estimates of Propositions \ref{prpu} and \ref{prpdt} to estimate $ \|\partial\wht u_\bA\|_{L^4(K_t)}$ instead of the bootstrap assumptions. Using this together with \eqref{eq:orthogonality.again}, we obtain
	\begin{equation}
	\hbox{First two terms in \eqref{eq:error.in.g.dphi2.3}} \ls \| \rd \phi_{1,\wht u}\|_{L^2(\Sigma_t)} (\|\rd \phi_0 \|_{L^{4}(\Sigma_t)} + \|\rd \phi_1 \|_{L^{4}(\Sigma_t)}) \ls \ep^3 \lambda^2 C_b(N) e^{A(N)t},
	\end{equation} %
which is better than what we need. The last two terms are completely analogous. This concludes the proof of \eqref{eq:g.error.3}.

Finally, for \eqref{eq:g.error.4}, we repeat the same argument after taking a $\rd_t$ derivative. There are two types of contribution, namely when $\rd_t$ hits on the tilded quantity, or when $\rd_t$ hits on the other quantity. In the former case, the argument is the same as \eqref{eq:g.error.3}, where we again rely on the fact that \eqref{eq:orthogonality.again} is $N$-independent. In the latter case, the argument is even easier, as we at worst have terms bounded above by $C(N)C_b(N)\ep^2 \lambda^2e^{A(N)t}$, which is much better than we need since $C(N) C_b(N) \ll C_b'(N)$. \qedhere

\end{proof}

%

\begin{proposition}\label{g1.prop}
	Let $\mfg\in \{N,\gamma,\beta^i\}$. For $\mfg_2$ as in \eqref{eq:mfg2.form}, \eqref{eq:g2.def.1}, \eqref{eq:g2.def.2}, \eqref{eq:g2.def.3} and ${{\bf \Gamma}}_0(\mfg)$ as in \eqref{G.g}, \eqref{G.N} and \eqref{G.b} (with metric components replaced by their background values) and $\Theta^{(1)}$ as in \eqref{g.O1.expand}, we have that $\Delta \mfg_2 - \Gamma(\mfg_0)^{\mu\nu}\Theta^{(1)}_{\mu\nu}$ takes the following form	
	\begin{equation}\label{eq:g1.expansion}
		\begin{split}
			&\Delta \mfg_2- {{\bf \Gamma}}(\mfg_0)^{\mu\nu}\Theta^{(1)}_{\mu\nu}\\
			=&\lambda\sum_{\bA}\left(\mathfrak G_{\bA}^{(\Delta)}(\mfg) \cos(\tfrac{a_\bA u_\bA}{\lambda})+\mathfrak G_{2\bA}^{(\Delta)}(\mfg) \sin(\tfrac{2a_\bA u_\bA}{\lambda}) + \mathfrak G_{3\bA}^{(\Delta)}(\mfg) \cos(\tfrac{3a_\bA u_\bA}{\lambda})\right)\\
			&+\lambda \sum_{\pm, \bA, \bB: \bB \neq \bA} \Big(\mathfrak G_{\pm,\bA,\bB}^{(\Delta)}(\mfg)\sin(\tfrac{ a_\bA u_\bA \pm a_\bB u_{\bB}}{\lambda})+
			\mathfrak G_{\pm,2\bA,\bB}^{(\Delta)}(\mfg)\cos(\tfrac{ 2a_\bA u_\bA \pm a_\bB u_{\gra B}}{\lambda}) \Big)+ \cdots,
		\end{split}
	\end{equation}
	where 
	\begin{itemize}
		\item $\mathfrak G_{\bA}^{(\Delta)}(\mfg)$, $\mathfrak G_{2\bA}^{(\Delta)}(\mfg)$, $\mathfrak G_{3\bA}^{(\Delta)}(\mfg)$, $\mathfrak G_{\pm,\bA,\bB}^{(\Delta)}(\mfg)$ , $\mathfrak G_{\pm,2\bA,\bB}^{(\Delta)}(\mfg)$  are all compactly supported in $[0,1]\times B(0,2R+2)$ and obey the following estimates for all $t \in [0,T)$:
		$$\|\mathfrak G^{(\Delta)}\|_{C^{7}(\Sigma_t)}+\|\rd_t \mathfrak G^{(\Delta)}\|_{C^{6}(\Sigma_t)}\leq C(N)\ep^2,$$
		\item the error term $\cdots$ obeys \eqref{eq:g.error.1}--\eqref{eq:g.error.4}.
	\end{itemize}
	
\end{proposition}

\begin{proof}
	Notice that $\mfg_2$ in \eqref{eq:mfg2.form}, \eqref{eq:g2.def.1}, \eqref{eq:g2.def.2}, \eqref{eq:g2.def.3} so that the mains terms in $\Delta \mfg_2$ are exactly designed to cancel the term ${{\bf \Gamma}}(\mfg_0)^{\mu\nu}\Theta^{(1)}_{\mu\nu}$. There are two types of error terms. First there are terms when one derivative hits on $\mathcal G_{2,\bA}$, etc, and the other one hits on the phase, which gives terms as in Proposition~7.2 in \cite{HL.HF}. These terms take the acceptable form. 
	
	The new error terms comes from the difference between $u_\bA$ and $u_\bA^0$. The main such term is
	\begin{equation}\label{eq:Delta.g1.main}
	\begin{split}
	&\sum_{\bA} \Big( 4 a_\bA^2 \q G_{2,\bA}(\mfg) (|\nabla u_\bA|^2-|\nabla u_\bA^0|^2 )\cos(\tfrac{2a_\bA u_\bA}{\lambda}) +  a_\bA^2 \q G_{1,\bA} (\mfg) (|\nabla u_\bA|^2-|\nabla u_\bA^0|^2 )\sin(\tfrac{a_\bA u_\bA}{\lambda}) \Big) \\
		&	+\sum_{\pm,\bA,\bB: \bB\neq \bA} \q G_{\pm,\bA,\bB}(|\nabla(a_\bA u_\bA \pm a_\bB u_{\gra B})|^2-|\nabla(a_\bA u^0_\bA \pm a_\bB u^0_{\gra B})|^2) \cos(\tfrac{a_\bA u_\bA \pm a_\bB u_{\gra B}}{\lambda}).
	\end{split}
	\end{equation}
	Indeed, all the other terms satisfy the error estimates \eqref{eq:g.error.1}--\eqref{eq:g.error.4}. We in particular make two observations about the order-$\lambda^2$ error terms other than \eqref{eq:Delta.g1.main}:
	\begin{enumerate}
	\item Either the terms only depend on the background (so we only have $C(N)$ in the $N$-dependence), 
	\item or else they depend on up to two spatial derivatives of $\wht u_\bA$, for which we can use Proposition~\ref{prpu} and Proposition~\ref{prpdt} (instead of the bootstrap assumptions) to obtain lower order bounds with acceptable $N$-dependence.
	\end{enumerate}
	
	It thus remains to consider \eqref{eq:Delta.g1.main}. There are contributions from $u_\bA^2$ and $\wht u_\bA$ (and $u_\bB^2$, $\wht u_\bB$). For the contribution from $u_\bA^2$, we recall \eqref{u1.def} and note that in all the order-$\lambda$ terms, we have the following phases
	$$\cos(\tfrac{a_\bA u_\bA}{\lambda})\cos(\tfrac{2a_\bA u_\bA}{\lambda}), \quad \cos(\tfrac{a_\bA u_\bA}{\lambda})\sin(\tfrac{a_\bA u_\bA}{\lambda}),\quad \cos(\tfrac{a_\bA u_\bA}{\lambda})\cos(\tfrac{a_\bA u_\bA \pm a_\bB u_{\gra B}}{\lambda}),$$
	which can in turn be expanded as terms in \eqref{eq:g1.expansion}. For the $\wht u_\bA$ contribution, they obey \eqref{eq:g.error.1}--\eqref{eq:g.error.4} since as above, we can use Proposition~\ref{prpu} and Proposition~\ref{prpdt} (instead of the bootstrap assumptions) to obtain lower order bounds with acceptable $N$-dependence. \qedhere
%
\end{proof}

Finally, we handle the term $\Upsilon(\mfg) - \Upsilon(\mfg_0)$ in \eqref{eq:main.wht.g.prep}. 
\begin{proposition}\label{up1.prop}
Let $\mfg\in \{N,\gamma,\beta^i\}$. $\Upsilon(\mfg)$ admits the following decomposition
	\begin{equation*}
		\begin{split}
			\Upsilon(\mfg) =&\: \Upsilon(\mfg_0) + \lambda\sum_{\bA}\left(\mathfrak G_{\bA}^{(\Upsilon)}(\mfg) \cos(\tfrac{a_{\bA} u_\bA}{\lambda})+\mathfrak  G_{2\bA}^{(\Upsilon)}(\mfg) \sin(\tfrac{2a_{\bA}u_\bA}{\lambda})\right) \\
			&\: + \lambda \sum_{\pm,\bA,\bB: \bB \neq \bA} \mathfrak  G_{\pm,\bA,\bB}^{(\Upsilon)}(\mfg)\sin(\tfrac{ a_\bA u_\bA \pm a_\bB u_{\bB}}{\lambda}) + \cdots,
		\end{split}
	\end{equation*}
	where $\mathfrak G_{\bA}^{(\Upsilon)}(\mfg)$, $\mathfrak  G_{2\bA}^{(\Upsilon)}(\mfg)$ and $\mathfrak  G_{\pm,\bA,\bB}^{(\Upsilon)}(\mfg)$ are all compactly supported in $[0,1]\times B(0,2R+2)$ and obey the estimates
	$$\|\mathfrak G^{(\Upsilon)} \|_{C^{7}(\Sigma_t)}+\|\rd_t \mathfrak G^{(\Upsilon)} \|_{C^{6}(\Sigma_t)}\ls C\ep^2 $$
	for all $t \in [0,T)$, and the error term $\cdots$ obeys \eqref{eq:g.error.1}--\eqref{eq:g.error.4}.
\end{proposition}
\begin{proof}
This is essentially the same as Proposition 7.5 in \cite{HL.HF} and thus we omit the details. The only difference here is that we also need $L^{\f 43}$ estimates for \eqref{eq:g.error.3}--\eqref{eq:g.error.4} as opposed to only $L^2$-based estimates.

Since $\mfg$ obeys better estimates than $\wht \phi$ and $\wht \varpi$, it is much easier to ensure that the error terms have the needed $N$-dependence. We do note that there are some subtleties concerning the weights, specifically that there are some borderline terms, but they are identical as Proposition 7.5 in \cite{HL.HF}. We also remark that it is only to deal with the weights here that we use the bootstrap assumptions \eqref{BA:g.L43} and \eqref{BA:dtg.L43}. \qedhere
\end{proof}

Plugging Proposition~\ref{phi.phi.diff}, Proposition~\ref{g1.prop} and Proposition~\ref{up1.prop} into \eqref{eq:main.wht.g.prep}, we obtain
\begin{proposition}\label{prop:whtg.eqn}
Let $\mfg\in \{N,\gamma,\beta^i\}$.
\begin{equation}
\begin{split}
\Delta \wht \mfg =&\: \lambda\sum_{\bA}\left(\mathfrak G_{\bA}(\mfg) \cos(\tfrac{a_\bA u_\bA}{\lambda})+\mathfrak G_{2\bA}(\mfg) \sin(\tfrac{2a_\bA u_\bA}{\lambda}) + \mathfrak G_{3\bA}(\mfg) \cos(\tfrac{3a_\bA u_\bA}{\lambda})\right)\\
			&\:+\lambda \sum_{\pm, \bA, \bB: \bB \neq \bA} \Big(\mathfrak G_{\pm,\bA,\bB}(\mfg)\sin(\tfrac{ a_\bA u_\bA \pm a_\bB u_{\bB}}{\lambda})+
			\mathfrak G_{\pm,2\bA,\bB}(\mfg)\cos(\tfrac{ 2a_\bA u_\bA \pm a_\bB u_{\gra B}}{\lambda}) \Big)\\
			&\: +\lambda \sum_{\substack{\pm_1,\pm_2,\bA,\bB,\bC: \\
			\bB\neq \bA,\,\bC \neq \bA,\bB}}\mathfrak G_{\pm_1,\pm_2,\bA,\bB,\bC} \cos(\tfrac{a_\bA u_\bA \pm_1 a_\bB u_\bB\pm_1 a_\bC u_\bC }{\lambda})+ \cdots,
		\end{split}
	\end{equation}
	where $\mathfrak G_{\bA}(\mfg)$, $\mathfrak G_{2\bA}(\mfg)$, $\mathfrak G_{3\bA}(\mfg)$, $\mathfrak G_{\pm,\bA,\bB}(\mfg)$ , $\mathfrak G_{\pm,2\bA,\bB}(\mfg)$  are all compactly supported in $[0,1]\times B(0,2R+2)$ and obey the estimates
		$$\|\mathfrak G \|_{C^{7}(\Sigma_t)}+\|\rd_t \mathfrak G \|_{C^{6}(\Sigma_t)}\leq C(N)\ep^2,$$
		for all $t \in [0,T)$, and the error term $\cdots$ obeys \eqref{eq:g.error.1}--\eqref{eq:g.error.4}.
\end{proposition}
\begin{proof}
We start with \eqref{eq:main.wht.g.prep} and compute that 
\begin{equation*}
\begin{split}
&\: \Gamma(\mfg)^{\mu\nu} \langle \rd_\mu  U, \rd_\nu U \rangle - \Gamma(\mfg_0)^{\mu\nu} \langle \rd_\mu  U_0, \rd_\nu U_0 \rangle \\
= &\: (\Gamma(\mfg)^{\mu\nu} - \Gamma(\mfg_0)^{\mu\nu}) \langle \rd_\mu  U_0, \rd_\nu U_0 \rangle + (\Gamma(\mfg)^{\mu\nu} - \Gamma(\mfg_0)^{\mu\nu})(\langle \rd_\mu  U, \rd_\nu U \rangle - \langle \rd_\mu  U_0, \rd_\nu U_0 \rangle) \\
&\: + \Gamma(\mfg_0)^{\mu\nu} (\langle \rd_\mu  U, \rd_\nu U \rangle - \langle \rd_\mu  U_0, \rd_\nu U_0 \rangle).
\end{split}
\end{equation*}
Notice that $(\Gamma(\mfg)^{\mu\nu} - \Gamma(\mfg_0)^{\mu\nu}) \langle \rd_\mu  U_0, \rd_\nu U_0 \rangle$ and $(\Gamma(\mfg)^{\mu\nu} - \Gamma(\mfg_0)^{\mu\nu})(\langle \rd_\mu  U, \rd_\nu U \rangle - \langle \rd_\mu  U_0, \rd_\nu U_0 \rangle)$ obey the bounds \eqref{eq:g.error.1}--\eqref{eq:g.error.4} using \eqref{eq:calG.est}, \eqref{BA:g.L4}--\eqref{BA:dtg.L2}, where for the latter term we also used Proposition~\ref{phi.phi.diff}. Thus, using Proposition~\ref{phi.phi.diff}, Proposition~\ref{g1.prop} and Proposition~\ref{up1.prop}, we obtain the desired conclusion after setting
\begin{align*}
\mathfrak G_{\bA}(\mfg) = &\: \Gamma(\mfg_0)^{\mu\nu}(\Theta^{(2)}_{\mu\nu})_{\bA} + \mathfrak G_{\bA}^{(\Upsilon)}(\mfg) - \mathfrak G_{\bA}^{(\Delta)}(\mfg)\\
\mathfrak G_{2\bA}(\mfg) = &\: \Gamma(\mfg_0)^{\mu\nu}(\Theta^{(2)}_{\mu\nu})_{2\bA}+ \mathfrak G_{2\bA}^{(\Upsilon)}(\mfg) - \mathfrak G_{2\bA}^{(\Delta)}(\mfg)\\
\mathfrak G_{3\bA}(\mfg) = &\: \Gamma(\mfg_0)^{\mu\nu}(\Theta^{(2)}_{\mu\nu})_{3\bA} - \mathfrak G_{3\bA}^{(\Delta)}(\mfg)\\
\mathfrak G_{\pm,\bA,\bB}(\mfg) = &\: \Gamma(\mfg_0)^{\mu\nu}(\Theta^{(2)}_{\mu\nu})_{\pm,\bA,\bB}+ \mathfrak G_{\pm,\bA,\bB}^{(\Upsilon)}(\mfg) - \mathfrak G_{\pm,\bA,\bB}^{(\Delta)}(\mfg)\\
\mathfrak G_{\pm,2\bA,\bB}(\mfg) = &\: \Gamma(\mfg_0)^{\mu\nu}(\Theta^{(2)}_{\mu\nu})_{\pm,2\bA,\bB} - \mathfrak G_{\pm,2\bA,\bB}^{(\Delta)}(\mfg)\\
\mathfrak G_{\pm_1,\pm_2,\bA,\bB,\bC} = &\: \Gamma(\mfg_0)^{\mu\nu}(\Theta^{(2)}_{\mu\nu})_{\pm_1,\pm_2,\bA,\bB,\bC}.
\end{align*}
\end{proof}

We are now ready to solve the elliptic equations. We need two analytic results. The first is a general result on solving elliptic equations in weighted Sobolev spaces (Theorem~\ref{thm:elliptic}), following \cite{McOwen}. The second is a Gagliardo--Nirenberg-type bound for weighted Sobolev spaces. The latter is standard but we could not locate a precise reference and thus give a proof for completeness.

We start with the following result\footnote{Note that \cite{McOwen} technically only gives the $k = 0$ case, but the $k \geq 1$ case can be derived from the $k=0$ case inductively using an argument in \cite[Corollary~2.8]{Huneau.constraint}.} in \cite{McOwen}.
\begin{thm}\label{laplacien}(Theorem 0 in \cite{McOwen})
	Let $m,k\in \mathbb{Z}_{\geq 0}$, $1<p<\infty$ and $-\frac{2}{p}+m<\delta<m+1-\frac{2}{p}$. The Laplace operator $\Delta :W^{k+2,p}_{\delta} \rightarrow W^{k,p}_{\delta+2}$ is an injection with closed range 
	$$\left \{h \in W^{k,p}_{\delta+2}\; | \;\int h q =0 \quad \forall q \in \cup_{i=0}^m \mathcal{H}_i \right \},$$
	where $\mathcal{H}_i$ is the set of harmonic polynomials of degree $i$.
	Moreover, any $r \in W^{k+2,p}_{\delta}(\mathbb R^2)$ obeys the estimate
	$$\|r \|_{W^{k+2,p}_{\delta}(\mathbb R^2)} \leq C(\delta,k,m,p)\|\Delta r\|_{W^{k,p}_{\delta+2}(\mathbb R^2)},$$
	where $C(\delta,k,m,p) > 0$ is a constant depending on $\delta$, $k$, $m$ and $p$.
\end{thm}

The following is a corollary to Theorem~\ref{laplacien}, which can be deduced in a similar manner as Corollary~2.7 in \cite{Huneau.constraint}: 
\begin{theorem}\label{thm:elliptic}
For any $k\in \mathbb Z_{\geq 0}$, $p \in (1,\infty)$, $\de \in (-\f 2p, 1-\f 2p)$, there is a $C = C(\de,k,p) >0$ such that the following holds. 

For every $h \in W^{k}_{-\alp+2}(\mathbb R^2)$, there exists a solution $w$ of 
	$$\Delta w =h$$
	which can be written 
	$$w(x)=\frac{1}{2\pi}\left(\int_{\mathbb R^2} h \right)\zeta(|x|){\log}(|x|) +w_r(x),$$
	where $\zeta$ is the cutoff function in Definition~\ref{def:cutoff}, and $w_r \in W^{k+2,p}_{\de}(\mathbb R^2)$ satisfies
	$\|w_r\|_{W^{k+2,p}_{\de}(\mathbb R^2)} \leq C \|h\|_ {W^{k,p}_{ \de+2}(\mathbb R^2)}$. Moreover, denoting the constant $w_a = \f 1{2\pi} \left(\int_{\mathbb R^2} h \right)$, it also holds that $|w_a| \leq C  \|h\|_ {L^{p}_{ \de+2}(\mathbb R^2)}$.
\end{theorem}
%

\begin{lemma}\label{lem:GN}
The following Gagliardo--Nirenberg-type bound holds\footnote{We remark that this can be further improved. Even using the same ideas of proof, one could replace the first term on the right-hand side by $\| h\|_{H^{k}_{-\alp_0}(\mathbb R^2)}^{\bt} \| h\|_{H^{k+1}_{-\alp_0}(\mathbb R^2)}^{1-\bt}$ for any $\bt \in [\f 13, \f 12)$, but we will not need the stronger statement.} whenever $k \geq 1$:
\begin{equation*}
\begin{split}
 \| h\|_{W^{k,4}_{-\alp_0+\f 12}(\mathbb R^2)} 
\ls &\: \| h\|_{H^{k}_{-\alp_0}(\mathbb R^2)}^{\f 13} \| h\|_{H^{k+1}_{-\alp_0}(\mathbb R^2)}^{\f 23} +  \| h\|_{W^{k-1,4}_{-\alp_0+\f 12}(\mathbb R^2)}.
\end{split}
\end{equation*}
\end{lemma}
\begin{proof}
Since $k \geq 1$, we integrate by parts to obtain
\begin{equation}\label{eq:GN.IBP}
\begin{split}
&\: \|\rd_x^{k} h\|_{L^{4}_{-\alp_0+ \f 12+k}(\mathbb R^2)}^4 \\
\ls &\: \int_{\mathbb R^2} |\rd_x^k h|^2 |\rd_x^{k-1} h| |\rd_x^{k+1} h| \langle x\rangle^{(-\alp_0+\f 12+k)4} \, \ud x + \int_{\mathbb R^2} |\rd_x^k h|^3 |\rd_x^{k-1} h| \langle x\rangle^{(-\alp_0+\f 12 + k)4-1} \, \ud x \\
\ls &\: \|\rd_x^{k} h\|_{L^4_{-\alp_0+ \f 12+k}(\mathbb R^2)}^2 \|\rd_x^{k-1} h\|_{L^{\infty}_{-\alp_0+1+(k-1)}(\mathbb R^2)} \|\rd_x^{k+1} h\|_{L^2_{-\alp_0+(k+1)}(\mathbb R^2)}  \\
&\: + \|\rd_x^{k} h\|_{L^{4}_{-\alp_0+ \f 12+k}(\mathbb R^2)}^3 \|\rd_x^{k-1} h\|_{L^{4}_{-\alp_0+ \f 12+(k-1)}(\mathbb R^2)}
\end{split}
\end{equation}
First using the weighted Sobolev embedding theorem in \cite[Appendix~I Theorem~3.4.1.(a)]{CB.book} and then using H\"older's inequality, we obtain
\begin{equation}\label{eq:GN.Li}
\begin{split}
\|\rd_x^{k-1} h\|_{L^{\infty}_{-\alp_0+1+(k-1)}(\mathbb R^2)} \ls &\: \| \rd_x^{k} h \|_{L^{\f 83}_{-\alp_0+\f 14+k}(\mathbb R^2)} +  \| \rd_x^{k-1} h \|_{L^{\f 83}_{-\alp_0+\f 14+(k-1)}(\mathbb R^2)} \\
\ls &\: \| \rd_x^{k} h \|_{L^{4}_{-\alp_0+\f 12+k}(\mathbb R^2)}^{\f 12}\| \rd_x^{k} h \|_{L^{2}_{-\alp_0+k}(\mathbb R^2)}^{\f 12} + \| h \|_{W^{k-1,4}_{-\alp_0+\f 12}(\mathbb R^2)}^{\f 12} \| h\|_{W^{k-1,2}_{-\alp_0}(\mathbb R^2)}^{\f 12}.
\end{split}
\end{equation}
Therefore, plugging \eqref{eq:GN.Li} into \eqref{eq:GN.IBP} and absorbing suitable factors of $\|\rd_x^{k} h\|_{L^{4}_{-\alp_0+ \f 12+k}(\mathbb R^2)}$ to the left-hand side, we thus obtain
\begin{equation*}
\begin{split}
&\: \| h\|_{W^{k,4}_{-\alp_0+\f 12}(\mathbb R^2)} \ls \|\rd_x^{k} h\|_{L^{4}_{-\alp_0+ \f 12+k}(\mathbb R^2)} + \| h\|_{W^{k-1,4}_{-\alp_0+\f 12}(\mathbb R^2)}\\
\ls &\: \| h\|_{W^{k,2}_{-\alp_0}(\mathbb R^2)}^{\f 13} \| h\|_{W^{k+1,2}_{-\alp_0}(\mathbb R^2)}^{\f 23} +  \| h\|_{W^{k-1,4}_{-\alp_0+\f 12}(\mathbb R^2)},
\end{split}
\end{equation*}
as desired. \qedhere
\end{proof}

We now use the equation in \eqref{prop:whtg.eqn} and the above lemmas to obtain the desired estimates for $\wht \mfg$.
\begin{proposition} \label{prop:g.est}
Let $\mfg\in \{N,\gamma,\beta^i\}$. Then the following estimates hold:
	\begin{align}
|\wht \mfg_{a}|(t)+\|\wht{\mfg}_r\|_{W^{1,\f 43}_{-\alp_0 - \f 12}(\Sigma_t)}\ls &\: C_{b}(N)\ep^2 \lambda^2 e^{A(N)t}, \label{eq:g.L43.recovered}\\
|\rd_t \wht \mfg_{a}|(t)+\|\rd_t\wht{\mfg}_r\|_{W^{1,\f 43}_{-\alp_0 - \f 12}(\Sigma_t)}\ls &\: C'_{b}(N)\ep^2 \lambda e^{A(N)t},\label{eq:dtg.L43.recovered} \\
\sum_{k\leq 3} \lambda^k \|\wht{\mfg}_r\|_{W^{k+1,4}_{-\alp_0 + \f 12}(\Sigma_t)}\ls &\:  C_{b}(N)\ep^2 \lambda^2e^{A(N)t} \label{eq:g.L4.recovered}\\
\sum_{k\leq 2} \lambda^k \|\rd_t\wht{\mfg}_r\|_{W^{k+1,4}_{-\alp_0 + \f 12}(\Sigma_t)}\ls &\: C'_{b}(N)\ep^2 \lambda e^{A(N)t},\label{eq:dtg.L4.recovered}\\
	\sum_{k\leq 3} \lambda^k \|\wht{\mfg}_r\|_{H^{k+2}_{-\alp_0}(\Sigma_t)}\ls &\: C_{b}(N)\ep^2 \lambda e^{A(N)t}, \label{eq:g.L2.recovered}\\
	\sum_{k\leq 2} \lambda^k\|\rd_t\wht{\mfg}_r\|_{H^{k+2}_{-\alp_0}(\Sigma_t)}\ls &\: C'_{b}(N)\ep^2 e^{A(N)t}. \label{eq:dtg.L2.recovered}
\end{align}
\end{proposition}
\begin{proof}
	Define $\wht \mfg^{main}$ by
	\begin{align*}
	\wht \mfg^{main}=&\: - \lambda^3\sum_{\bA}\left(\f{\mathfrak G_{\bA}(\mfg)}{a_\bA^2|\nab u_\bA^0|^2} \cos(\tfrac{a_\bA u_\bA}{\lambda})+ \f{\mathfrak G_{2\bA}(\mfg)}{4a_\bA^2|\nab u_\bA^0|^2} \sin(\tfrac{2a_\bA u_\bA}{\lambda}) + \f{\mathfrak G_{3\bA}(\mfg)}{9 a_\bA^2|\nab u_\bA^0|^2} \cos(\tfrac{3a_\bA u_\bA}{\lambda})\right)\\
	&\: - \lambda^3 \sum_{\pm, \bA, \bB: \bB \neq \bA} \Big(\f{\mathfrak G_{\pm,\bA,\bB}(\mfg)}{|\nab (a_\bA u_\bA^0 \pm a_\bB u_{\bB}^0)|^2}\sin(\tfrac{ a_\bA u_\bA \pm a_\bB u_{\bB}}{\lambda})+
			\f{\mathfrak G_{\pm,2\bA,\bB}(\mfg)}{|\nab (2a_\bA u_\bA^0 \pm a_\bB u_{\gra B}^0)|^2} \cos(\tfrac{ 2a_\bA u_\bA \pm a_\bB u_{\gra B}}{\lambda}) \Big) \\
	&\:- \lambda^3 \sum_{\substack{\pm_1,\pm_2,\bA,\bB,\bC:\\
			\bB\neq \bA,\,\bC \neq \bA,\bB}} \f{\mathfrak G_{\pm_1,\pm_2,\bA,\bB,\bC}}{|\nab (a_\bA u_\bA \pm_1 a_\bB u_\bB\pm_2 a_\bC u_\bC)|^2} \cos(\tfrac{a_\bA u_\bA \pm_1 a_\bB u_\bB\pm_2 a_\bC u_\bC }{\lambda}),
		\end{align*}
		where the $\mathfrak G$ terms are as in Proposition~\ref{prop:whtg.eqn}.
		
		It is straightforward to check that $\wht \mfg^{main}$ is compactly supported and obeys all the desired estimates for $\wht \mfg_r$. Thus we only need to bound $\wht \mfg - \wht \mfg^{main}$. For this we observe that $\Delta \wht \mfg^{main}$ takes away the main terms in Proposition~\ref{prop:whtg.eqn} so that
\begin{align}
	\label{estdelta1}\|\Delta (\wht \mfg	-\wht \mfg^{main} )\|_{W^{0,\frac{4}{3}}_{-\alp_0+\f 32}(\Sigma_t)}&\lesssim \ep^2 \lambda^2 C_b(N)e^{A(N)t},\\
		\label{estdelta2}\|\partial_t \Delta (\wht \mfg	-\wht \mfg^{main}  )\|_{W^{0,\frac{4}{3}}_{-\alp_0+\f 32}(\Sigma_t)}&\lesssim \ep^2 \lambda
			C'_b(N)e^{A(N)t},\\
	\label{estdelta3}\sum_{k\leq 3}\lambda^k\|\Delta (\wht \mfg	-\wht \mfg^{main}  )\|_{H^{k}_{-\alp_0+2}(\Sigma_t)}&\leq \ep^2 \lambda^2 C(N)C_b(N)e^{A(N)t}, \\
	\label{estdelta4}\sum_{k\leq 2}\lambda^k\|\partial_t \Delta (\wht \mfg	-\wht \mfg^{main}  )\|_{H^{k}_{-\alp_0+2}(\Sigma_t)}&\leq
	\ep^2 \lambda C(N)C'_b(N)e^{A(N)t}.  
	\end{align}
	
	By Theorem~\ref{thm:elliptic}, we thus obtain 
	\begin{align}
	\label{estdelta1.inverted}|\wht \mfg_{a}|(t)+\|\wht \mfg_r	-\wht \mfg^{main}\|_{W^{2,\frac{4}{3}}_{-\alp_0-\f 12}(\Sigma_t)}&\lesssim \ep^2 \lambda^2 C_b(N)e^{A(N)t},\\
		\label{estdelta2.inverted}|\rd_t \wht \mfg_{a}|(t)+\|\partial_t (\wht \mfg_r	-\wht \mfg^{main}  )\|_{W^{2,\frac{4}{3}}_{-\alp_0-\f 12}(\Sigma_t)}&\lesssim \ep^2 \lambda
			C'_b(N)e^{A(N)t},\\
	\label{estdelta3.inverted}\sum_{k\leq 3}\lambda^k\|\wht \mfg_r	-\wht \mfg^{main} \|_{H^{k+2}_{-\alp_0}(\Sigma_t)}&\leq \ep^2 \lambda^2 C(N)C_b(N)e^{A(N)t}, \\
	\label{estdelta4.inverted}\sum_{k\leq 2}\lambda^k\|\partial_t  (\wht \mfg_r	-\wht \mfg^{main}) \|_{H^{k+2}_{-\alp_0}(\Sigma_t)}&\leq
	\ep^2 \lambda C(N)C'_b(N)e^{A(N)t}.  
	\end{align}
	
	Therefore, \eqref{estdelta1.inverted} and \eqref{estdelta2.inverted} give the desired estimates \eqref{eq:g.L43.recovered} and \eqref{eq:dtg.L43.recovered}, while \eqref{estdelta3.inverted} and \eqref{estdelta4.inverted} give the desired estimates \eqref{eq:g.L2.recovered} and \eqref{eq:dtg.L2.recovered}. (In the case of the $H^{k+2}_{-\alp_0}$ bounds, there are extra $\lambda$ powers to spare.)
	
	Finally, for the $L^4$ based norms, first note that by Theorem~3.4.1.(b) in Appendix~I of \cite{CB.book}, we have
\begin{equation}
\| h \|_{W^{k,4}_{-\alp_0+\f 12}(\mathbb R^2)} \ls \| h \|_{W^{k+1,\f 43}_{-\alp_0-\f 12}(\mathbb R^2)}.
\end{equation}
Using this, we thus get the $k=0$ case of \eqref{eq:g.L4.recovered} and \eqref{eq:dtg.L4.recovered} from \eqref{estdelta1.inverted} and \eqref{estdelta2.inverted}. For the $k \geq 1$ cases in \eqref{eq:g.L4.recovered} and \eqref{eq:dtg.L4.recovered}, we use Lemma~\ref{lem:GN} and \eqref{estdelta3.inverted}, \eqref{estdelta4.inverted} to obtain inductively that 
\begin{equation*}
\begin{split}
\sum_{k=1}^3 \lambda^k \| \wht \mfg_r - \wht \mfg^{main} \|_{W^{k+1,4}_{-\alp_0+\f 12}(\Sigma_t)} \ls &\:  \ep^2 \lambda^{\f 73} C(N) C_b(N) e^{A(N)t}, \\
\sum_{k=1}^2 \lambda^k \| \rd_t(\wht \mfg_r - \wht \mfg^{main}) \|_{W^{k+1,4}_{-\alp_0+\f 12}(\Sigma_t)} \ls &\: \ep^2 \lambda^{\f 43} C(N) C_b(N) e^{A(N)t},
\end{split}
\end{equation*}
which is better than what we need (with extra $\lambda^{\f 13}$ power to absorb constants $C(N)$). \qedhere

%
%
%
\end{proof}
\section{Estimates for the solution to the Raychaudhuri equation}\label{sec:chi}

In this section, we prove estimates for $\chi_\bA$. 

Recall that $\chi_\bA= \Box_g u_\bA$ and that it satisfies the equation (see \eqref{eq:Raychaudhuri})
\begin{equation}\label{eq:Raychaudhuri.2}
L_\bA(\chi_\bA)+\chi_\bA^2 =-\mathrm{Ric}(L_\bA, L_\bA).
\end{equation}
The Ricci curvature can be computed using \eqref{sys}:
\begin{equation}\label{eq:Ric.LL}
	\f 12 \mathrm{Ric}(L_\bA, L_\bA)=(L_\bA\phi)^2 +\frac{1}{4}  e^{-4\phi} (L_\bA\varpi)^2.
\end{equation}
We will use \eqref{eq:Raychaudhuri.2}--\eqref{eq:Ric.LL} to derive the estimates for $\chi_\bA$.

\begin{lemma}\label{lm:rll}
	The following transport equation holds for $\chi_{\bA} - \chi_{\bA}^0$:
\begin{equation}\label{eq:Ray.diff}
\begin{split}
			&\: L_\bA(\chi_{\bA} - \chi_{\bA}^0) + 2 \chi_\bA^0 (\chi_{\bA} - \chi_{\bA}^0)\\
			= 
&\:			\sum_{\bB:\bB\neq \bA} \Big( \q H^{0,(\bA)}_{1,\bB} \sin(\tfrac{a_\bB u_\bB}{\lambda})+\q H^{0,(\bA)}_{2,\bB} \cos(\tfrac{2a_\bB u_\bB}{\lambda})\Big) +\sum_{\substack{\pm,\bB,\bC:\\\bB\neq \bA \\ \bC \neq \bB, \bA}} \q H^{0,(\bA)}_{\pm,\bB,\bC}\cos(\tfrac{a_\bB u_\bB\pm a_\bC u_\bC}{\lambda})\\
		&\:	+\lambda \sum_{\bB} \Big( \q H^{1,(\bA)}_{1,\bB} \cos(\tfrac{a_\bB u_\bB}{\lambda})+\q H^{1,(\bA)}_{2,\bB} \sin(\tfrac{2a_\bB u_\bB}{\lambda}) \Big)+\lambda\sum_{\bB: \bB\neq \bA}\q H^{1,(\bA)}_{3,\bB} \cos(\tfrac{3a_\bB u_\bB}{\lambda}) \\
			&\:	+	\lambda \sum_{\pm,\bB,\bC: \bB\neq \bC} \Big(\q H^{1,(\bA)}_{\pm,2\bB,\bC}\cos(\tfrac{2a_\bB u_\bB\pm a_\bC u_\bC}{\lambda})+\q H^{1,(\bA)}_{\pm,\bB,\bC}\sin(\tfrac{a_\bB u_\bB\pm a_\bC u_\bC}{\lambda}) \Big)\\
			&\:+	\lambda \sum_{\substack{\pm_1,\pm_2,\bB,\bC,\bD:\\
			\bC\neq \bB\\
			\bD \neq \bB,\bC}} \q H^{1,(\bA)}_{\pm_1,\pm_2,\bB,\bC,\bD}\cos(\tfrac{a_\bB u_\bB\pm_1 a_\bC u_\bC\pm_2 a_\bD u_\bD}{\lambda})+\cdots
\end{split}
\end{equation}
where $\q H$ are background quantities supported in $[0,1]\times B(0,2R+2)$ satisfying
\begin{equation}\label{eq:H.est.bad}
\|\q H\|_{C^7(K_t)} \leq C(N)\ep^2
\end{equation}
and 
\begin{equation}\label{eq:chi.error}
\sum_{k\leq 3}\lambda^k \| \cdots\|_{H^k(K_t)} \leq \ep^2 \lambda^2C(N)C_b(N) e^{A(N)t}.
\end{equation}
Moreover, $\q H^{0,(\bA)}_{1,\bB}$, $\q H^{0,(\bA)}_{2,\bB}$, $\q H^{0,(\bA)}_{\pm,\bB,\bC}$, $\q H^{1,(\bA)}_{1,\bA}$, $\q H^{1,(\bA)}_{1,\bA}$ depend only on $F^\phi_{\cdot}$, $F^\varpi_{\cdot}$, $\phi_0$, $\varpi_0$, $g_0$ and $u^0_{\cdot}$, and obey the improved estimate (as compared to \eqref{eq:H.est.bad})
\begin{equation}
\label{eq:H.est.good}
\|\q H\|_{C^8(K_t)}  \leq C(N)\ep^2.
\end{equation}
\end{lemma}

\begin{proof}
We will say that a term is an ``acceptable high-frequency term'' if it can be written in the form of one of the highly oscillatory terms on the right-hand side of \eqref{eq:Ray.diff}, and that it is an ``acceptable error'' if it can be written as one of the $\cdots$ terms.

\pfstep{Step~1: Computing the Ricci curvature} We first show that 
\begin{align*}
			\mathrm{Ric}(L_\bA, L_\bA)= &\:2(L^0_\bA \phi_0)^2 +\frac{1}{2}  e^{4\phi_0} (L^0_\bA \varpi_0)^2 
			+\sum_{\bB: \bB \neq \bA}(L^0_\bA(u^0_\bB))^2 \Big( (F^\phi_\bB)^2+\frac{1}{4}  e^{4\phi_0} (F^\varpi_\bB)^2 \Big)\\
&\:			+\hbox{acceptable high-frequency terms} + \hbox{acceptable error}.
			\end{align*}
We start with the $(L_\bA \phi)^2$ term in \eqref{eq:Ric.LL}. We compute
\begin{equation}\label{eq:Ric.LL.Lphi}
\begin{split}
		L_\bA(\phi)= &(g_0^{-1})^{\alpha\beta}\partial_\alpha (u_\bA^0 + u_\bA^2)\partial_\beta\phi_0
			-\sum_{\bB:\bB \neq \bA}  F^\phi_\bB  (g_0^{-1})^{\alpha\beta}(\partial_\alpha u_\bA^0\partial_\beta u_\bB^0 + \partial_\alpha u_\bA^0\partial_\beta u_\bB^2 + \partial_\alpha u_\bA^2\partial_\beta u_\bB^0)\sin(\tfrac{a_\bB u_\bB}{\lambda}) 
		\\
		&\: -\lambda\sum_{\bB } a_\bB^{-1}(g_0^{-1})^{\alpha\beta}\partial_\alpha u_\bA^0\partial_\beta F^\phi_\bB \cos(\tfrac{a_\bB u_\bB}{\lambda}) \\
		&+	\lambda \sum_{\bB: \bB\neq \bA}(g_0^{-1})^{\alpha\beta}\partial_\alpha u_\bA^0\partial_\beta u_\bB^0 \Big(F^{1,\phi}_\bB \cos(\tfrac{a_\bB u_\bB}{\lambda})-2 F_\bB^{2,\phi} \sin(\tfrac{2a_\bB u_\bB}{\lambda})\Big) \\
	&+ \lambda \sum_{\pm,\bB,\bC:\bB\neq \bC}   (\pm 1) F^\phi_{\bB\bC}(g_0^{-1})^{\alpha\beta}\partial_\alpha u_\bA^0\partial_\beta(a_\bB u_\bB^0 \pm a_\bC u_\bC^0) \sin(\tfrac{a_\bB u_\bB \pm a_\bC u_\bC}{\lambda }) +\cdots,
	\end{split}
	\end{equation}
where we used the eikonal equation for $u_\bA$. We note that
\begin{subequations}
\begin{empheq}{align}
(g_0^{-1})^{\alpha\beta}\partial_\alpha u_\bA^2\partial_\beta\phi_0
= &\:  \lambda \mathfrak v_\bA (g_0^{-1})^{\alpha\beta}\partial_\alpha u_\bA^0 \partial_\beta\phi_0\cos(\tfrac{a_\bA u_\bA}{\lambda})+\cdots, \label{eq:Ric.LL.u2.1}\\
(g_0^{-1})^{\alpha\beta}\partial_\alpha u_\bA^0\partial_\beta u_\bB^2 = &\: \lambda \mathfrak v_\bB (g_0^{-1})^{\alpha\beta}\partial_\alpha u_\bA^0 \partial_\beta u_\bB^0 \cos(\tfrac{a_\bB u_\bB}{\lambda})+\cdots,\\
(g_0^{-1})^{\alpha\beta}\partial_\alpha u_\bA^2\partial_\beta u_\bB^0	 = &\: \lambda \mathfrak v_\bA (g_0^{-1})^{\alpha\beta}\partial_\alpha u_\bB^0 \partial_\beta  u_\bA^0 \cos(\tfrac{a_\bA u_\bA}{\lambda})+ \cdots.\label{eq:Ric.LL.u2.3}
		\end{empheq}
		\end{subequations}
Plugging \eqref{eq:Ric.LL.u2.1}--\eqref{eq:Ric.LL.u2.3} into \eqref{eq:Ric.LL.Lphi} and squaring, we see that the contribution from $(L_\bA\phi)^2$ takes the acceptable form.

The contribution from $\f 14 e^{4\phi} (L_\bA \varpi)^2$ can be treated in a similar manner, except that we also expand
$$e^{-4\phi} = e^{-4\phi_0} - 4 \lambda e^{-4\phi_0} \sum_{\bA} a_{\bA}^{-1} F_\bA^\phi \cos(\tfrac{a_\bA u_\bA}\lambda)+ \cdots.$$
There are terms coming from multiplying $4 \lambda e^{-4\phi_0} \sum_{\bA} a_{\bA}^{-1} F_\bA^\phi \cos(\tfrac{a_\bA u_\bA}\lambda)$ with the $O(\lambda^0)$ part of $(L_\bA \varpi)^2$, but it is again straightforward to check that they take the acceptable form. 

\pfstep{Step~2: Completing the argument} First note that the Raychaudhuri equation of the background gives
$$L_\bA^0 \chi_\bA^0 + (\chi_\bA^0)^2 = - 2(L^0_\bA \phi_0)^2 - \frac{1}{2}  e^{4\phi_0} (L^0_\bA \varpi_0)^2 
			- \sum_{\bB: \bB \neq \bA}(L^0_\bA(u^0_\bB))^2 \Big( (F^\phi_\bB)^2+\frac{1}{4}  e^{4\phi_0} (F^\varpi_\bB)^2 \Big),$$
where we used the eikonal equation for the background to eliminate the $\bB = \bA$ term in the last sum. We then write
\begin{align}
&\: L_\bA (\chi_\bA - \chi^0_{\bA}) + 2\chi_\bA^0 (\chi_\bA - \chi^0_{\bA}) \notag \\
=&\:  L_\bA \chi_\bA - L_{\bA}^0 \chi^0_{\bA} + (\chi_\bA)^2 - (\chi_\bA^0)^2 \notag \\
&\: - (\chi_\bA - \chi^0_{\bA})^2 - (g^{-1} - g_0^{-1})^{\alp\bt} \rd_\alp u_\bA \rd_\bt \chi^0_{\bA} - (g_0^{-1})^{\alp\bt} \rd_\alp (u_\bA-u_\bA^0) \rd_\bt \chi^0_{\bA} \notag \\
= &\:  - \mathrm{Ric}(L_\bA,L_\bA) + 2(L^0_\bA \phi_0)^2 + \frac{1}{2}  e^{4\phi_0} (L^0_\bA \varpi_0)^2 
			+ \sum_{\bB: \bB \neq \bA}(L^0_\bA(u^0_\bB))^2 \Big( (F^\phi_\bB)^2+\frac{1}{4}  e^{4\phi_0} (F^\varpi_\bB)^2 \Big) \label{eq:L.chi.diff.final.1}\\
			&\: - (\chi_\bA - \chi^0_{\bA})^2 - (g^{-1} - g_0^{-1})^{\alp\bt} \rd_\alp u_\bA \rd_\bt \chi^0_{\bA} - (g_0^{-1})^{\alp\bt} \rd_\alp (u_\bA-u_\bA^0) \rd_\bt \chi^0_{\bA}, \label{eq:L.chi.diff.final.2}
\end{align}
where we have also used \eqref{eq:Raychaudhuri}. We finally observe that \eqref{eq:L.chi.diff.final.1} is of the acceptable form thanks to Step~1, while \eqref{eq:L.chi.diff.final.2} is of the acceptable form thanks to the parametrices for $u_\bA$, $g$ and $\chi_\bA$. \qedhere
\end{proof}

In Lemma~\ref{lm:rll}, the following terms cannot be directly estimated as they give an $O(\lambda)$ (instead of $O(\lambda^2)$) contribution.
\begin{itemize}
\item all the $O(\lambda^0)$ terms, and
\item the parallel $O(\lambda^1)$ terms, i.e., those with phases $\cos(\tfrac{a_\bB u_\bB}\lambda)$ and $\sin(\tfrac{2 a_\bB u_{\bB}}\lambda)$ when $\bB = \bA$.
\end{itemize}
This is captured by the $\chi_\bA^1$ terms in the parametrix (see \eqref{eq:chi.para} and \eqref{chi1.def}). The following lemma summarizes how $\chi_\bA^1$ removes the main terms.
\begin{lemma}\label{lem:Ray.diff.tilde}
Define $\wht \chi_\bA$ by \eqref{eq:chi.para}, \eqref{chi1.def}, \eqref{eq:chi1.def.1.1}--\eqref{eq:chi1.def.1.3} and \eqref{eq:chi1.def.2.1}--\eqref{eq:chi1.def.2.2}. Then the following transport equation holds for $\wht \chi_{\bA}$:
\begin{equation*}
\begin{split}
			&\: L_\bA \wht \chi_{\bA} + 2 \chi_\bA^0  \wht \chi_{\bA} \\
			= 
		&\:	\lambda \sum_{\bB:\bB\neq \bA} \Big( \wht{\q H}^{1,(\bA)}_{1,\bB} \cos(\tfrac{a_\bB u_\bB}{\lambda})+\wht{\q H}^{1,(\bA)}_{2,\bB} \sin(\tfrac{2a_\bB u_\bB}{\lambda}) + \wht{\q H}^{1,(\bA)}_{3,\bB} \cos(\tfrac{3a_\bB u_\bB}{\lambda}) \Big)\\
			&\:	+	\lambda \sum_{\pm,\bB,\bC: \bB\neq \bC} \Big(\wht{\q H}^{1,(\bA)}_{\pm,2\bB,\bC}\cos(\tfrac{2a_\bB u_\bB\pm a_\bC u_\bC}{\lambda})+\wht{\q H}^{1,(\bA)}_{\pm,\bB,\bC}\sin(\tfrac{a_\bB u_\bB\pm a_\bC u_\bC}{\lambda} \Big) \\
			&\:+	\lambda \sum_{\substack{\pm_1,\pm_2,\bB,\bC,\bD:\\
			\bC\neq \bB\\
			\bD \neq \bB,\bC}} \wht{\q H}^{1,(\bA)}_{\pm_1,\pm_2,\bB,\bC,\bD}\cos(\tfrac{a_\bB u_\bB\pm_1 a_\bC u_\bC\pm_2 a_\bD u_\bD}{\lambda})+\cdots,
\end{split}
\end{equation*}
where $\wht{\q H}$ are background quantities supported in $[0,1]\times B(0,2R+2)$, 
$\|\wht{\q H}\|_{C^{7}(K_t)}  \leq C(N) \ep^2$ for $t \in [0,T)$,
and the error term $\cdots$ satisfies \eqref{eq:chi.error}. 
\end{lemma}
\begin{proof}

%
%
%
%
%

We recall \eqref{chi1.def} and compute
\begin{align*}
		L_\bA(\chi_\bA^1)=&\:\sum_{\bB: \bB \neq \bA}2a_\bB \mathfrak X^{(\bA)}_{2,\bB} g^{\alpha\beta}\partial_\alpha u_\bA \partial_\beta u_\bB\cos ( \tfrac{2a_\bB u_\bB}{\lambda})\\
		&\:-\sum_{\bB:\bB \neq \bA} a_\bB \mathfrak X^{(\bA)}_{1,\bB} g^{\alpha\beta}\partial_\alpha  u_\bA \partial_\beta u_\bB \sin( \tfrac{a_\bB u_\bB}{\lambda})\\
		&\:+\sum_{\substack{\pm,\bB,\bC: \\ \bB \neq \bB, \bC \neq \bA,\bB}} \mathfrak X^{(\bA)}_{\pm,\bB,\bC} g^{\alpha\beta}\partial_\alpha u_\bA \partial_\beta(a_\bB u_\bB\pm a_\bC u_\bC)\cos( \tfrac{a_\bB u_\bB\pm a_\bC u_\bC}{\lambda}) \\
	&\:+  \lambda \sum_{\bB } L_\bA(\mathfrak X^{(\bA)}_{2,\bB}) \sin ( \tfrac{2a_\bB u_\bB}{\lambda})
	+\lambda \sum_{\bB }L_\bA( \mathfrak X^{(\bA)}_{1,\bB} )\cos( \tfrac{a_\bB u_\bB}{\lambda})\\
&\:	+\lambda \sum_{\substack{\pm,\bB,\bC: \\ \bB \neq \bB, \bC \neq \bA,\bB}} L_\bA(\mathfrak X^{(\bA)}_{\pm,\bB,\bC} )\sin ( \tfrac{a_\bB u_\bB\pm a_\bC u_\bC}{\lambda}), 
	\end{align*}
	where we have the eikonal equation for $u_\bA$ to get $\bB \neq \bA$ in some terms above. 
	
	First observe that $ \mathfrak X^{(\bA)}_{2,\bB}$, $\mathfrak X^{(\bA)}_{1,\bB}$ (in the case $\bB \neq \bA$) and $\mathfrak X^{(\bA)}_{\pm,\bB,\bC}$ terms (see \eqref{eq:chi1.def.1.1}--\eqref{eq:chi1.def.1.3}) are chosen so as to remove the $O(1)$ terms with phases  $\cos(\tfrac{2a_\bB u_\bB}{\lambda})$, $\sin(\tfrac{a_\bB u_\bB}{\lambda})$ (with $\bB \neq \bA$) and $\cos(\tfrac{a_\bB u_\bB\pm a_\bC u_\bC}{\lambda})$ respectively. Thus, for terms with those phases we gain a power of $\lambda$.
	
	Similarly, $\mathfrak X^{(\bA)}_{1,\bA}$ and $\mathfrak X^{(\bA)}_{2,\bA}$ (see \eqref{eq:chi1.def.2.1}--\eqref{eq:chi1.def.2.2}) are chosen to remove the $O(\lambda)$ terms with phases $\cos(\tfrac{a_\bA u_\bA}{\lambda})$, $\sin(\tfrac{2a_\bA u_\bA}{\lambda})$. Thus, for terms with those phases, instead of a sum over all $\bB$, we now have a sum over all $\bB$ such that $\bB \neq \bA$. \qedhere

	\end{proof}
	
\begin{prp}\label{prop:chi}
	The following estimates hold for $\wht \chi_\bA$ for $t \in [0,T)$:
$$\sum_{k \leq 3}
\lambda^k\|\wht \chi_\bA\|_{H^{k}(K_t)}\leq  \lambda^2C_b(N) \ep e^{A(N)t}.$$
	\end{prp}
\begin{proof}
	Notice that the main terms on the right-hand side of the equation in Lemma~\ref{lem:Ray.diff.tilde} are $O(\lambda)$ and none of them oscillates in the $u_\bA$ direction. Hence $L_\bA^0$ of the phases are bounded away from $0$ since all the $u_{\bB}^0$'s are null adapted (see Lemma~\ref{lemma.adapt}). Thus, we can define $\wht \chi_\bA^{main}$ as follows to remove all the terms in Lemma~\ref{lem:Ray.diff.tilde}:
%
	\begin{equation*}
	\begin{split}
	&\: \lambda^{-2} \wht \chi_\bA^{main} \\
	= 
&\:			\sum_{\bB:\bB\neq \bA} ((g_0^{\alpha \beta})^{-1}\partial_\alpha u_\bA^0 \partial_\beta u_\bB^0)^{-1} \Big( \wht{\q H}^{1,(\bA)}_{1,\bB} \sin(\tfrac{a_\bB u_\bB}{\lambda})-\f 12\wht{\q H}^{1,(\bA)}_{2,\bB} \cos(\tfrac{2a_\bB u_\bB}{\lambda}) +\f 13 \wht{\q H}^{1,(\bA)}_{3,\bB} \sin(\tfrac{3a_\bB u_\bB}{\lambda}) \Big)\\
		&\: -  \sum_{\substack{\pm,\bB,\bC:\bB \neq \bC}} ((g_0^{\alpha \beta})^{-1}\partial_\alpha u_\bA^0 \partial_\beta (a_\bB u^0_\bB\pm a_\bC u^0_\bC))^{-1} \Big(  \wht{\q H}^{1,(\bA)}_{\pm,\bB,\bC}\cos(\tfrac{a_\bB u_\bB\pm a_\bC u_\bC}{\lambda}) \Big) \\
			&\:	+ \sum_{\pm,\bB,\bC: \bB\neq \bC} ((g_0^{\alpha \beta})^{-1}\partial_\alpha u_\bA^0 \partial_\beta (2a_\bB u^0_\bB\pm a_\bC u^0_\bC))^{-1} \wht{\q H}^{1,(\bA)}_{\pm,2\bB,\bC}\sin(\tfrac{2a_\bB u_\bB\pm a_\bC u_\bC}{\lambda}) \\
			&\:+	\sum_{\substack{\pm_1,\pm_2,\bB,\bC,\bD:\\
			\bC\neq \bB\\
			\bD \neq \bB,\bC}} ((g_0^{\alpha \beta})^{-1}\partial_\alpha u_\bA^0 \partial_\beta (a_\bB u^0_\bB\pm a_\bC u^0_\bC\pm a_\bD u^0_\bD))^{-1}  \wht{\q H}^{1,(\bA)}_{\pm_1,\pm_2,\bB,\bC,\bD}\sin(\tfrac{a_\bB u_\bB\pm_1 a_\bC u_\bC\pm_2 a_\bD u_\bD}{\lambda}).
\end{split}
\end{equation*}

By the estimates in Lemma~\ref{lem:Ray.diff.tilde}, $\| \wht \chi_\bA^{main} \|_{H^8(K_t)} \ls C(N) \ep^2 \lambda^2$. Moreover, $\wht \chi_\bA^{main}$ is chosen so that by Lemma~\ref{lem:Ray.diff.tilde} we have
\begin{equation}
\sum_{k\leq 3} \lambda^k \| L_\bA (\wht \chi_\bA - \wht \chi_\bA^{main}) + 2 \chi_\bA^0 (\wht \chi_\bA -\wht \chi_\bA^{main} )\|_{H^k(K_t)} \leq \ep^2 \lambda^2C(N)C_b(N) e^{A(N)t}.
\end{equation}

	Applying Lemma~\ref{lmtransport}, we obtain the desired bound for $\wht \chi_\bA - \wht \chi_\bA^{main}$ (where we have used the hierarchy of constants \eqref{eq;constant.hierarchy} and that 
	$$\int_0^t C(N) C_b(N) e^{A(N)s} \, \ud s \leq \f{C(N) C_b(N)}{A(N)} e^{A(N)t}.$$ Combining this with the bound for $\wht \chi_\bA^{main}$, we deduce the desired conclusion. \qedhere
	\end{proof}

\section{Putting everything together}\label{sec:together}
	\begin{proof}[Proof of Theorem~\ref{thm:dust.by.vacuum}]
	We now conclude the proof of Theorem~\ref{thm:dust.by.vacuum} by showing that all the bootstrap assumptions in Section~\ref{sec:bootstrap.assumptions} are improved. A routine continuity argument implies that the solution exists up to time $t = 1$. Since the bounds in the bootstrap assumptions are stronger than those stated in Theorem~\ref{thm:dust.by.vacuum} (after choosing $B(N) = 2C_b'(N)e^{A(N)}$), this finishes the argument.
	
	To check that all the bootstrap assumptions are improved, note that \eqref{BA:wave} is improved in Proposition~\ref{prop:wave.est.final}, \eqref{BA:g.L43}--\eqref{BA:dtg.L2} are improved in Proposition~\ref{prop:g.est}, \eqref{BA:u.L4}--\eqref{BA:u.L2} are improved in Proposition~\ref{prpu} and Proposition~\ref{prpdt}, and finally, \eqref{BA:chi} is improved in Proposition~\ref{prop:chi}. \qedhere
	
	\end{proof}

\bibliographystyle{DLplain}
\bibliography{HFlimit}

\end{document}